%% file: aligncomp.tex
\newif\ifarxiv\arxivtrue
\ifarxiv % remove line numbers which are disallowed by ArXiv
\documentclass[acmsmall,screen]{acmart}%\settopmatter{printfolios=true,printccs=false,printacmref=false}
\else
\documentclass[acmsmall,review]{acmart}
\fi

\acmJournal{TOPLAS}
\acmVolume{1} % ALERT these numbers must be non-empty
\acmNumber{1}
\acmArticle{1}
\acmYear{} % ALERT hack
\acmDOI{} % \acmDOI{10.114/nnnnnnn.nnnnnnn}
\startPage{1}
\setcopyright{none}

\usepackage{latexsym}
\usepackage{hyperref}
\usepackage{color} 
\usepackage{listings} 
\usepackage{mathpartir} 
\usepackage{stmaryrd} 
\usepackage{relsize}
\usepackage{tikz}
\usetikzlibrary{automata,positioning,arrows}

\usepackage{microtype} 
\usepackage{arydshln} % dashed lines in table 

\input{macros}
\renewcommand{\mod}{\mathbin{\mathsf{mod}}}
\newcommand{\KE}{\mathit{KE}} % KAT expression
\newcommand{\Hyp}{\mathconst{Hyp}} % 

\begin{document}

\newtheorem{remark}[theorem]{Remark} % not sure why but this needs to be here not earlier

\ifarxiv
\pagestyle{plain} % to get page numbers and eliminate header - must remove in camera ready
\fi

%\title{A relational Hoare logic with rewriting and its alignment completeness}
%\title{Alignment completeness of a relational Hoare logic with rewriting. {\small Working draft \today}}

\title{Making Relational Hoare Logic Alignment Complete}

\author{Anindya Banerjee}
\affiliation{\institution{IMDEA Software Institute}\country{Spain}}
\author{Ramana Nagasamudram}
\author{David A. Naumann}
\affiliation{\institution{Stevens Institute of Technology}\country{USA}}
%\email{webmaster@marysville-ohio.com}
%\affiliation{%
%  \institution{Institute for Clarity in Documentation}
%  \streetaddress{P.O. Box 1212}
%  \country{USA}
\authornote{The second and third authors were partially supported by 
NSF grant CNS 1718713.}

\begin{abstract} % Kent Beck guidelines 
% Context/Importance
% Problem
% Startling sentence
% Implications
In relational verification, judicious alignment of computational steps facilitates proof of relations between programs using simple relational assertions. Relational Hoare logics (RHL) provide compositional rules that embody various alignments of executions.  Seemingly more flexible alignments can be expressed in terms of product automata based on program transition relations. A single degenerate alignment rule (self-composition), atop a complete Hoare logic, comprises a RHL that is complete in the ordinary logical sense. The notion of alignment completeness was previously proposed as a more satisfactory measure, and some rules were shown to be alignment complete with respect to a few ad hoc forms of alignment automata. This paper proves alignment completeness with respect to a general class of alignment automata, for a RHL comprised of standard rules together with  a rule of semantics-preserving rewrites based on Kleene algebra with tests. Besides solving the open problem of general alignment completeness, this result bridges between human-friendly syntax-based reasoning and automata representations that facilitate automated verification.
\end{abstract}
%Keywords: logics of programs; relational properties; 2-safety; semantic completeness; relational Hoare logic; inductive assertion method

\maketitle

\emph{\color{red}NOTE: the results in this unpublished document have been incorporated into a more recent document with substantial additional results.
See ``Alignment complete relational Hoare logics for some and all'',
ArXiv 2307.10045.}

\section{Introduction}\label{sec:intro}

A ubiquitous problem in programming is reasoning about relational properties, such as equivalence between two programs in the sense that from any given input they produce the same output.  Relational properties are also of interest for a single program. For example, a basic notion in security is noninterference: any two executions with the same public inputs should result in the same public outputs, even if the secret inputs differ.  

For unary, pre-post properties, Turing~\cite{TuringChecking49}
reasoned using assertions at intermediate points in a program, expressing facts that hold at intermediate steps of computation.
Floyd~\cite{Floyd67} formalized the inductive assertion method as a sound and complete technique for proving correctness of imperative programs.  Hoare~\cite{Hoare69} formalized the method as a logic of correctness judgments about programs, which has come to be called a logic of programs.
Hoare's logic not only formalizes the inductive assertion method in terms of the textual representation of programs (vs.\ control flow graphs/automata in the earlier works), it also
formalizes reasoning that manipulates specifications and combines proofs
(e.g., the rules of consequence and disjunction).

Relational properties involve the alignment of pairs of states in two executions.  At the least,
the initial and final states are aligned: the specification comprises a pre-relation,
meant to constrain the initial pair of states, and a post-relation meant to constrain the pair of final states.  For effective reasoning one usually aligns intermediate points as well.  Alignment of executions is often expressed by alignment of program text ---think of a programmer viewing two versions of a similar program, side by side on screen.
Formally, alignment can be expressed Floyd-style, by a product of control-flow graphs or automata, and it can also be expressed syntactically.
To give an example of the latter, write $\{\R\} c\sep d \{\S\}$ for the judgment that any pair of terminated executions of commands $c,d$ from states related by $\R$ end in states related by $\S$.  If $c$ and $d$ are both sequences, say $c=c_0;c_1$ and $d=d_0;d_1$, 
we may wish to align the intermediate points, so we can decompose the reasoning to judgments $\{\R\} c_0\sep d_0 \{\Q\}$ and $\{\Q\} c_1\sep d_1 \{\S\}$
from which $\{\R\} c_0;c_1\sep d_0;d_1 \{\S\}$ follows by general rule.

A number of researchers have introduced relational Hoare logics with such rules~\cite{Francez83,Benton:popl04,Yang07relsep} and by now there are many works on relational verification using either automata representations or program logics.  
The fundamental way to evaluate a logic is in terms of soundness and completeness ---but
there is a problem with completeness.  Consider this rule:
from $\{\R\} c\sep \skipc \{\Q\}$ and $\{\Q\} \skipc\sep d \{\S\}$
infer $\{\R\} c\sep d \{\S\}$. It embodies a degenerate form of alignment where 
the left program's final states are aligned with initial states of program on the right.
Such an alignment is useful: If $c$ and $d$ act on disjoint variables, the unary program $c;d$ can be used as a ``product program'' representing their paired executions and reducing the problem 
to unary reasoning~\cite{BartheDArgenioRezk}. 
Put differently, the judgments $\{\R\} c\sep \skipc \{\Q\}$ and $\{\Q\} \skipc\sep d \{\S\}$
 can as well be construed as unary judgments
$\{\R\} c \{\Q\}$ and $\{\Q\} d \{\S\}$ and proved using (unary) Hoare logic.
But sequential alignment has high cost in the complexity of assertions needed.  The intermediate $\Q$ may need to give a full functional description of the left 
computation ---even if $c$ and $d$ are identical and with lockstep alignment it suffices to use very simple relations.  
The problem is that the sequential alignment rule, together with unary rules, is complete on its own, in the usual sense: true judgments can be proved
(relative to reasoning about assertions~\cite{Cook78}).  

To solve the problem, the notion of \emph{alignment completeness}~\cite{NagasamudramN21} 
relates provability in program logic with provability using alignment product automata.  
In~\cite{NagasamudramN21}, several rules of relational Hoare logic are justified by completeness with respect to particular classes of alignment automata, such as lockstep products.  
Alignment using automata has been found effective in practice.  Automata can express data-dependent alignments in addition to those which are described purely in term of control structure.  Here is one formulation: an alignment product is a transition system based on the transition systems of the programs being related, together with conditions $L$, $R$, and $J$ on state-pairs, that designate whether the product should take a step on the left only, or the right, or a joint step on both sides.  
\emph{The main contribution of this paper is a relational Hoare logic and proof that it is alignment complete with respect to general alignment products.}

For data-dependent alignment, the crucial inference rule is Beringer's conditionally aligned loop rule~\cite{Beringer11},
which uses assertions to express conditions under which one considers left-only (or right-only or joint) iteration of the loop bodies.  But with product automata, one may have alignment of arbitrary subprogram steps.  Such considerations led researchers to propose various additional inference rules such as one that aligns a sequence on the left with a loop on the right~\cite{Francez83}, but this path seems fraught.  
A key insight is to use pure control-structure rewrites, such as loop unrolling, to manipulate the programs under consideration to be more similar syntactically, 
thereby broadening the applicability of same-structure rules~\cite{BartheGHS17,BNN16}.

%% Hoare's logic is for proving correctness judgments, but it relies on another judgment: entailment between assertions.  This is a problem because for sufficiently expressive assertion languages there is no complete deductive system.  Cook solved the problem by defining relative completeness, roughly meaning: a true correctness judgment can be proved using the rules of Hoare logic together with an oracle for all true entailments.  
Our relational Hoare logic relies on entailment of (relational) assertions but it also relies on another judgment: command equivalence.  
Equivalences such as loop unrolling reflect the meaning of control structures,
independent of the meaning of conditional expressions and assignments.
Kleene algebra with tests~\cite{Kozen97} (KAT) is a robust formalism that treats such equivalences
propositionally (in the sense of propositional logic).
Our logic uses KAT as means to prove command equivalence, rather than 
relying on custom rules or an oracle for semantic equivalence.

\subsection{Outline and contributions}

Sect.~\ref{sec:related} touches on additional related work.
Sect.~\ref{sec:logics} spells out the standard imperative language and semantics with which we work.  To expose the main ideas we use partial correctness semantics of imperative programs over integer variables.  We use Dijkstra's guarded command language~\cite{AptOld3} which 
facilitates succinct expression of a normal form central to our approach.   
This section gives the rules of unary and relational Hoare logic including the crucial rewrite rules that allow commands to be replaced by equivalent ones.
These rules are used to derive some one-sided rules (i.e., relate with skip).

Sect.~\ref{sec:aut} spells out the representation of programs by automata,
i.e., transition systems with explicit control points, and the inductive assertion method which is based on annotating control points.  It also formalizes alignment automata and shows they are sound and semantically complete for proving relational judgments.  A key ingredient for soundness is the notion of \emph{adequacy} for an alignment automaton, which means that the automaton covers all pairs of executions of the programs to be aligned.
Adequacy is easy to achieve, taking the alignment conditions $L,R,J$ to all be true, but stronger conditions are needed for alignment to facilitate reasoning with simple relational assertions. 

Sect.~\ref{sec:KATnf} is devoted to a key technical result which is our 
\underline{first contribution:} using KAT, any command can be rewritten to an equivalent command in \emph{automaton normal form}, a form that mimics a fetch-execute instruction loop using an explicit program counter variable.  

Sect.~\ref{sec:autUnary} uses automaton normal form for unary reasoning, 
making our \underline{second contribution:} a Floyd completeness theorem
which says any proof using the inductive assertion method can be 
represented in Hoare logic using essentially the same assertions.  
This has conceptual interest because it connects different reasoning methods
(and we speculate that it may have value as basis for end-to-end foundational verification with tools combining the methods). 
A similar result was already proved, by a different technique, in~\cite{NagasamudramN21}.
In both cases the proofs bring to light techniques which are then generalized to the
relational case.

Sect.~\ref{sec:acomplete} is the main result and 
\underline{third contribution:} the proof that our relational Hoare logic, dubbed RHL+, is alignment complete.  
That is, any proof of a relational judgment using an alignment automaton 
can be represented in RHL+ using essentially the same assertions
and the same left/right/joint alignment conditions.  
This relies on automaton normal form for the programs being related,
in a way similar to what is done in Sect.~\ref{sec:autUnary}.
Adequacy of the alignment automaton corresponds in an elegant way
to the side condition of the conditionally aligned loop rule.
In unary Hoare logic, the disjunction rule enables reasoning by cases but it is not necessary for completeness.  The relational disjunction rule plays an important role in our completeness result.  

Sect.~\ref{sec:cook} revisits Cook completeness in light of our results.
Completeness of Hoare logic is a direct consequence of 
the theorem of Sect.~\ref{sec:autUnary} together with completeness of the inductive assertion method.  
Completeness of RHL+ is similarly a consequence of the theorem of 
Sect.~\ref{sec:acomplete} together with completeness of adequate alignment automata.
By contrast with prior Cook completeness results for relational Hoare logics,
which rely on the sequential product rule and a complete unary Hoare logic,
RHL+ makes no use of unary judgments.

Sect.~\ref{sec:discuss} concludes with remarks on how the results can be generalized to richer programming languages.  

Sect.~\ref{sec:examples} comprises three worked examples illustrating the 
ideas used in Sects.~\ref{sec:autUnary} and~\ref{sec:acomplete}.
The first example is a unary verification, done by rewriting the program to automaton normal form as in the general proof in Sect.~\ref{sec:autUnary}.
The second example is a relational verification using lockstep alignment,
and the third is a relational verification using alignment conditioned on data.
These examples are meant to introduce the reader to the key ingredients of the later proofs, and as motivation for some technical details.

\subsection{Related work}\label{sec:related}

Aside from Nagasamudram and Naumann~\cite{NagasamudramN21} we are not aware of any
other alignment completeness results.   
Their alignment completeness results are all for special forms 
of alignment product that are contrived to closely reflect certain proof rules.
%They also introduce a notion of Floyd completeness which is similar to our result in Sect.~\ref{sec:autUnary} (where we discuss the differences in detail).
In passing they describe general alignment automata.

Cook completeness has been proved for several relational Hoare logics (e.g.,~\cite{SousaD2016,BartheGHS17}),
using in each case the completeness of a single proof rule (sequential product or self composition) that reduces a relational judgment to a unary one, together with a complete Hoare logic.
As we observe in Sect.~\ref{sec:discuss}, we obtain Cook completeness as a corollary of alignment completeness and without recourse to a unary logic.

We are not aware of unary Hoare logics that feature a rewriting rule,
except those in connection with relational logic;
but verification tools often use correctness-preserving rewriting and translation.
The relational Hoare logics of Barthe et al~\cite{BartheGHS17} and Banerjee et al~\cite{BNN16} feature a rewriting rule.
In both of these works, a custom set of rules is provided for command equivalence, 
such as loop unrolling.  The equivalence judgment of~\cite{BartheGHS17} includes a relational precondition whereas that of~\cite{BNN16} is unconditional.  
Both works use rewriting to broaden the applicability of a small set of syntax-directed rules, for specific verifications and in the case of~\cite{BartheGHS17} also to derive additional rules 
(along the lines \rn{AlgnIf} and \rn{IfSkip} in Sect.~\ref{sec:derived}).  
Rewriting by command equivalence is combined with relational reasoning 
in an extension of KAT called BiKAT~\cite{AntonopoulosEtal2023}.

The rule of conditionally aligned loops appears first in Beringer~\cite{Beringer11};
variations appear in Barthe et al~\cite{BartheGHS17},
in the full version of Banerjee et al~\cite{BNN16},
and in~\cite{BNNN19}.
Several published logics only support lockstep alignment of loops,
or lockstep until one terminates.

What we call relational property is a kind of 2-safety~\cite{TerauchiA2005}.
Cartesian Hoare logic~\cite{SousaD2016,PickFG18} reasons about $k$-safety, for $k$ that is fixed throughout a proof.
D’Osualdo et al~\cite{DosualdoFD2022} develop a logic for $k$-safety that features rules for combining judgments with varying $k$;
they have a completeness result based essentially on the degenerate sequential alignment
as explained in~\cite[appendix C]{DosualdoFD2022}.
The system of D'Osualdo et al includes a rewriting rule based on semantic refinement 
(i.e., refinement is not formalized in a deductive system but rather must be proved
in the metalogic).  

With regard to termination, unary Hoare logics are for total or partial correctness.  
The $\forall\forall$ relational judgment considered in this paper is like partial correctness.
Benton's~\cite{Benton:popl04}  work introducing the term relational Hoare logic uses co-termination: either both terminate in related states or neither terminates.
Another form is relative termination~\cite{hawblitzelklr13} which says if the left run terminates then so does the right.  This form has the advantage that relational judgments compose transitively.  
Rinard and Marinov~\cite{RinardMarinov99,RinardCredible99} give a logical formulation of verification conditions for a $\forall\exists$ relation, to prove correctness of compiler transformations acting on control flow graphs.
Antonopoulos et al~\cite{AntonopoulosEtal2023} call this form forward simulation,
and provide deductive rules for both forward simulation and backward simulation
in a KAT-based algebraic framework.

Kozen~\cite{Kozen97} uses KAT to prove that every command is simulated by one with a single loop,
effectively using additional primitives to snapshot the values of branch conditions.
Our normal form is like that of Hoare and Sampaio~\cite{HoareHS93} which uses an
explicit program counter variable.
They prove every program can be reduced to normal form using a set of refinement laws, demonstrating a kind of completeness of the laws.
They use data refinement to justify the introduction of the program counter, which is akin to our use of the ghost elimination rule to enable using the normal form to prove correctness of the original program.

Results like our adequacy Proposition~\ref{prop:adequate-sound}
have been proved for several notions of alignment product that are similar to ours.
(The term fairness is sometimes used for what we call adequacy.)
Our $L,R,J$ corresponds to the ``alignment predicate'' in Churchill et al~\cite{ChurchillP0A19},
the ``composition function'' in Shemer et al~\cite{ShemerGSV19},
and ``scheduler'' in Unno et al~\cite{UnnoTerauchiKoskinen21}.
These works focus on automated search for good alignments and annotations 
using solvers for restricted assertion languages.

The survey of Hoare logic by Apt and Olderog~\cite{AptO19} explains its origins in
the work of Turing~\cite{TuringChecking49} and Floyd~\cite{Floyd67}.
Beckert and Ulbrich survey relational verification~\cite{BeckertU18}
and some notions of relational Hoare logic are sketched in~\cite{NaumannISOLA20}.
Another survey of relational logic can be found in the related work section of Banerjee et al~\cite{BNNN19}.
Kov{\'a}cs et al~\cite{KovacsSF13} introduce the term alignment
and make very clear that alignment is ultimately about executions not syntax.
Almost all works on relational verification involve alignment in some form;
an exception is the Trace logic of Barthe et al~\cite{BartheEGGKM2019} which 
allows to express relations between various intermediate steps without
distinguishing any particular alignment.

To streamline the technical development and focus on what's interesting
we use shallow embedding of assertions as is popular in work on program logics
(e.g.,~\cite{NipkowCSL02,OHearn2019,Pierce:SF2}).
This sidesteps the need to formalize expressiveness of the assertion language or 
to formulate completeness as being relative to provability of entailments~\cite{Cook78,AptOld3}.

%% Churchill et al - alignment predicate plays role like our L,R,J; they refer to it as a weak invariant but do not separately prove that it's an inductive invariant.  (they use testing to search for good alignment predicates, and then verification using it in conjunction with the relational assertions that correspond to our annotation and VCs. 
%% TODO: they mention their method can't express alignments that depend on unbounded input (e.g., a data-dependent inner loop on one side corresponding to single step on other side.  
%% Maybe generize ours from w mod 2 to w mod v for some input value v (assume v,v' positive).

%% Shemer et al: 'composition function' is akin to alignment automaton.
%% They use (small-step) transition system representation, formulated as guarded Horn clauses.
%% Their composition function maps a pair (k-tuple in their setting) to which executions should step next; like our L,R,J.
%% They infer for the composition funct, under constraint to ensure that it is exhaustive (all tuples covered by one of the cases) and exclusive (only one subset of executions is scheduled at each point; we do that by LO,RO,JO and control points)

\section{The programming language and its relational Hoare logic}\label{sec:logics}

\subsection{The language and its (unary) Hoare logic}\label{sec:unary}

Floyd's formulation of the inductive assertion method is based on annotation of control points.  In our development, we use program syntax with labelled control points,
and automata using explicit control points.
The labelled guarded command syntax (\dt{GCL}) is defined as follows, where $x$ ranges over integer variables, $e$ ranges over integer and boolean expressions, $n$ ranges over integer literals.
\[ \begin{array}[t]{lcl}
c & ::= & \lskipc{n} \:\mid\: \lassg{n}{x}{e} \:\mid\: c;c 
       \:\mid\: \lifgc{n}{gcs} \:\mid\: \ldogc{n}{gcs} \\
gcs & ::= & e \gcto c \:\mid\: e \gcto c \gcsep gcs 
\end{array} \]
We omit details about expressions except to note that boolean expressions are given from some primitives $bprim$ and the logical operators are written $\land,\lor,\neg$.
A \dt{guarded command} has the form $e \gcto c$ where $e$ is a boolean expression.
The category $gcs$ is essentially non-empty lists of guarded commands and we sometimes treat it as such.

This somewhat silly command $c0$ will be a running example.  
Here $\mod$ means remainder.
%\[ c0:  \lassg{1}{x}{1} ;
%        \ldogc{2}{ x > 0 \gcto \lifgc{3}{ x\mod 2 = 0 \gcto \lassg{4}{x}{x-1}
%                                 \gcsep x\mod 2 \neq 0 \gcto \lassg{5}{x}{x-2} } } \]
\[ \graybox{c0:} \quad
        \lassg{1}{x}{y} ;
        \keyw{do}^2 x > 0 \gcto 
                            \begin{array}[t]{l}
                            \keyw{if}^3 \: x\mod 2 = 0 \gcto \lassg{4}{x}{x-1} \\
                            \gcsep \: x\mod 2 \neq 0 \gcto \lassg{5}{x}{x-2} 
                             ~ \keyw{fi} ~ \keyw{od}
                             \end{array} 
\]

A \dt{variable store} is a mapping from program variables to their values.
Later we consider automata, with arbitrary sets of stores that need not be variable stores.
But in this section we say simply ``store'', meaning variable store.

We assume expressions $e$ are always defined.
We write $\means{e}(s)$ for the value of expression $e$ in store $s$.  
In case $e$ is a boolean expression, the value is in $\{true,\mathit{false}\}$; otherwise 
it is in $\Z$.
For boolean expression $e$, we sometimes treat $\means{e}$ as a set of states 
as opposed to the characteristic function thereof.

To define partial correctness, we use a standard big-step semantics.
We write $\means{c}\, s\, t$ to express that from initial store $s$ the command $c$ can terminate with final store $t$.  
(Omitted definitions can be found in the appendix and in the Coq development.)
We call $\means{c}$ the \dt{denotation} of $c$.
%[This is ceval in the Coq development.]

Later we define a predicate, $\ok$, on commands that says their labels are unique and positive.  Labels play an important role in some results, for which the $\ok$ condition is needed. But many definitions and results do not involve labels or require $\ok$; for those definitions and results we omit labels, meaning that any labels are allowed.

For any $gcs$, define $\enab(gcs)$ as the disjunction of the guards.
For example, $\enab(x>0\gcto y:=1 \gcsep z < 1\gcto y:=2)$ is $x>0 \lor z<1$.

The standard semantics for guarded commands considers $\lifgc{}{gcs}$ to fail if none of its guards is enabled~\cite{AptOld3}.
Modeling failure would clutter the semantics without shedding any light on our main result.
So we disallow such ifs, as follows.

\begin{definition}[programming language]\label{def:lang}
A \dt{well formed command} is a command $c$ according to the grammar, such that
\begin{itemize}
\item $c$ is typable in the sense that guards are boolean expressions and both integer and boolean operators are used sensibly, considering that all variables have type int.
\item $c$ satisfies the \dt{$\totalIf$} condition which says:
for every subprogram $\lifgc{}{gcs}$ of $c$, 
$\enab(gcs)$ is semantically equivalent to true.  That is, $\means{\enab(gcs)}=\means{true}$
\end{itemize}
In the sequel we say \dt{command} to mean well formed command.
We refer to the language of well formed commands as the \dt{guarded command language}, (\dt{GCL}).
\end{definition}

We choose this semantic formulation in order to highlight the minimal requirements needed 
for the alignment completeness theorem.  In a logic without a rewrite rule, the $\totalIf$ condition can simply be enforced by side condition in the rule for if-commands:
$precond\imp\enab(gcs)$ (see~\cite[rule 32]{AptOld3}).
This is insufficient in our logics, because if-commands can be rewritten to other forms.

\begin{remark}\label{rem:totIfSyn}
\upshape
Here is a syntactic way to ensure the semantic condition $\totalIf$
on $\lifgc{}{gcs}$.
Require that for each $e\gcto c$ in $gcs$,
there is $e^-\gcto d$ in $gcs$ for some $d$, where $e^-$ 
is the syntactic complement of $e$.
That is, $(\neg e)^- \eqdef e$ and if $e$ does not have outermost negation
then $e^- \eqdef \neg e$.
This subsumes the usual if-else:\footnote{Note 
  that a single if-else rules out commands like
  $\lifgc{}{x\geq 0 \gcto y:=0 \gcsep x \leq 0 \gcto y:=1}$
  which satisfy $\totalIf$ but are not deterministic.
  But we can write it as 
  $\lifgc{}{ x\geq 0 \gcto y:=0 \gcsep x \ngeq 0 \gcto y:=1 \gcsep
             x > 0 \gcto y:=0 \gcsep x \not > 0 \gcto y:=1 }$.}
$\lifgc{}{e\gcto c \gcsep \neg e\gcto d}$.
\qed\end{remark}

\paragraph{Program logic} 

The partial correctness judgment is written $c:\spec{P}{Q}$ rather than the conventional $\{P\}c\{Q\}$.
We treat predicates as shallowly embedded in the metalanguage,
so $P$ and $Q$ are sets of stores.
But we use formula notation for clarity, e.g., $P\land Q$ means their intersection.
A boolean expression $e$ in a formula stands for $\{s\mid\means{e}(s)=true\}$,
(which we sometimes write as $\means{e}$).
We often write  $P\imp Q$ to mean valid implication, i.e., $P\subseteq Q$.
For a set $P$ of variable stores, we use substitution notation $\subst{P}{x}{e}$
with standard meaning: $s\in\subst{P}{x}{e}$ iff $\update{s}{x}{\means{e}(s)}\in P$.
Finally, existential quantification is defined by
$\some{x}{P} \eqdef \{ s \mid \some{v}{\update{s}{x}{v}\in P}\}$.
To express that a formula is independent of a variable,
we define \graybox{$\indep(x,P)$} to mean that $P=\some{x}{P}$.

We use $\models$ to indicate \dt{valid} correctness judgment, defined as follows:
\begin{equation}\label{eq:valid}
\graybox{$\models c: \spec{P}{Q}$}
\quad\eqdef\quad
\mbox{For all $s,t$ such that 
$\means{c}\, s\, t$,
if $s\in P$ then $t\in Q$.} 
\end{equation}

\begin{figure}
\begin{mathpar}

\mprset{sep=1.3em}

\inferrule[Rewrite]{
    c: \spec{P}{Q} \\ c\kateq d
}{
    d: \spec{P}{Q} }

\inferrule[Ghost]{
    c: \spec{P}{Q} \\ \ghost(x,c) \\ \indep(x,P)\\ \indep(x,Q) 
}{
    \erase(x,c) : \spec{P}{Q}
}

\inferrule[Do]{
   c: \spec{e\land P}{P} \mbox{ for every $e\gcto c$ in $gcs$}
}{
    \dogc{gcs}: \spec{P}{P\land \neg \enab(gcs) }
}

\inferrule[Asgn]{}{ \lassg{}{x}{e} : \spec{\subst{P}{x}{e}}{P} }

\inferrule[Skip]{}{ 
\skipc:\spec{P}{P} 
}

\inferrule[Seq]
{
c:\spec{P}{R} \\ d:\spec{R}{Q}
}{
c;d : \spec{P}{Q} 
}

\inferrule[If]
{
c: \spec{e\land P}{Q} \mbox{ for every $e\gcto c$ in $gcs$}
}{
%    \ifgc{gcs}: \spec{P\land\enab(gcs)}{Q}
    \ifgc{gcs}: \spec{P}{Q}
}

\inferrule[Conseq]
{
P\imp R \\ 
c : \spec{R}{S} \\
S\imp Q
}{
c : \spec{P}{Q} 
}

\inferrule[False]{}{ 
c:\spec{\mathit{false}}{P} 
}

\end{mathpar}
\caption{%Selected 
Proof rules of HL+
%(see appendix Fig.~\ref{fig:HL}  for the others)
}\label{fig:HLplus}
\end{figure}
Figure \ref{fig:HLplus} gives rules of the proof system we call \dt{HL+}.
It comprises the standard rules of Hoare logic (denoted \dt{HL}) plus a rule for elimination of ghost variables~\cite{FilliatreGP16} and a rule for rewriting programs.  

In a system where assertions are given by formulas, rule \rn{Ghost}
would require that variable $x$ is not in the free variables of $P$ or $Q$.
For our shallow embedding using sets of stores this is expressed by $\indep(x,P)$.
We define $\ghost(x,c)$  to mean that variable $x$ occurs in $c$ only in assignments to $x$.  
Define $\erase(x,c)$ to be $c$ with every assignment to $x$ replaced by skip.

Because we are using shallow embedding for assertions, HL+ does not need to be accompanied by a formal system for proving entailments between assertions.
But the logic also uses command equality $c\kateq d$, in rule \rn{Rewrite}.  
For this we rely on Kleene algebra with tests (\dt{KAT}) as formalized 
in Sect.~\ref{sec:KATnf}.
For soundness of rule \rn{rewrite} we just need that the relation $\kateq$ implies equal denotations, i.e., $c\kateq d$ implies $\means{c} = \means{d}$.
The formal statement of soundness for HL+ is Prop.~\ref{prop:HLsound} in Sect.~\ref{sec:KATequiv} where we define $\kateq$.

%\dn{``rewrite rules'' may be confusing, sounds like a rewriting system;
%but ``equivalence rule'' may be confusing w.r.t. equivalence as a relational judgment.  What to do?}

\subsection{Relational specs and proof rules}

For relational pre- and post-conditions we use (binary) relations on stores,
by shallow embedding just like for unary predicates.  
To express that a unary predicate holds in the left store of a pair,
we write \graybox{$\leftF{P}$} for the set of pairs $(s,t)$ where $s\in P$.
Combined with our treatment of boolean expressions $e$ as predicates,
this means $\leftF{e}$ says that $e$ is true in the left store.
Similarly, $\rightF{P}$ is the set of $(s,t)$ where $t\in P$.
Notation: 
\graybox{$\bothF{P\sep Q}$} abbreviates  $\leftF{P}\land\rightF{Q}$,
and \graybox{$\bothF{P}$} abbreviates  $\bothF{P\sep P}$.
We also write $\leftex{e}=\rightex{e'}$ for the set of $(s,t)$ where
$\means{e}(s) = \means{e'}(t)$.

For relation $\R$ on stores we write 
$\subst{\R}{x|}{e|}$ for substitution of $e$ for $x$ in the left state.
Similarly $\subst{\R}{|x}{|e}$ substitutes on the right
and $\subst{\R}{x|x'}{e|e'}$ does both.
The definitions are a straightforward generalization of unary substitution,
e.g., $(s,t)\in \subst{\R}{x|x'}{e|e'}$ 
iff $( \update{s}{x}{v}, \update{t}{x'}{e'})\in \R$.
We write $\some{x|x'}{\R}$ for quantification over $x$ on the left side and $x'$ on the right, again generalizing the unary version.
We need this form to define \graybox{$\indep(x|x',\R)$} which says $\R$ is independent from $x$ on the left and $x'$ on the right:  $\indep(x|x',\R)$ iff  $\R = \some{x|x'}{\R}$.

For store relations $\R,\S$,
a relational spec is written $\rspec{\R}{\S}$.
Relational correctness for $c$ and $c'$ is written $c\sep c': \rspec{\R}{\S}$
(as opposed to $\{\R\} c\sep c' \{\S\}$ that we use in Sect.~\ref{sec:intro}).
We use $\models$ to indicate valid judgment,
and define \graybox{$\models c\sep c': \rspec{\R}{\S}$} $\eqdef$
\[ 
\mbox{For all $s,s',t,t'$ such that 
$\means{c}\, s\, t$ and $\means{c'}\, s'\, t'$,
if $(s,s')\in \R$ then $(t,t')\in \S$.}
\]

\begin{remark}\upshape
Guards in guarded commands need not be mutually exclusive, so commands may be nondeterministic.  However, with interesting relations as postconditions,
the above $\forall\forall$ notion of correctness is only validated by programs 
that are sufficiently deterministic.  Indeed, a spec like $\rspec{\leftex{in}=\rightex{in}}{\leftex{out}=\rightex{out}}$ expresses determinacy.
\qed\end{remark}

\begin{figure*}[t]
\begin{small}
\begin{mathpar}
%\mprset{sep=1.4em}

\inferrule[rRewrite]{ c\sep c': \rspec{\R}{\S} \\ c\kateq d \\  c'\kateq d' }
{ d\sep d': \rspec{\R}{\S}  }

\inferrule[rGhost]{
c\sep c': \rspec{\R}{\S} \\ \ghost(x,c) \\ \ghost(x',c') \\ 
%      x\notin FV(\R,\S) \\ x'\notin FV(\R,\S) }
     \indep(x|x',\R) \\ \indep(x|x',\S) }
{   \erase(x,c)\sep \erase(x',c') : \rspec{\R}{\S}  }

\inferrule[dIf]{
  c\sep c' : \rspec{\R\land \leftF{e}\land\rightF{e'}}{\S} 
\quad\mbox{for all $e\gcto c$ in $gcs$ and $e'\gcto c'$ in $gcs'$}
}{
  \ifgc{gcs} \Sep \ifgc{gcs'} : 
%     \rspec{\R \land \leftF{\enab(gcs)}\land\rightF{\enab(gcs')}}{\S}
     \rspec{\R}{\S}
}

\inferrule[dAsgn]{}{
x:=e \sep x':=e' : \rspec{\subst{\R}{x|x'}{e|e'}}{\R}
}

\inferrule[AsgnSkip]{}{
x:=e \sep \skipc  : \rspec{\subst{\R}{x|}{e|}}{\R}
}

\inferrule[SkipAsgn]{}{
\skipc\sep x:=e : \rspec{\subst{\R}{|x}{|e}}{\R}
}

\inferrule[dSkip]{}{
\skipc \sep \skipc : \rspec{\R}{\R}
}

\inferrule[dSeq]{
  c\sep c' : \rspec{\R}{\S} \\
  d\sep d' : \rspec{\S}{\T}
}{
  c;d \Sep c';d' : \rspec{\R}{\T}
}

%% \inferrule[dDo]{
%%   c\sep \skipc : \rspec{\Q\land \leftF{e}\land\Lrel }{\Q}
%% \quad\mbox{for all $e\gcto c$ in $gcs$} 
%% \\
%%   \skipc\sep c' : \rspec{\Q\land \rightF{e'}\land\R }{\Q}
%% \quad\mbox{for all $e'\gcto c'$ in $gcs'$} 
%% \\
%%   c\sep c' : \rspec{\Q\land \leftF{e}\land\rightF{e'} \land \neg\Lrel \land \neg\R}{\Q}
%% \quad\mbox{for all $e\gcto c$ in $gcs$ and $e'\gcto c'$ in $gcs'$} \\
%% \Q\imp (\leftex{\enab(gcs)} = \rightex{\enab(gcs')}  
%%           \lor (\Lrel \land \leftF{\enab(gcs)})
%%           \lor (\R \land \rightF{\enab(gcs')}))
%% }{ 
%%   \dogc{gcs} \Sep \dogc{gcs'} : \rspec{\Q}{\Q\land \neg\leftF{\enab(gcs)}\land\neg\rightF{\enab(gcs')}}
%% }

% reformatted, same as commented-out rule 
\inferrule[dDo]{
  {\begin{array}{l}
  c\sep \skipc : \rspec{\Q\land \leftF{e}\land\Lrel }{\Q}
\quad\mbox{for all $e\gcto c$ in $gcs$} 
\\
  \skipc\sep c' : \rspec{\Q\land \rightF{e'}\land\R }{\Q}
\quad\mbox{for all $e'\gcto c'$ in $gcs'$} 
\\
  c\sep c' : \rspec{\Q\land \leftF{e}\land\rightF{e'} \land \neg\Lrel \land \neg\R}{\Q}
\quad\mbox{for all $e\gcto c$ in $gcs$ and $e'\gcto c'$ in $gcs'$} 
\\
\Q\imp (\leftex{\enab(gcs)} = \rightex{\enab(gcs')}  
          \lor (\Lrel \land \leftF{\enab(gcs)})
          \lor (\R \land \rightF{\enab(gcs')}))
  \end{array}}
}{ 
  \dogc{gcs} \Sep \dogc{gcs'} : \rspec{\Q}{\Q\land \neg\leftF{\enab(gcs)}\land\neg\rightF{\enab(gcs')}}
}

\inferrule[rConseq]{
  \P\imp \R \\
  c\sep d : \rspec{\R}{\S} \\
  \S \imp \Q 
}{
  c\sep d : \rspec{\P}{\Q} \\
}

\inferrule[rDisj]{
  c\sep d : \rspec{\Q}{\S} \\
  c\sep d : \rspec{\R}{\S} \\
}{
  c\sep d : \rspec{\Q\lor\R}{\S} \\
}

\inferrule[rFalse]{}{
  c\sep d : \rspec{\mathit{false}}{\R} 
}

\end{mathpar}
\end{small}
\vspace*{-1ex}
\caption{The rules of RHL+}
\label{fig:RHL}
\end{figure*}

Fig.~\ref{fig:RHL} gives the proof rules for relational judgments, including ``diagonal'' rules 
like \rn{dIf} that relate two commands of the same kind,
and one-sided rules \rn{AsgnSkip} and \rn{SkipAsgn}.  
We are not aware of a relational Hoare logic that has been formulated for guarded commands,
but the rules are straightforward adaptations of rules found in prior work,
and they have been proved sound (see Prop.~\ref{prop:RHLsound}).
Rule \rn{rGhost} is a straightforward adaptation of the unary rule which is well known.

\subsection{Derived rules}\label{sec:derived}

Our selection of rules is not minimal for alignment completeness---\rn{rIf} can be omitted!
However, we are not arguing that in practice one should make maximal use of \rn{rRewrite} as
we will do to prove alignment completeness.  In practice, a deductive proof should follow the original program structure to the extent possible, and also make use of derived rules like this one sometimes presented as a basic rule in relational Hoare logics (e.g.,~\cite{Benton:popl04,Yang07relsep}):
\[\inferrule[AlgnIf]{
  c \sep c' : \rspec{\R\land \leftF{e}}{\S} \\
  d \sep d' : \rspec{\R\land\leftF{\neg e}}{\S} \\ 
  \R\imp \leftF{e}=\rightF{e'} 
}{
  \lifgc{}{e\gcto c \gcsep \neg e\gcto d} \sep 
  \lifgc{}{e'\gcto c' \gcsep \neg e'\gcto d'} :\rspec{\R}{\S} 
}\]
It is easily derived using \rn{dIf} and \rn{rFalse}.  So is the similar rule for lockstep alignment of loops.

The one-sided rule \rn{SeqSkip} is derivable using \rn{dSeq} and \rn{rRewrite} with the equivalence $\skipc;\skipc\kateq\skipc$.
\[\inferrule[SeqSkip]{
  c \sep\skipc : \rspec{\R}{\Q} \\
  d \sep \skipc : \rspec{\Q}{\S} 
}{
  c;d \sep \skipc : \rspec{\R}{\S}
}\]
By repeated use of \rn{rDisj} and \rn{rConseq} one can derive
\[
\inferrule[rDisjN]{
  c\sep d : \rspec{\Q_i}{\S} \mbox{ for all $i\in X$} \\
 \mbox{$X$ is a finite set}
}{
  c\sep d : \rspec{\quant{\lor}{i}{i\in X}{\Q_i}}{\R} \\
}
\]

For our alignment completeness result we do not need other one-sided rules.
However, one can derive general rules \rn{IfSkip} and \rn{DoSkip} (and right-side rules as well),
obtaining a left-side rule for each construct in the language, and 
\emph{mutatis mutandis} for the right.  
Here is the left-side rule for if:
\[\inferrule[IfSkip]{
  c\sep\skipc : \rspec{\R\land \leftF{e}}{\S} 
\quad\mbox{for all $e\gcto c$ in $gcs$}
}{
  \lifgc{}{gcs} \Sep \skipc : \rspec{\R}{\S}
}
\]
To derive it, first use \rn{rIf} to derive
$ \lifgc{}{gcs} \Sep 
  \lifgc{}{true \gcto \skipc} : \rspec{\R}{\S}
$, 
Now apply \rn{rRewrite} using the equivalence
 $\lifgc{}{true \gcto \skipc} \kateq \skipc$ (and reflexivity of $\kateq$).

Here is a one-sided loop rule:
\[\inferrule{ % *[left=IfSkip]{
  c\sep \skipc : \rspec{\Q\land \leftF{e}}{\Q} \quad \mbox{for all $e\gcto c$ in $gcs$}
}{
  \ldogc{}{gcs} \sep \skipc : \rspec{\Q}{\Q\land\neg\leftF{\enab(gcs)}}
}
\]
We can get the conclusion by \rn{rRewrite}, using the equivalence $\ldogc{}{\mathit{false}\gcto\lskipc{}} \kateq \lskipc{}$,
from 
\[  \ldogc{}{gcs} \sep \ldogc{}{\mathit{false}\gcto\lskipc{}} : \rspec{\Q}{\Q\land\neg\leftF{\enab(gcs)}} \]
The latter follows by \rn{rConseq} from 
\begin{equation}\label{eq:X}
 \ldogc{}{gcs} \sep \ldogc{}{\mathit{false}\gcto\lskipc{}} : \rspec{\Q}{\Q\land\neg\leftF{\enab(gcs)}\land\neg\rightF{\mathit{false}}} 
 \end{equation}
(Note that $\enab(\mathit{false}\gcto\lskipc{})$ is $\mathit{false}$.) 
To prove (\ref{eq:X}) we instantiate \rn{dDo} with $\Lrel:=true$ and $\R:=true$.
The side condition holds: it is equivalent to 
$\Q \imp (\leftex{\enab(gcs)} = \rightex{\mathit{false}}) \lor  \leftF{\enab(gcs)}$,
which amounts to the tautology $\Q\imp \neg\leftF{\enab(gcs)} \lor  \leftF{\enab(gcs)}$.
The joint premises have the form 
\[ c\sep \skipc : \rspec{\Q\land \leftF{e}\land\rightF{\mathit{false}}\land\neg true\land\neg true}{\Q} \]
for $e\gcto c$ in $gcs$; these can be derived using \rn{rFalse} and \rn{rConseq}.
Similarly for the right-only premises, which have precondition $\rightF{\mathit{false}}$.
The left-only premises have the form 
\[ c\sep \skipc : \rspec{\Q\land \leftF{e}\land true}{\Q} \]
which follow by \rn{rConseq} from the corresponding premise 
$c\sep \skipc : \rspec{\Q\land \leftF{e}}{\Q}$ of the one-sided loop rule.

\section{Alignment automata, adequacy, and verification conditions}\label{sec:aut}

\subsection{Automata, alignment automata and adequacy}

We adapt a number of technical definitions from~\cite{NagasamudramN21},
where automata are formulated in a way that can represent program semantics using, in essence, a finite control flow graph.

\begin{definition}\label{def:automaton}
An \dt{automaton} is a tuple 
\[ (Ctrl,Sto,\init,\fin,\trans) \]
where $Sto$ is a set (called the data stores),
$Ctrl$ is a finite set (the control points)
that contains distinct elements $\init$ and $\fin$,
and ${\trans} \subseteq (Ctrl\times Sto)\times(Ctrl\times Sto)$ is the transition relation.
We require 
$(n,s)\trans (m,t)$ to imply $n\neq \fin$ and $n\neq m$ 
and call these the \dt{finality} and \dt{non-stuttering} conditions respectively.
Absence of stuttering loses no generality and facilitates definitions involving product automata.
A \dt{state} of the automaton is an element of $Ctrl\times Sto$.
\end{definition}
Let $s,t$ range over stores and $n,m$ over control points.

We use the term ``alignment product'' informally, in reference to various constructions in the literature.
We define a particular construction that we call alignment automaton.

\begin{definition}\label{def:alignProd}
\begin{sloppypar}
Let $A = (Ctrl,Sto,\init,\fin,\trans)$ and 
$A' = (Ctrl',Sto',\init',\fin',\trans')$ be automata.
Let $L$, $R$, and $J$ be subsets of $(Ctrl\times Ctrl')\times (Sto\times Sto')$.
The \dt{alignment automaton} \graybox{$\aprod(A,A',L,R,J)$} is 
the automaton 
\end{sloppypar}
\[ ((Ctrl\times Ctrl'), (Sto\times Sto'), (\init,\init'), (\fin,\fin'), \biTrans) \]
where $\biTrans$ is defined by:
$((n,n'),(s,s')) \biTrans ((m,m'),(t,t'))$ iff one of these conditions holds:
\begin{list}{}{}
\item[LO:]\quad
$((n,n'),(s,s'))\in L$ and
$(n,s)\trans(m,t)$ and $(n',s')=(m',t')$
\item[RO:]\quad
$((n,n'),(s,s'))\in R$ and 
$(n,s)=(m,t)$ and $(n',s')\trans'(m',t')$
\item[JO:]\quad
$((n,n'),(s,s'))\in J$ 
and $(n,s)\trans(m,t)$ and $(n',s')\trans'(m',t')$
\end{list}
The names for the cases refer to \underline{l}eft-\underline{o}nly, 
\underline{r}ight-\underline{o}nly, and \underline{jo}int steps.
The sets $L$, $R$, and $J$ are called \dt{alignment conditions}.
\end{definition}

Notice that the states of $\aprod(A,A',L,R,J)$ are 
$((Ctrl\times Ctrl')\times (Sto\times Sto'))$
according to Def.~\ref{def:automaton}.
So $L$, $R$, and $J$ are sets of alignment automaton states.

We sometimes write $[n|n']$ for the set of states where control is at $(n,n')$, i.e.,
\[ \graybox{$[n|n']$} \eqdef \{ ((i,i'),(s,s')) \mid n=i \land n'=i'\} \]
For example $[\fin|\fin']$ is the set of terminated states.
Similarly, let $\graybox{$[n|*]$} \eqdef \{ ((i,i'),(s,s')) \mid n=i \}$.

\begin{example}\label{eg:alignaut}
Taking $L$ and $R$ to be false and $J$ true, we obtain an automaton that 
runs $A$ and $A'$ in lockstep.
Taking $L,R,J$ all true, we obtain a nondeterministic alignment automaton that represents very many alignments.  Taking $L,R,J$ all false, we obtain an alignment automaton that represents no alignments whatsoever.

Taking $J$ to be false, $L$ to be $[*|\init']$, and $R$ to be $[\fin|*]$, we obtain an automaton that runs only $A$, unless it terminates, in which case it proceeds to run $A'$.
\qed\end{example}

\begin{remark}
\upshape
One generalization of Def.~\ref{def:alignProd} changes ``iff'' to ``only if'', i.e.,
transitions via $\biTrans$ must satisfy LO, RO, or JO, but it is not required to have all such transitions.  For example, according to Def.~\ref{def:alignProd},
if the left automaton $A$ has two successors from $(n,s)$, 
and $((n,m),(s,t))$ is in $L$, then there are two transitions 
from $((n,m),(s,t))$ via $\biTrans$.  The generalization would allow one of the transitions to be omitted.  We eschew it because it requires a more complicated treatment of adequacy
that does not shed light on our main result.
Another generalization of Def.~\ref{def:alignProd} replaces $(Ctrl\times Ctrl')$
by an arbitrary finite set equipped with a projection to $(Ctrl\times Ctrl')$.
This is the formulation in~\cite{NagasamudramN21}; its motivation is handling fault-avoiding partial correctness, which again we avoid for the sake of simplicity. 
For automata from programs, instead of endowing the product with extra control, one can add ghost code (e.g., a step counter) to be exploited by the alignment conditions.
So our formulation is essentially as general.  
\qed
\end{remark}

To be of use in proving relational properties,
an alignment automaton should be chosen that covers all the possible outcomes of the underlying automata.

\begin{definition}[adequacy]\label{def:adequacy}
Consider an alignment automaton $\aprod(A,A',L,R,J)$
and relation $\Q\subseteq Sto\times Sto'$.
The alignment automaton  is \dt{$\Q$-adequate} provided for
all $(s,s')\in\Q$ and $t,t'$ with
$(\init,s)\trans^*(\fin,t)$ and $(\init',s')\trans'^*(\fin',t')$,
we have
$((\init,\init'),(s,s'))\biTrans^* ((\fin,\fin'),(t,t'))$.
\end{definition}

An alignment product may be adequate for reasons that are specific to the 
underlying automata, for example it may not cover all traces but still cover all outcomes.
It may even cover all traces of stores, while not covering all traces when control points are considered.
We focus on products that are adequate in the sense that they cover 
all traces.

Our formulation of Def.~\ref{def:adequacy} is meant to be very general, 
but for some purposes it is useful to slightly restrict the alignment conditions $L,R,J$.

\begin{definition}[live alignment conditions] %\label{def:live}
A triple $(L,R,J)$ is \dt{live} for automata $A,A'$ provided
\[ \begin{array}{l} 
\all{ ((n,n'),(s,s'))\in L }{(n,s)\in\dom(\trans) }, \\
\all{ ((n,n'),(s,s'))\in R }{(n',s')\in\dom(\trans') }, \mbox{and}\\
\all{ ((n,n'),(s,s'))\in J }{(n,s)\in\dom(\trans) \land (n',s')\in\dom(\trans')} 
\end{array}\]
\end{definition}
In Example~\ref{eg:alignaut}, the first example is not live, because $J$ 
includes all states including those that satisfy $[\fin|\fin']$ and have no successors.

\begin{lemma}\label{lem:LRJmin}
\upshape
Consider any $A,A',L,R,J$ and let 
\[ 
\begin{array}{l}
L^\nbl = \{((n,n'),(s,s'))\in L \mid (n,s)\in\dom(\trans)\} \\
R^\nbl = \{((n,n'),(s,s'))\in R \mid (n',s')\in\dom(\trans')\} \\
J^\nbl = \{((n,n'),(s,s'))\in J \mid  (n,s)\in\dom(\trans) \land (n',s')\in\dom(\trans')\}
\end{array}
\]
The alignment automaton $\aprod(A,A',L^\nbl,R^\nbl,J^\nbl)$ is identical to   
$\aprod(A,A',L,R,J)$,
and $(L^\nbl,R^\nbl,J^\nbl)$ is live for $A,A'$.
\end{lemma}
\begin{proof}
Liveness is by construction. To show that the two alignment automata are the same,
note first that they have identical sets of control points, stores, initial and final points.  To see that the transition relations are also identical,
consider  condition LO of Def.~\ref{def:alignProd},
instantiated by $L$ versus $L^\nbl$.
For any $((n,n'),(s,s'))$ in $L$ but not in $L^\nbl$, 
there is no transition by $A$ from $(n,s)$, so condition LO holds 
for neither $L$ nor $L^\nbl$.  
\emph{Mut.\ mut.} for $R$ and $J$.
\end{proof}

Later we construct automata from programs, and for such automata every state has a successor except when control is final (Lemma~\ref{lem:autLive}).  For programs,
an arbitrary triple $(L,R,J)$ can be made live very simply:
$(L\setminus[\fin|*],R\setminus[*|\fin'],J\setminus[\fin|\fin'])$
where $\setminus$ is set subtraction.
In light of Lemma~\ref{lem:LRJmin}, restricting to live alignment conditions loses no generality.  This lets us slightly streamline the technical development, while retaining 
the very general formulation of Def.~\ref{def:alignProd}.

\medskip
\textsc{Assumption 1}.  % ALERT hard coded number
\emph{From here onward, all considered alignment conditions are assumed to be live for their  underlying automata.}

The following strengthening of adequacy ensures that all underlying traces are covered, not just outcomes.  

\begin{definition}[manifest adequacy] %\label{def:manifest}
Consider an alignment automaton $\aprod(A,A',L,R,J)$ and relation $\Q\subseteq Sto\times Sto'$.
The alignment automaton  is \dt{manifestly $\Q$-adequate} provided 
that $L \lor R \lor J \lor [\fin|\fin']$ is $\Q$-invariant.
That is, $L\lor R \lor J \lor [\fin|\fin']$ holds at every state that is reachable from some 
$((\init,\init'),(s,s'))$ such that $(s,s')\in\Q$.
\end{definition}

\begin{lemma} %\label{lem:manifest}
\upshape
If alignment automaton $\aprod(A,A',L,R,J)$ is manifestly $\Q$-adequate 
then it is $\Q$-adequate (in the sense of Def.~\ref{def:adequacy}).
\end{lemma}
\begin{proof}
Informally, for any pair of terminated traces $T,T'$ of $A$ and $A'$,
where $\Q$ holds initially, there is a terminated trace of 
$\aprod(A,A',L,R,J)$ of which $T$ and $T'$ are the left and right
``projections''.
This implies $\Q$-adequacy. 
The idea is essentially liveness of the alignment automaton: owing to invariance of $L \lor R \lor J \lor [\fin|\fin']$, 
we can always take a step via $\biTrans$ on one or both sides---unless both sides are stuck or terminated, in which case the given pair of traces is covered.

To be precise we make the following claim:\\
If $(s,s')\in\Q$, $((\init,\init'),(s,s'))\biTrans^*((n,t),(n',t'))$,
$(n,t)\trans^*(\fin,u)$, and $(n',t')\trans'^*(\fin',u')$
then $((\init,\init'),(s,s'))\biTrans^*((\fin,u),(\fin',u'))$.

The claim implies $\Q$-adequacy, by taking $((n,t),(n',t'))=((\init,s),(\init',s'))$.
The claim is proved by induction on the sum of the lengths of the runs 
$(n,t)\trans^*(\fin,u)$, and $(n',t')\trans'^*(\fin',u')$.
In the base case the sum is zero; then we have $((n,t),(n',t'))=((\fin,u),(\fin',u'))$ and we are done.

For the induction step, 
suppose $(s,s')\in\Q$ and $((\init,\init'),(s,s'))\biTrans^*((n,t),(n',t'))$.
Suppose $(n,t)\trans^*(\fin,u)$, and $(n',t')\trans'^*(\fin',u')$ where at least one of these has a non-zero number of steps.
Since $((n,t),(n',t'))$ is reached from $\Q$, by manifest $\Q$-adequacy we have 
$((n,t),(n',t')) \in L \lor R \lor J$.  Here  $[\fin|\fin']$ can be excluded owing to the non-zero length unary computation.
So a step can be taken via $\biTrans$ from $((n,t),(n',t'))$, via the LO, RO, or JO condition
in Def.~\ref{def:alignProd}, owing to liveness of $(L,R,J)$.  Consider the LO case.  The unary state $(n,t)$ must not be final, so there is $(m,r)$ with $(n,t)\trans(m,r)\trans^*(\fin,u)$.
The alignment automaton can go by LO to $((m,r),(n',t'))$.
So we have 
$((\init,\init'),(s,s'))\biTrans^*((m,r),(n',t'))$
and traces $(m,r)\trans^*(\fin,u)$, $(n',t')\trans'^*(\fin',u')$ 
for which the induction hypothesis applies, yielding the result.
Cases RO and JO are similar.
\end{proof}

\subsection{Correctness of automata}

Generalizing slightly from Sect.~\ref{sec:logics},
we consider specs $\spec{P}{Q}$ where $P$ and $Q$ are sets of automaton stores,
not necessarily variable stores.
Satisfaction of a spec by an automaton is written 
\graybox{$A \models \spec{P}{Q}$} and defined to mean:
\[ \mbox{For all $s,t$ such that $(\init,s)\trans^*(\fin,t)$,
if $s\in P$ then $t\in Q$.} 
\]

Let $A$ and $A'$ be automata with store sets $Sto$ and $Sto'$ respectively.
Again generalizing slightly from Sect.~\ref{sec:logics},
we consider relational specs $\rspec{\R}{\S}$
where $\R$ and $\S$ are relations from $Sto$ to $Sto'$.
Satisfaction of the spec by the pair of automata $A,A'$ is written 
\graybox{$A,A'\models\rspec{\R}{\S}$} and defined 
to mean:
\begin{quote}
For all $s,s',t,t'$
such that $(\init,s)\trans^*(\fin,t)$ and $(\init',s')\trans'^*(\fin',t')$,
if $(s,s')\in\R$ then $(t,t')\in\S$.
\end{quote}

A store relation $\Q$ for $A,A'$ can be seen as a store predicate on the stores
of an alignment automaton $\aprod(A,A',L,R,J)$ because the latter are pairs of stores.
Hence, for relational spec $\rspec{\Q}{\S}$, 
the unary spec $\spec{\Q}{\S}$ makes sense 
for $\aprod(A,A',L,R,J)$.

\begin{proposition}[adequacy sound and complete]\label{prop:adequate-sound}
\upshape
Suppose that $\aprod(A,A',L,R,J)$ is $\Q$-adequate.
Then $\aprod(A,A',L,R,J)\models \spec{\Q}{\S}$
if and only if 
$A,A'\models\rspec{\Q}{\S}$.
\end{proposition}
\begin{proof}
Suppose $\aprod(A,A',L,R,J)\models \spec{\Q}{\S}$.
To prove $A,A'\models\rspec{\Q}{\S}$, consider any pair of terminated 
traces of $A$ and $A'$ from 
initial states $(\init,s)$ and $(\init',s')$ such that $(s,s')\in\Q$.
By adequacy and definition of the alignment automaton,
we can run $\aprod(A,A',L,R,J)$ from $((\init,\init'),(s,s'))$ and get a terminated trace
with the same final pair of stores.  Those stores satisfy $\S$, 
by $\aprod(A,A',L,R,J)\models \spec{\Q}{\S}$.

Conversely, suppose $A,A'\models\rspec{\Q}{\S}$.  
To prove $\aprod(A,A',L,R,J)\models \spec{\Q}{\S}$, consider any terminated trace of 
$\aprod(A,A',L,R,J)$ from initial state 
$((\init,\init'),(s,s'))$ with $(s,s')\in\Q$.
By definition of the alignment product, there are traces of $A$ and $A'$ 
which reach the same final stores.  Those stores satisfy $\S$,
by $A,A'\models\rspec{\Q}{\S}$.  
%The proof is straightforward, using for example that   
%an initial state $((\init,\init'),(s,s'))$ of the alignment automaton corresponds to a pair 
%of initial states $(\init,s)$ and $(\init',s')$ of the underlying automata.
\end{proof}

\subsection{Inductive assertion method} %\label{sec:IAM}

The inductive assertion method~\cite{Floyd67} (\dt{IAM}) is formulated 
in terms of automata.
Given automaton $A$ and spec $\spec{P}{Q}$,
an \dt{annotation} is a function $an$ from control points to store predicates 
such that $an(\init)=P$ and $an(\fin)=Q$.
The requirement $\init\neq\fin$ in Def.~\ref{def:automaton} ensures that annotations exist for any spec.
In Floyd's original formulation, an annotation only needs to be defined on a subset of control points that cut every loop in the control flow graph, but such an annotation can always be extended to one for all control points.  Focusing on what are called ``full'' annotations in~\cite{NagasamudramN21} helps streamline our development.

We \dt{lift} $an$ to a function $\hat{an}$ that yields states:
%% \begin{equation}\label{eq:hat}
%% \hat{an}(n) = \{ (n,s) \mid s\models an(n) \} 
%% \end{equation}
%% Put differently: $\hat{an}(n) = \{n\} \times an(n)$.
$\hat{an}(n) = \{ (m,s) \mid s\in an(n) \}$
or equivalently 
\begin{equation}\label{eq:hatALT}
(m,s)\in\hat{an}(n) \quad\mbox{iff}\quad s\in an(n) \quad\mbox{for all $m,n,s$} 
\end{equation}
For each pair $(n,m)$ of control points
there is a \dt{verification condition} (\dt{VC}):
\begin{equation}\label{eq:VC}
\POST(\tranSeg{n,m})(\hat{an}(n)) \subseteq \hat{an}(m) 
\end{equation}
Here $\POST$ gives the direct image (i.e., strongest postcondition) of a relation,
and \graybox{$\tranSeg{n,m}$} is the transition relation restricted to starting control point $n$ and ending   $m$, i.e.,
\begin{equation}\label{eq:tranSeg}
\mbox{
$(i,s)\tranSeg{n,m}(j,t) \; $ iff $\; i=n$, $j=m$, and $(n,s)\trans(m,t)$
}
\end{equation}
%Using the universal pre-image,\footnote{For relation $R$ 
%   on states and set $X$ of states,
%   $\PRE(R)(X) = \{ \alpha \mid \all{\beta}{\alpha R \beta \imp \beta\in X}\}$
%   and $\POST(R)(X) = \{ \beta \mid \some{\alpha}{\alpha\in X \land \alpha R \beta %} \}$.
%   } %footnote
%(\ref{eq:VC}) is equivalent to
%\( \hat{an}(n) \subseteq \PRE(\trans)(\hat{an}(m)) \).
The VC (\ref{eq:VC}) says that for every transition from control point $n$ and store $s\in an(n)$,
if the step goes to control point $m$ with store $t$, then $t$ is in $an(m)$.
Annotation $an$ is \dt{valid} if the VC is  true for every pair $(n,m)$ of control points.

In most automata, including those we derive from programs, some pairs $(n,m)$ have no transitions,
in other words $\tranSeg{n,m}$ is empty.
In that case the VC (\ref{eq:VC}) is true regardless of $an(n)$ and $an(m)$.

\begin{proposition}[soundness of IAM~\cite{Floyd67}]\label{prop:IAM}
\upshape
If $A$ has an annotation for $\spec{P}{Q}$ that is valid then $A\models \spec{P}{Q}$.
\end{proposition}

\begin{proposition}[semantic completeness of IAM~\cite{Floyd67}]\label{prop:iamcomplete}
\upshape
Suppose $A\models \spec{P}{Q}$.
Then there is an annotation for $\spec{P}{Q}$ that is valid.
\end{proposition}
To prove completeness, one sets $an(\init) = P$,
$an(\fin) = Q$,
and for all other control points $n$ set $an(n)$ to be the strongest invariant
of traces from $P$.

Combining Prop.~\ref{prop:IAM} with Prop.~\ref{prop:adequate-sound} we get the following.

\begin{corollary}[soundness of alignment automata]
\label{cor:relIAM}
\upshape
Suppose that $\aprod(A,A',L,R,J)$ is $\Q$-adequate and 
$an$ is an annotation of $\aprod(A,A',L,R,J)$ for $\spec{\Q}{\S}$.
If $an$ is valid then $A,A'\models \rspec{\Q}{\S}$.
\end{corollary}

Manifest adequacy is important for two reasons.  First, it ensures coverage of all traces, not just outcomes, and so ensures soundness of alignment automata for fault-avoiding partial correctness; but that is beyond the scope of this paper.  Second, it is amenable to inductive proof. 
A sufficient condition is for 
$L \lor R \lor J \lor [\fin|\fin']$ to be an inductive invariant, but
the following is a better alternative.

\begin{lemma}\label{lem:anLRJF}
\upshape 
Suppose $an$ is a valid annotation of $\aprod(A,A',L,R,J)$ for $\rspec{\Q}{\S}$.
If $an(i,j)\imp L \lor R \lor J \lor [\fin|\fin']$ for every $i,j$
then $\aprod(A,A',L,R,J)$ is manifestly $\Q$-adequate.
\end{lemma}
\begin{proof}
Given the implications $an(i,j)\imp L \lor R \lor J \lor [\fin|\fin']$,
validity of $an$ implies that $L \lor R \lor J \lor [\fin|\fin']$ 
is $\Q$-invariant (although it need not be an inductive invariant on its own).  
\end{proof}

Together, Corollary~\ref{cor:relIAM} and Lemma~\ref{lem:anLRJF} provide a method to verify
$A,A'\models \rspec{\Q}{\S}$:
Find alignment conditions $L,R,J$ and annotation $an$ of $\aprod(A,A',L,R,J)$
such that the annotation is valid and 
$an(i,j)\imp L \lor R \lor J \lor [\fin|\fin']$ for every $i,j$.

\begin{remark}\upshape
Churchill et al~\cite{ChurchillP0A19} and Shemer et al~\cite{ShemerGSV19}
use essentially this method, with various techniques for finding $L,R,J$
expressible in SMT-supported assertion languages.  
Churchill et al 
also use testing to evaluate whether a candidate $L,R,J$ is invariant.  
This suggests  a variation on the method:
First find $L,R,J$ such that $L \lor R \lor J \lor [\fin|\fin']$ 
is an inductive $\Q$-invariant of $\aprod(A,A',L,R,J)$.
Then find a valid annotation.
\qed\end{remark}

\begin{example}
This example shows that the alignment conditions of an alignment automaton
need not be exhaustive.
For brevity, write $sk^i$ for the assignment $\lassg{i}{y}{y}$ which acts like skip.
\[
\begin{array}{ll}
c: & 
\lassg{1}{x}{x+1} ; \:
sk^2; \:
\lassg{3}{x}{x+1} ; \:
sk^4; \:
\lassg{5}{x}{x+1} \\
d: & 
sk^1; \:
\lassg{2}{x}{x+1} ;
\lassg{3}{x}{x+1} ;
\lassg{4}{x}{x+1} ; \:
sk^5
\end{array}
\]
Then we have
\[ \models c\sep d :\rspec{x=x'}{x=x'} \]
and it can be proved using this alignment of control point pairs,
where we use 6 as final label.
%\footnote{
%% In relRL we might like to express the alignment as 
%% \[ \begin{array}{l}
%% (\skipc|y:=y);        %1 1
%% (x:=x+1|x:=x+1);      %1 2
%% (y:=y|\skipc);        %2 3
%% (x:=x+1|x:=x+1) \\    %3 3
%% ;(y:=y|\skipc);       %4 4
%% (x:=x+1|x:=x+1);      %5 4
%% (\skipc|y:=y)         %6 5
%% \end{array}
%% \]
%% which we can get by weaving and the quotienting that lets us introduce skips.
%% For example, 
%% $ (\lassg{1}{x}{x+1} ; \ldots \Sep \lassg{1}{y}{y} ;\lassg{2}{x}{x+1};\ldots) $
%% is identified with  
%% $ (\skipc;\lassg{1}{x}{x+1} ; \ldots \Sep \lassg{1}{y}{y} ;\lassg{2}{x}{x+1};\ldots) $
%% and then we can do sequence weaving to
%% $(\skipc|y:=y);(x:=x+1;\ldots|x:=x+1;\ldots)$.
%% }
\[ (1,1), (1,2), (2,3), (3,3), (4,4), (5,4), (6,5), (6,6) \]
This is represented by an alignment automaton
using 
$L\eqdef[2|3]\lor[4|4]$,
$R\eqdef[1|1]\lor[6|5]$, and 
$J\eqdef[1|2]\lor[3|3]\lor[5|4]$.  Note that $(L,R,J)$ is live.
The same automaton is obtained from 
$L\eqdef[2|*]\lor[4|*]$,
$R\eqdef[*|1]\lor[*|5]$, and the same $J$; again this is live. 

Define $an(i,j)$ to be $x=x'$, for all $(i,j)$ in
\[ \{ (1,1), (1,2), (2,3), (3,3), (4,4), (5,4), (6,5), (6,6) \} \]
and $an(i,j)=\mathit{false}$ for all others.
The annotation is valid and the alignment automaton is manifestly
$(x=x')$-adequate.

Note that $L\lor R\lor J\lor [\fin|\fin']$ does not include all states,
e.g., it excludes states with control $(3,2)$.
Note also that $L\lor R\lor J\lor [\fin|\fin']$ is $(x=x')$-invariant, i.e., holds in 
all states reachable for an initial state satisfying the precondition $x=x'$.  

If we drop $[2|3]$ from $L$ then $L\lor R\lor J\lor [\fin|\fin']$ is no longer invariant,
since it does not hold when control is at $(2,3)$; and the automaton is no longer adequate.
\qed\end{example}

%% \begin{example}
%% \dn{Possible tiny example to sketch early in the paper?}
%% Consider the program 
%% \[\lifgc{1}{x>0\gcto \lassg{2}{y}{y+1} \gcsep x<0\gcto \lassg{3}{y}{y-1}} \]
%% It satisfies $\rspec{x=x'\land y=y'}{y=y'}$, and 
%% this can be proved using a lockstep automaton, where we write the alignment 
%% conditions as formulas:
%% \(\begin{array}[t]{ll}
%% L: & \mathit{false} \\
%% R: & \mathit{false} \\
%% J: & \quant{\lor}{i}{1\leq i \leq 3}{\bothF{ \tpc i\sep \tpc i }}
%% \end{array}
%% \).
%% \\
%% Here $J^\nbl$ is 
%% \[ (\leftF{ \tpc i\land x\neq 0 }\land\rightF{ \tpc i\land x\neq 0 })
%% \lor 
%% (\leftF{ \tpc 2 }\land\rightF{ \tpc 2 })
%% \lor 
%% (\leftF{ \tpc 3 }\land\rightF{ \tpc 3 })
%% \]
%% because the conditional has no transitions from states where $x=0$.
%% The product is $(x=x'\land y=y')$-adequate, manifestly since
%% $J^\nbl\lor [\fin|\fin']$ is an inductive $(x=x'\land y=y')$-invariant.
%% \end{example}

\begin{corollary}[semantic completeness of alignment automata]
\label{cor:relIAMcomplete}
\upshape
Suppose $A,A'\models \rspec{\S}{\T}$.
Then there are $L,R,J$ and an annotation $an$ of $\aprod(A,A',L,R,J)$ for $\spec{\S}{\T}$ 
such that the alignment automaton is manifestly $\S$-adequate 
and the annotation is valid.
\end{corollary}
\begin{proof}
Let $L,R,J$ be the set of all states except $[\fin|\fin']$,
i.e., $L=R=J=((Ctrl\times Ctrl')\times(Sto\times Sto'))\setminus [\fin|\fin']$,
which is manifestly $true$-adequate and hence $\S$-adequate.
(Subtracting $[\fin|\fin']$ ensures liveness of $L,R,J$.)
By Prop.~\ref{prop:adequate-sound} 
from $A,A'\models \rspec{\S}{\T}$ we have 
$\aprod(A,A',L,R,J) \models \spec{\S}{\T}$.
So by completeness of IAM (Prop.~\ref{prop:iamcomplete}) there is an annotation of 
$\aprod(A,A',L,R,J)$ for $\spec{\S}{\T}$ that is valid.  
\end{proof}

\subsection{Automata from programs and their VCs}

\begin{figure}[t] 
%\begin{footnotesize}
\begin{small}
\begin{mathpar}
\inferrule{ e\gcto c \mbox{ is in } gcs \\ \means{e}(s) = \mathit{true} }
{ \config{ \lifgc{n}{gcs} }{s} \ctrans \config{c}{s} }

\inferrule{ e\gcto c \mbox{ is in } gcs \\ \means{e}(s) = \mathit{true} }
{ \config{ \ldogc{n}{gcs} }{s}  \ctrans 
  \config{ c;\ldogc{n}{gcs} }{s} 
}

\inferrule{ \enab(gcs)(s) = \mathit{false} }
{ \config{\ldogc{n}{gcs} }{s} \ctrans  \config{\lskipc{-n}}{s}
}

\inferrule{
\config{c}{s} \ctrans \config{d}{t} }
{ \config{c;b}{s} \ctrans \config{d;b}{t} }

\inferrule{}{ \config{\lassg{n}{x}{e}}{s} \ctrans \config{\lskipc{-n}}{\update{s}{x}{\means{e}(s)}} }

\inferrule{}{ \config{\lskipc{n};c}{s} \ctrans \config{c}{s} }

\end{mathpar}
%\end{footnotesize}
\end{small}
\vspace*{-3ex}
\caption{Transition semantics (with $n$ ranging over $\Z$).}\label{fig:progtrans}
\end{figure}

Labels on commands serve as basis for defining the automaton for a program,
to which end we make the following definitions.
Write \graybox{$\ok(c)$} to say no label in $c$ occurs more than once and all labels are positive.  
We have $\ok(c0)$ for the running example $c0$ defined in Sect.~\ref{sec:unary}.
Write $\lab(c)$ for the label of $c$ and $\labs(c)$ for its set of labels.
(Details in Appendix~\ref{app:details}.)

For small-step semantics,
we write $\config{c}{s} \ctrans \config{d}{t}$ if command $c$ with store $s$ transitions to continuation command $d$ and store $t$.
The transition relation $\ctrans$ is defined in Fig.~\ref{fig:progtrans}.

Negative labels are used in the small-step semantics
in a way that facilitates defining the automaton of a program.
In a configuration reached from an $\ok$ command, 
the only negative labels are those introduced by the transition for assignment
and the transition for termination of a loop.
%(As an aside, the other loop transition duplicates the loop body, creating non-unique labels.)
For every $c,s$, either $\config{c}{s}$ has a successor via $\ctrans$
or $c$ is $\lskipc{n}$ for some $n\in\Z$.  (This relies on $\totalIf$ of Def.~\ref{def:lang}.)
When a negative label is introduced, the configuration is either terminated 
or has the form $\config{\lskipc{-n};c}{s}$
in which case the next transition is 
$\config{\lskipc{-n};c}{s} \ctrans \config{c}{s}$.
% which eliminates the negative-label skip.

\begin{lemma}[big/small step consistency]\label{lem:bigsmall}
\upshape
For any $c,s,t$,
\[ %begin{equation}\label{eq:cevalMultiStepEquiv} % named after Coq lemma
\means{c}\, s\, t
\quad\mbox{ iff }\quad
\config{c}{s} \ctrans^* \config{\lskipc{n}}{t}
\mbox{ for some $n$}
\]
\end{lemma}

Write \graybox{$\sub(n,c)$} for the sub-command of $c$ with label $n$, if $n$ is in $\labs(c)$.
%For example, $\sub(3,c0)$ is $\lifgc{3}{\ldots}$.
Let $c$ and $\fin$ such that $\ok(c)$ and $\fin\notin\labs(c)$.
We write $\fsuc(n,c,\fin)$ for the \dt{following successor} of $n$ in the control flow graph of $c$, in the sense made precise in Fig.~\ref{fig:fsuc}.
Note that $\fin$ serves as a final or exit label.
(The definition of $\fsuc(n,c,\fin)$ assumes $\ok(c)$, $n\in\labs(c)$, and $\fin\notin\labs(c)$.)
The key case in the definition is for loops: the following successor is the control point after termination of the loop.
For the running example, we have $\fsuc(2,c0,6) = 6$ and $\fsuc(4,c0,6)=2$.

\begin{figure}[t]
\begin{small}
\[
\begin{array}{l@{\;}c@{\;}l}
\fsuc(n,\lskipc{n},f)          & \eqdef & f \\
\fsuc(n,\lassg{n}{x}{e},f)     & \eqdef & f \\
\fsuc(n,c;d,f)                 & \eqdef & \fsuc(n,c,\lab(d)) \mbox{ , if $n\in\labs(c)$} \\
                               & \eqdef & \fsuc(n,d,f) \mbox{ , otherwise} \\ 
\fsuc(n,\lifgc{n}{gcs}, f) &\eqdef & f \\
\fsuc(m,\lifgc{n}{gcs}, f) &\eqdef & \fsuc(m,c,f) \mbox{ , if $e\gcto c \in gcs$ and $m\in\labs(c)$} \\
\fsuc(n,\ldogc{n}{gcs}, f) &\eqdef & f \\
\fsuc(m,\ldogc{n}{gcs}, f) &\eqdef & \fsuc(m,c,n) \mbox{ , if $e\gcto c \in gcs$ and $m\in\labs(c)$} 
\end{array}
\]
\end{small}
\vspace*{-1ex}
\caption{Following successor $\fsuc(n,c,f)$}
\label{fig:fsuc}
\end{figure}

Define \graybox{$\okf(c,f)$} (``ok, fresh'') to abbreviate the conjunction
of $\ok(c)$ and $f\notin\labs(c)$.

\begin{definition}\label{def:aut}
%Let $c$ and $f\in\nat$ satisfy $\ok(c)$ and $f\notin\labs(c)$.
Suppose $\okf(c,f)$. 
The \dt{automaton of $c$ for $f$}, written \graybox{$\aut(c, f)$}, is
\[ (\labs(c)\union \{ f \}, (\Var\to\Z), \lab(c), f, \trans) \] 
where 
$(n,s)\trans (m,t)$ iff either
\begin{small}
\begin{ditemize}
\item $
\some{d}{ \config{\sub(n,c)}{s}\ctrans\config{d}{t} 
   \land \lab(d) > 0 \land m = \lab(d) }$, or 
\item % format hack:
$\some{d}{ \config{\sub(n,c)}{s}\ctrans\config{d}{t} \land \lab(d) < 0 \land m = \fsuc(n,c,f) }$, or 
\item $\sub(n,c)=\lskipc{n} \land m = \fsuc(n,c,f) \land t = s$
\end{ditemize}
\end{small}
\end{definition}
The first two cases use the semantics of Fig.~\ref{fig:progtrans}
for a sub-command on its own.  
The second case uses $\fsuc$ for a sub-command that takes a terminating step 
(either an assignment or a loop).
The third case handles skip, which on its own would be stuck 
but which should take a step when it occurs as part of a sequence.

\begin{example}\label{ex:steps}
We remark that the small step and automaton computations are similar but do not quite match 
step-by-step.\footnote{In~\cite{NagasamudramN21}, 
using similar definitions, 
the authors claim (Lemma 4) that the automaton steps correspond one-to-one with steps via $\ctrans$. This is not quite true, because 
traces via $\ctrans$ include the steps of the form 
$\config{\lskipc{n};c}{s} \ctrans \config{c}{s}$ where $n<0$.
There are no corresponding steps in the automaton, as the examples show.
This does not undercut the other results in~\cite{NagasamudramN21},
which only depend on initial and final configurations of a program
versus its automaton.
}
For the command $\lassg{1}{x}{0};\lassg{2}{y}{0}$ we have (for some stores $s_0,s_1,s_2$)
the terminated computation 
\[ \config{\lassg{1}{x}{0};\lassg{2}{y}{0}}{s_0}
\ctrans
\config{\lskipc{-1};\lassg{2}{y}{0}}{s_1}
\ctrans
\config{\lassg{2}{y}{0}}{s_1}
\ctrans
\config{\lskipc{-2}}{s_2}
\]
This corresponds to the following computation of 
$\aut((\lassg{1}{x}{0};\lassg{2}{y}{0}),3)$:
\[ (1,s_0)\trans(2,s_1)\trans(3,s_2) \]
As another example, we have
\[ \config{\lskipc{1};\lassg{2}{y}{0}}{t_0}
\ctrans
\config{\lassg{2}{y}{0}}{t_0}
\ctrans
\config{\lskipc{-2}}{t_1}
\]
and for $\aut((\lskipc{1};\lassg{2}{y}{0}),3)$ we have 
\[ (1,t_0)\trans(2,t_0)\trans(3,t_1) \]
The state $(3,t_1)$ is terminated.
There is no transition from $(3,t_1)$ because none of the three
conditions in Def.~\ref{def:aut} holds,
noting that $\sub(3,(\lskipc{1};\lassg{2}{y}{0}))$
is not defined.
\qed\end{example}

\begin{lemma}[automaton liveness]\label{lem:autLive}
\upshape
The only stuck states of $\aut(c, f)$ are terminated ones, i.e.,  where the control is $f$.
\end{lemma}
\begin{proof}
For any state $(n,s)$ where $n\neq f$, $n$ is the label of a command in $c$
and by definition of the small step relation $\ctrans$,
and the automaton transition $\trans$, 
there at least one successor of $(n,s)$.
In the case where $n$ labels an if-command, 
this relies on the $\totalIf$ condition (Def.~\ref{def:lang}).
\end{proof}

\begin{lemma}[automaton consistency]\label{lem:autConsistent}
\upshape
Suppose $\ok(c)$ and $f\notin\labs(c)$ and  let $n=\lab(c)$.
For any $s,t$ we have 
\[ \means{c}\, s\, t \quad\mbox{iff}\quad (n,s)\trans^*(f,t) \quad\mbox{(in $\aut(c,f)$)} \]
Hence 
\[ \models c:\spec{P}{Q} \quad\mbox{iff}\quad \aut(c,f)\models\spec{P}{Q} \] 
for any $P,Q$, and also 
\[ \models c\sep c':\rspec{\Q}{\S} \quad\mbox{iff}\quad \aut(c,f),\aut(c',f')\models\rspec{\Q}{\S} \] 
for any $\Q,\S$ and any $c',f'$ with $\ok(c')$ and $f'\notin\labs(c')$.
\end{lemma}
\begin{proof}
The proof of the first equivalence uses Lemma~\ref{lem:bigsmall}
and also the fact that 
$\config{c}{s}\ctrans^*\config{\lskipc{k}}{t}$ (for some $k$)
iff 
$(n,s)\trans^*(f,t)$
(noting the slight complication shown in Example~\ref{ex:steps}).
The other equivalences follow by definitions from the first.
\end{proof}

\paragraph{Verification conditions}

\begin{figure*}
\begin{footnotesize}
\begin{tabular}{lll}
if $\sub(n,c)$ is\ldots & and $m$ is \ldots & then the VC for $(n,m)$ is equivalent to\ldots \\\hline

$\lskipc{n}$ &
$\fsuc(n,c,f)$ & $an(n)\imp an(m)$ 
\\

$\lassg{n}{x}{e}$ &
$\fsuc(n,c,f)$ & $an(n)\imp \subst{an(m)}{x}{e}$ 
\\

$\lifgc{n}{gcs}$ & $\lab(d)$ where $e\gcto d$ is in $gcs$
& $an(n)\land e \imp an(m)$ 
\\

$\ldogc{n}{gcs}$ & $\lab(d)$ where $e\gcto d$ is in $gcs$
& $an(n)\land e \imp an(m)$ 
\\

$\ldogc{n}{gcs}$ & $\fsuc(n,c,f)$ & $an(n)\land \neg \enab(gcs) \imp an(m)$
\\[.5ex]
\multicolumn{3}{l}{
In all other cases, there are no transitions from $n$ to $m$ so the VC is $true$ by definition.}
\end{tabular}
\end{footnotesize}
\vspace*{-1ex}
\caption{VCs for the automaton $aut(c,f)$ of $\ok$ program $c$ and annotation $an$.}
\label{fig:VC}
\end{figure*}

By inspection of the transition semantics in Fig.~\ref{fig:progtrans},
there are five kinds of transitions
for $\ctrans$ and also for the automaton relation $\trans$ derived from it.
(There is a sixth rule for $\ctrans$ that says a transition can occur for the first command in a sequence, but that is used together with one of the other five, 
and it is not relevant to Def.~\ref{def:aut}.)
Thus there are five kinds of verification conditions, which can be derived from 
the semantic definitions.

\begin{lemma}[VCs for programs]\label{lem:VCprog}
\upshape
Consider $c$ and $f$ such that $\ok(c,f)$, and 
let $an$ be an annotation of $\aut(c,f)$.
For each pair $n,m$ of labels, the VC of Eqn.~(\ref{eq:VC}) 
can be expressed as in Fig.~\ref{fig:VC}.
\end{lemma}
Recall that we are treating assertions as sets of stores, but with formula notations.
So in Fig.~\ref{fig:VC}, $an(n)\land e$ means $an(n)\intersect\means{e}$.
The entries in Fig.~\ref{fig:VC} are obtained 
from the condition (\ref{eq:VC}) on $\hat{an}$ by straightforward 
unfolding of definitions ---in particular, Def.~\ref{def:aut} of $\aut(c,f)$
and the definition of $\ctrans$ in Fig.~\ref{fig:progtrans}.

\begin{proof}
\textbf{Case} $\lskipc{n}$ with $m = \fsuc(n,c,f)$.
To show: the VC is equivalent to $an(n)\imp an(m)$.

The VC of (\ref{eq:VC}) is
\[ \POST(\tranSeg{n,m})(\hat{an}(n)) \subseteq \hat{an}(m) \]
which equivales (using def $\POST$)
\[ \all{i,s,j,t}{ (i,s)\in\hat{an}(n) \land (i,s)\tranSeg{n,m}(j,t)
     \; \imp \; (j,t)\in \hat{an}(m) } \]
which equivales (using def $\tranSeg{n,m}$ and one-point rule of predicate calculus)
\[ \all{s,t}{ (n,s)\in\hat{an}(n) \land (n,s)\trans(m,t)
     \imp (m,t)\in \hat{an}(m) } \]
which equivales (using def (\ref{eq:hatALT}))
\[ \all{s,t}{ s\in an(n) \land (n,s)\trans(m,t) \imp t\in an(m) } \]
which equivales (using def $\trans$ for the case of skip with successor $m$)
\[ \all{s,t}{ s\in an(n) \land s = t \imp t\in an(m) } \]
which equivales (using one-point rule)
\[ \all{s}{ s\in an(n) \imp s\in an(m) } \]
This is equivalent to $an(n)\subseteq an(m)$ which we write as $an(n)\imp an(m)$.

\medskip
\textbf{Case}
$\lassg{n}{x}{e}$ with $m=\fsuc(n,c,f)$. To show: the VC is equivalent to $an(n)\imp \subst{an(m)}{x}{e}$.

As in the previous case, the VC is equivalent to 
\[ \all{s,t}{ s\in an(n) \land (n,s)\trans(m,t) \imp t\in an(m) } \]
which equivales (using def $\trans$ for case assignment $\lassg{n}{x}{e}$ with successor $m$)
\[ \all{s}{ s\in an(n) \imp \update{s}{x}{\means{e}(s)}\in an(m) } \]
which equivales (using def of substitution)
\[ \all{s}{ s\in an(n) \imp s \in \subst{an(m)}{x}{e} } \]
which we write as $an(n) \imp \subst{an(m)}{x}{e}$.

\medskip
\textbf{Case}
$\lifgc{n}{gcs}$ with $m=\lab(d)$ where $e\gcto d$ is in $gcs$.
To show: the VC is equivalent to $an(n)\land e \imp an(m)$.

As in the previous case, the VC is equivalent to 
\[ \all{s,t}{ s\in an(n) \land s\in\means{e} \land (n,s)\trans(m,t) \imp t\in an(m) } \]
which equivales (using def $\trans$ for if-case with successor $m$ and guard $e$)
\[ \all{s,t}{ s\in an(n) \land s\in\means{e} \land s=t \imp t\in an(m) } \]
which equivales (using one-point rule)
\[ \all{s}{ s\in an(n) \land s\in\means{e} \imp s\in an(m) } \]
which we write as $an(n)\land e \imp an(m)$.

The two cases for loop are similar.
% 
%% $\ldogc{n}{gcs}$ & $\lab(d)$ where $e\gcto d$ is in $gcs$
%% & $an(n)\land e \imp an(m)$ 
%% \\
%% $\ldogc{n}{gcs}$ & $\fsuc(n,c,f)$ & $an(n)\land \neg \enab(gcs) \imp an(m)$
\end{proof}

\begin{remark}\upshape
For practical purposes, it is straightforward to define the control flow graph of a program,
and thereby identify pairs $n,m$ for which there are definitely no transitions,
e.g., $(1,3)$ in $c0$. 
These can be excluded from the definition of valid annotation,
because $\tranSeg{n,m}$ is empty so the VC holds regardless of the annotation. 
We have no other need to use control flow graphs so for simplicity we do not exclude the vacuous VCs.
\qed\end{remark}

\begin{example}\label{ex:c0an}
For the running example $c0$, suppose we have some $P,Q$ and annotation $an$ 
of $\aut(c0,6)$ for $\spec{P}{Q}$.
So $an(1)=P$, $an(6)=Q$.
Here are the nontrivial VCs according to Fig.~\ref{fig:VC}.
For assignment, we note that $an(n)\imp \subst{an(m)}{x}{e}$ 
is equivalent to $\models \lassg{n}{x}{e}:\spec{an(n)}{an(m)}$, by semantics of assignment
and substitution.
\begin{equation}\label{eq:VCc0}
\begin{array}{l}
%\models \lassg{1}{x}{1} : \spec{an(1)}{an(2)} \\
an(1)\imp\subst{an(2)}{x}{y} \\
x>0 \land an(2) \imp an(3) \\
\neg(x>0) \land an(2) \imp an(6)\\
x\mod 2=0 \land an(3) \imp an(4) \\
x\mod 2\neq 0 \land an(3) \imp an(5) \\
%\models \lassg{4}{x}{x-1}: \spec{an(4)}{an(2)} \\
an(4)\imp\subst{an(2)}{x}{x-1} \\
%\models \lassg{5}{x}{x-2}: \spec{an(5)}{an(2)} \\
an(5)\imp\subst{an(2)}{x}{x-2}
\end{array}
\end{equation}
Suppose that the VCs hold, i.e., the annotation is valid.
Then by Prop.~\ref{prop:IAM} we have $\aut(c0,f)\models \spec{P}{Q}$
so by Lemma~\ref{lem:autConsistent} we have
$\models c0:\spec{P}{Q}$.
\qed\end{example}

From $\models c0:\spec{P}{Q}$, 
by completeness
of HL~\cite{Cook78,AptOld3} there is a proof of $c0:\spec{P}{Q}$ in HL.
Even more, by the Floyd completeness result of~\cite{NagasamudramN21} 
there is a proof in HL that uses just the judgments associated with the 
annotated automaton. 
In Sect.~\ref{sec:c0unary} we show how to obtain a proof in HL+ that embodies the automaton-based proof explicitly. 

\subsection{Relational VCs}\label{sec:RVC}

\begin{figure*}
\begin{small}
\begin{tabular}{lll}
if $\sub(n,c)$ is \ldots     & 
and $m$ is\ldots &
then the VC for $((n,n'),(m,n'))$ is equivalent to $\ldots$ 
\\\hline
$\lskipc{n}$ & $\fsuc(n,c,f)$ & $L \land \breve{an}(n,n')\imp \hat{an}(m,n')$ 
\\
$\lassg{n}{x}{e}$ & $\fsuc(n,c,f)$ & $L\land \breve{an}(n,n')\imp \subst{\hat{an}(m,n')}{x|}{e|}$ 
\\
$\lifgc{n}{gcs}$ & $\lab(d)$ where $e\gcto d$ is in $gcs$ 
& $L\land \breve{an}(n,n')\land \leftF{e} \imp \hat{an}(m,n')$ 
\\ 
$\ldogc{n}{gcs}$ & $\lab(d)$ where $e\gcto d$ is in $gcs$ 
& $L\land \breve{an}(n,n')\land \leftF{e} \imp \hat{an}(m,n')$ 
\\ 
$\ldogc{n}{     gcs}$ & $\fsuc(n,c,f)$ 
& $L\land \breve{an}(n,n')\land \neg\leftF{\enab(gcs)} \imp \hat{an}(m,n')$ 
\\
\multicolumn{3}{l}{
In all other cases, there are no transitions from $(n,n')$ to $(m,n')$ so the VC is $true$ by definition.}
\end{tabular}
\end{small}
\vspace*{-1ex}
\caption{The left-only VCs for annotation $an$ of $\aprod(\aut(c,f),\aut(c',f'),L,R,J)$.}
\label{fig:RVClo}
\end{figure*}

% FORMATTING FOR TABLE BELOW - arydshln package
\setlength{\dashlinedash}{.5ex}
\setlength{\dashlinegap}{1ex}

\begin{figure*}
\begin{footnotesize}
\begin{tabular}{lll}
if $\begin{array}[t]{l} 
    \sub(n,c) \\
    \sub(n',c') 
    \end{array}$
    are\ldots     & 
and $\begin{array}[t]{l} m \\ m' \end{array}$ are\ldots &
then the VC for $((n,n'),(m,m'))$ is equivalent to $\ldots$ \\\hline

$\lskipc{n}$ &
$\fsuc(n,c,f)$ & $J \land \breve{an}(n,n')\imp \hat{an}(m,m')$ 
\\
$\lskipc{n'}$ &
$\fsuc(n',c',f')$ & 
\\\hdashline

$\lassg{n}{x}{e}$ &
$\fsuc(n,c,f)$ & $J \land \breve{an}(n,n')\imp \subst{\hat{an}(m,m')}{x|x'}{e|e'}$ 
\\
$\lassg{n'}{x'}{e'}$ &
$\fsuc(n',c',f')$ & 
\\\hdashline

$\lifgc{n}{gcs}$ & $\lab(d)$ where $e\gcto d$ is in $gcs$ 
& $J \land \breve{an}(n,n')\land \leftF{e} \land \rightF{e'} \imp \hat{an}(m,m')$ 
\\
$\lifgc{n'}{gcs'}$ & $\lab(d')$ where $e'\gcto d'$ is in $gcs'$
& 
\\\hdashline

$\ldogc{n}{gcs}$ & $\lab(d)$ where $e\gcto d$ is in $gcs$ 
& $J \land \breve{an}(n,n')\land \leftF{e} \land \rightF{e'} \imp \hat{an}(m,m')$ 
\\
$\ldogc{n'}{gcs'}$ & $\lab(d')$ where $e'\gcto d'$ is in $gcs'$
& 
\\\hdashline

$\ldogc{n}{gcs}$ & $\fsuc(n,c,f)$  
& $J \land \breve{an}(n,n')\land \neg\leftF{\enab(gcs)} \land \neg\rightF{\enab(gcs')} \imp \hat{an}(m,m')$ 
\\
$\ldogc{n'}{gcs'}$ & $\fsuc(n',c',f')$ & 
\\\hdashline

$\ldogc{n}{gcs}$ & $\lab(d)$ where $e\gcto d$ is in $gcs$ 
& $J \land \breve{an}(n,n')\land \leftF{e} \land \neg\rightF{\enab(gcs')} \imp \hat{an}(m,m')$ 
\\
$\ldogc{n'}{gcs'}$ & $\fsuc(n',c',f')$ & 
\\\hdashline

$\ldogc{n}{gcs}$ & $\fsuc(n,c,f)$ 
& $J \land \breve{an}(n,n')\land \neg\leftF{\enab(gcs)} \land \rightF{e'} \imp \hat{an}(m,m')$ 
\\
$\ldogc{n'}{gcs'}$ & $\lab(d)$ where $e\gcto d$ is in $gcs$ 
\\\hdashline

$\lassg{n}{x}{e}$ & $\fsuc(n,c,f)$ & $J \land \breve{an}(n,n')\imp \subst{\hat{an}(m,m')}{x|}{e|}$  \\
$\lskipc{n'}$ & $\fsuc(n',c',f')$ & 
\\\hdashline

%% $\lskipc{n}$ & $\fsuc(n,c,f)$ 
%% & $J \land \breve{an}(n,n')\land \rightF{e'} \imp \hat{an}(m,m')$ 
%% \\
%% $\lifgc{n'}{gcs'}$ & $\lab(d')$ where $e'\gcto d'$ is in $gcs'$
%% & 
%% \\\hdashline

$\lassg{n}{x}{e}$ & $\fsuc(n,c,f)$ 
& $J \land \breve{an}(n,n')\land \rightF{e'} \imp \subst{\hat{an}(m,m')}{x|}{e|}$ 
\\
$\lifgc{n'}{gcs'}$ & $\lab(d')$ where $e'\gcto d'$ is in $gcs'$
\\\hdashline

$\lassg{n}{x}{e}$ & $\fsuc(n,c,f)$ & $J\land \breve{an}(n,n')\land \rightF{e'} \imp \subst{\hat{an}(n,m')}{x|}{e|}$  \\
$\ldogc{n'}{gcs'}$ & $\lab(d')$ where $e'\gcto d'$ is in $gcs'$ \\\hdashline

$\lassg{n}{x}{e}$ & $\fsuc(n,c,f)$ & $J\land \breve{an}(n,n')\land \neg\rightF{e'} \imp \subst{\hat{an}(n,m')}{x|}{e|}$  \\
$\ldogc{n'}{gcs'}$ & $\fsuc(n',c',f')$ & 

\\[.5ex]
\multicolumn{3}{l}{Omitted: the other 14 cases with nontrivial VCs.}
\end{tabular}

\end{footnotesize}
\vspace*{-1ex}
\caption{Selected joint VCs for $\aprod(\aut(c,f),\aut(c',f'),L,R,J)$ and annotation $an$.}
\label{fig:RVCjo}
\end{figure*}

Consider an alignment automaton $\aprod(A,A',L,R,J)$ where the underlying automata $A$ and $A'$ are obtained from programs by Def.~\ref{def:aut}.
An annotation of $\aprod(A,A',L,R,J)$ thus maps pairs of control points of $A,A'$ to
relations on variable stores.  
Verification conditions are associated with tuples $((n,n'),(m,m'))$,
for which take $n:=(n,n')$ and $m:=(m,m')$ in Eqn.~(\ref{eq:VC}). 

For unary VCs, there are five cases in Fig.~\ref{fig:VC}, corresponding to the five transition rules in Fig.~\ref{fig:progtrans}
as noted preceding Lemma~\ref{lem:VCprog}.
For a given pair $((n,n'),(m,m'))$ of alignment automaton control points,
the transitions go only via LO, or only via RO, or only via JO in Def.~\ref{def:alignProd}.
If $n=m$ (resp.\ $n'=m'$) the transitions are via LO (resp.\ RO),
because the non-stuttering condition for automata (Def.~\ref{def:automaton})
ensures there is no unary transition on the side where control does not change.
If $n\neq m $ and $n'\neq m'$ the transitions only go via JO.
In terms of the fixed-control transition relation used 
in the definition (\ref{eq:VC}), applied to alignment automata control, we have that
$\tranSeg{(n,n'),(m,n')}$ only has transitions via LO,
$\tranSeg{(n,n'),(n,m')}$ only has transitions via RO,
and if $n\neq m $ and $n'\neq m'$ then
$\tranSeg{(n,n'),(m,m')}$ only has transitions via JO.

Fig.~\ref{fig:RVClo} gives the five VCs for transitions that go by the LO condition.  
The VCs for RO are symmetric (and omitted).
There are 25 combinations for JO transitions of an alignment automaton; some of their VCs are in Fig.~\ref{fig:RVCjo}.
Note that, by contrast with Fig.~\ref{fig:VC} 
and Lemma~\ref{lem:VCprog} we do not eliminate lift (hat) notation 
in Figs.~\ref{fig:RVClo} and~\ref{fig:RVCjo}.
That is because the alignment conditions $L$, $R$, $J$ are sets of states not sets of stores.
We return to this later, using the ``program counter'' variable $pc$ to encode the control part of the state.
For abbreviation we define 
\[ \graybox{$\breve{an}(n,n')$} = \hat{an}(n,n')\land [n|n'] \]
(Mnemonic: tipped c for control.)

Besides using  using formula symbols for intersections and set inclusion
(as we do throughout the paper),
there is another abuse of notation in Figs.~\ref{fig:RVClo} and~\ref{fig:RVCjo},
namely the use of $\leftF{e}$ and $\rightF{e}$ for sets of states as opposed to sets of stores.
Yet another abuse is substitution notation for alignment automaton states,
lifted from store relations to state relations.
(The interested reader can adapt the definitions from the ones for store relations.) 

\begin{lemma}[VCs for program alignment automata] %\label{lem:VCalign}
\upshape
Consider $\ok$ programs $c$ and $c'$, with alignment automaton
$\aprod(\aut(c,f),\aut(c',f'),L,R,J)$.
Let $an$ be an annotation.
For each pair $(n,n'),(m,m')$ of alignment automaton control points, the VC 
can be expressed 
as in Fig.~\ref{fig:RVCjo}, Fig.~\ref{fig:RVClo}, and in the omitted RO cases that are symmetric to Fig.~\ref{fig:RVClo}.
\end{lemma}
\begin{proof}
Consider the case of left-only $\lskipc{n}$ for $((n,n'),(m,n'))$ where $m = \fsuc(n,c,f)$.
The VC is
    \[ \POST(\biTranSeg{(n,n'),(m,n')})(\hat{an}(n,n')) \subseteq \hat{an}(m,n') \]
which is equivalent to (by defs)
\[ \all{i,i',s,s',j,j',t,t'}{
    \begin{array}[t]{l}
     ((i,i'),(s,s'))\in \hat{an}(n,n') \land 
     ((i,i'),(s,s'))\biTranSeg{(n,n'),(m,n')} ((j,j'),(t,t')) \\
      \imp ((j,j'),(t,t')) \in\hat{an}(m,n') 
     \end{array} 
} \]
which is equivalent (using def (\ref{eq:hatALT}) of hat) to
\[ \all{i,i',s,s',j,j',t,t'}{
    \begin{array}[t]{l}
     (s,s')\in an(n,n') \land 
     ((i,i'),(s,s'))\biTranSeg{(n,n'),(m,n')} ((j,j'),(t,t')) \\
      \imp (t,t') \in an(m,n') 
     \end{array} 
} \]
Because control does not change, the transition must go by condition LO as noted earlier, so
the preceding is equivalent to 
\[ \all{i,i',s,s',j,j',t}{
    \begin{array}[t]{l}
     (s,s')\in an(n,n') \land 
     ((i,i'),(s,s'))\biTranSeg{(n,n'),(m,n')} ((j,j'),(t,s'))  \\
     \land \: ((i,i'),(s,s'))\in L \imp (t,s') \in an(m,n') 
     \end{array} 
} \]
which is equivalent, by definition (\ref{eq:tranSeg}), to 
\[ \all{i,i',s,s',j,j',t}{
    \begin{array}[t]{l}
     (s,s')\in an(n,n') \land 
     ((i,i'),(s,s'))\biTrans ((j,j'),(t,s')) \land (i,i')=(n,n') \land (j,j')=(m,n')\\ 
     \land \: ((i,i'),(s,s'))\in L \imp (t,s') \in an(m,n') 
     \end{array} 
} \]
which is equivalent, by one-point rule of predicate calculus, to
\[ %begin{equation}\label{eq:RVCgeneral}
\all{i,i',s,s',t}{
    \begin{array}[t]{l}
     (s,s')\in an(n,n') \land 
     ((n,n'),(s,s'))\biTrans ((m,n'),(t,s')) \land (n,n')=(i,i') \\
     \land \: ((i,i'),(s,s'))\in L \imp (t,s') \in an(m,n') 
     \end{array} 
} 
\]
At this point we use that the command at $n$ is skip and $m$ is its following successor,
so the above is equivalent to 
\[ 
\all{i,i',s,s'}{
    \begin{array}[t]{l}
     (s,s')\in an(n,n') \land (n,n')=(i,i') 
     \land ((i,i'),(s,s'))\in L \imp (s,s') \in an(m,n') 
     \end{array} 
} \]
which is equivalent (using def (\ref{eq:hatALT}) of lift, and  def of the control predicate $[n|n']$) to 
\[ 
\all{i,i',s,s'}{
    \begin{array}[t]{l}
     ((i,i'),(s,s'))\in \hat{an}(n,n') \land ((i,i'),(s,s'))\in[n|n'] 
     \land ((i,i'),(s,s'))\in L \imp \\ 
     ((i,i'),(s,s')) \in \hat{an}(m,n') 
     \end{array} 
} \]
which we can write as $L\land \hat{an}(n,n')\land [n|n'] \imp \hat{an}(m,n')$,
and even more succinctly as 
$L \land \breve{an}(n,n') \imp \hat{an}(m,n')$ using the $\breve{an}$ abbreviation.

It is for the sake of this last step that we did not eagerly simplify using the one-point rule with $(i,i')=(n,n')$.   For the quantifier-free version to make sense, everything needs to be a predicate of the same type.

All the left-only cases are proved similarly, as are the right-only and joint cases.
\end{proof}

\begin{remark}\upshape
There are many 
$((n,n'),(m,n'))$ with no transitions, because there are no transitions along $(n,m)$ or $(n',m')$ in the underlying automata.  The purpose of alignment conditions $L,R,J$ is to make pairs $(n,n')$ unreachable in the product, so they can be annotated with false and so have trivial VCs.
Typically $L,R,J$ are chosen so joint steps are only taken when both sides 
are similar programs, so combinations like ``assignment on the left and if on the right''  ---the last four rows in Fig.~\ref{fig:RVCjo} and the ones omitted--- are less likely to be relevant.

In these terms, the automaton used in the proof of Cor.~\ref{cor:relIAMcomplete} 
is the worst possible, in that it has the maximal number of nontrivial verification conditions.
\qed\end{remark}

\subsection{Encoding relational VCs in terms of store relations}

The unary VCs are given in terms of store predicates.
For relational VCs, although the annotation comprises store relations,
the relational VCs involve alignment conditions ($L,R,J$) which are sets of 
alignment automaton states.  
So (as noted earlier) the relational VCs are given in terms of states, using $\hat{an}$ 
and $\breve{an}$ in Figs.~\ref{fig:RVClo} and~\ref{fig:RVCjo}.  

For use in RHL+ proofs about programs in automaton normal form,
we will encode the VCs in terms of relations on stores that use 
a variable $pc$ (for ``program counter'') to encode control information.  
In Sect.~\ref{sec:RVC} we give exact characterizations of the VCs.
Those ``primary VCs'' entail the encoded ones, and the latter suffice for our purposes.

We often use the following abbreviations for assignments and tests of 
the chosen $pc$ variable.
\begin{equation}\label{eq:pcabbrev}
\mbox{
\graybox{$\spc n$} for $pc:=n$
\qquad 
\graybox{$\tpc n$} for $pc=n$ 
\qquad  for literal $n\in\nat$
}
\end{equation}
To be precise, we define $\spc n$ to be $\lassg{0}{pc}{n}$, choosing $0$ arbitrarily as the label. Our uses of $\spc n$ will be in contexts where labels are irrelevant
in commands that are not required to be $\ok$.
Note that the notations $\spc n$ and $\tpc n$ depend, implicitly, on the choice of the variable $pc$.

\begin{definition}\label{def:pc:encode}
Let $R$ be a set of states of an alignment automaton, 
and suppose the stores of the underlying automata are variable stores.
Let $x$ be a variable.
We say $R$ is \dt{independent from $x$ on both sides} provided
$R=\some{x|x}{R}$, or equivalently 
\begin{equation}\label{eq:indep} ((n,n'),(s,s'))\in R \imp 
((n,n'),(\update{s}{x}{i},\update{s'}{x}{j})) \in R 
\quad\mbox{for all $i,j$}
\end{equation}
Let $pc$ be a variable such that $R$ is independent from $pc$ on both sides.
Define \graybox{$\encode{R}$} to be a relation on stores that
encodes $R$ using $pc$ as follows.
\[ \encode{R} \eqdef \{ (s,s') \mid  ((s(pc),s'(pc)),(s,s'))\in R  \} \]
We call $\encode{R}$ the \dt{$pc$-encoded $R$}, or just ``encoded $R$''. 
\end{definition}

\begin{lemma}\label{lem:state-to-store}
\upshape
Let $R$ be a relation on alignment automaton states
that is independent of $pc$ on both sides.
If $(s,s') \in \encode{R}\land\bothF{\tpc n\sep\tpc n'} $
then $((n,n'),(s,s')) \in R$.
\end{lemma}
\begin{proof}
Suppose $(s,s')\in  \encode{R}\land\bothF{\tpc n\sep\tpc n'}$.
Then $s(pc)=n$ and $s'(pc)=n'$
by definition of $\bothF{\tpc n\sep\tpc n'}$,
hence $((n,n'),(s,s'))\in R$ by definition of $\encode{R}$.
\end{proof}

\begin{figure*}
\begin{small}
\begin{tabular}{lll}
if $\sub(n,c)$ is \ldots     & 
and $m$ is\ldots &
then the encoded VC for $((n,n'),(m,n'))$ is $\ldots$ 
\\\hline
$\lskipc{n}$ & $\fsuc(n,c,f)$ & $\encode{L}\land\bothF{\tpc n\sep\tpc n'} \land an(n,n')\imp an(m,n')$ 
\\
$\lassg{n}{x}{e}$ & $\fsuc(n,c,f)$ & $\encode{L}\land\bothF{\tpc n\sep\tpc n'}\land an(n,n')\imp \subst{an(m,n')}{x|}{e|}$ 
\\
$\lifgc{n}{gcs}$ & $\lab(d)$ where $e\gcto d$ is in $gcs$ 
& $\encode{L}\land\bothF{\tpc n\sep\tpc n'}\land an(n,n')\land \leftF{e} \imp an(m,n')$ 
\\ 
$\ldogc{n}{gcs}$ & $\lab(d)$ where $e\gcto d$ is in $gcs$ 
& $\encode{L}\land\bothF{\tpc n\sep\tpc n'}\land an(n,n')\land \leftF{e} \imp an(m,n')$ 
\\ 
$\ldogc{n}{     gcs}$ & $\fsuc(n,c,f)$ 
& $\encode{L}\land\bothF{\tpc n\sep\tpc n'}\land an(n,n')\land \neg\leftF{\enab(gcs)} \imp an(m,n')$ 
\end{tabular}
\end{small}
\vspace*{-1ex}
\caption{The left-only $pc$-encoded VCs for annotation $an$ of  $\aprod(\aut(c,f),\aut(c',f'),L,R,J)$.}
\label{fig:RVClo-encoded}
\end{figure*}

\begin{figure*}
\begin{footnotesize}
\begin{tabular}{lll}
if $\begin{array}[t]{l} 
    \sub(n,c) \\
    \sub(n',c') 
    \end{array}$
    are\ldots     & 
and $\begin{array}[t]{l} m \\ m' \end{array}$ are\ldots &
then the encoded VC for $((n,n'),(m,m'))$ is $\ldots$ \\\hline

$\lskipc{n}$ &
$\fsuc(n,c,f)$ & $\encode{J}\land\bothF{\tpc n\sep\tpc n'} \land an(n,n')\imp an(m,m')$ 
\\
$\lskipc{n'}$ &
$\fsuc(n',c',f')$ & 
\\\hdashline

$\lassg{n}{x}{e}$ &
$\fsuc(n,c,f)$ & $\encode{J}\land\bothF{\tpc n\sep\tpc n'} \land an(n,n')\imp \subst{an(m,m')}{x|x'}{e|e'}$ 
\\
$\lassg{n'}{x'}{e'}$ &
$\fsuc(n',c',f')$ & 
\\\hdashline

$\lifgc{n}{gcs}$ & $\lab(d)$ where $e\gcto d$ is in $gcs$ 
& $\encode{J} \land\bothF{\tpc n\sep\tpc n'}\land an(n,n')\land \leftF{e} \land \rightF{e'} \imp an(m,m')$ 
\\
$\lifgc{n'}{gcs'}$ & $\lab(d')$ where $e'\gcto d'$ is in $gcs'$
& 

\end{tabular}

\end{footnotesize}
\vspace*{-1ex}
\caption{Selected $pc$-encoded joint VCs for annotation $an$ of $\aprod(\aut(c,f),\aut(c',f'),L,R,J)$.}
\label{fig:RVCjo-encoded}
\end{figure*}

\begin{lemma}\label{lem:liftRVC}
\upshape
Let $an$ be an annotation of an alignment automaton 
$\aprod(\aut(c,f),\aut(c',f'),L,R,J)$ for commands $c,c'$.
Suppose $pc$ is a fresh variable in the sense that 
it does not occur in $c$ or $c'$, and all of $L,R,J$ are independent from
$pc$ on both sides.
Then each of the left-only relational VCs of Fig.~\ref{fig:RVClo} 
implies the same condition with $\encode{L}\land\bothF{\tpc n\sep\tpc n'}$ in place of $L\land[n|n']$, and $\hat{an}$ changed to $an$,
as shown in Fig.~\ref{fig:RVClo-encoded}; \emph{mut.\ mut.} for $R$ in the right-only VCs.
The joint relational VCs of Fig.~\ref{fig:RVCjo} imply the same 
condition with 
$\encode{J}\land\bothF{\tpc n\sep\tpc n'}$ in place of $J\land[n|n']$,
as shown in Fig.~\ref{fig:RVCjo-encoded}.
\end{lemma}

\begin{proof}
Consider the left-only skip case.
The primary VC is $L \land \breve{an}(n,n')\imp \hat{an}(m,n')$
(first row of Fig.~\ref{fig:RVClo}),
which abbreviates 
\begin{equation}\label{eq:primVC}
 L \land [n|n'] \land \hat{an}(n,n')\imp \hat{an}(m,n') 
 \end{equation}
We show this implication entails 
the first line in Fig.~\ref{fig:RVClo-encoded}, i.e.,
\[ \encode{L}\land\bothF{\tpc n\sep\tpc n'} \land an(n,n')\imp an(m,n') \]
To prove the latter implication, observe for any $(s,s')$
\[\begin{array}{lll}
    & (s,s')\in \encode{L}\land (s,s')\in\bothF{\tpc n\sep\tpc n'} \land (s,s')\in an(n,n') \\
\imp 
    & ((n,n'),(s,s'))\in L \land (s,s')\in an(n,n') & \mbox{Lemma~\ref{lem:state-to-store}} \\
\iff
    & ((n,n'),(s,s'))\in L \land((n,n'),(s,s'))\in[n|n']\land (s,s')\in an(n,n')  & \mbox{def $[n|n']$} \\
\iff 
    & ((n,n'),(s,s'))\in L \land((n,n'),(s,s'))\in[n|n']\land ((n,n'),(s,s'))\in \hat{an}(n,n') & \mbox{def (\ref{eq:hatALT})} \\
\imp 
    & ((n,n'),(s,s'))\in \hat{an}(m,n') & \mbox{primary VC (\ref{eq:primVC}) } \\
\iff 
    & (s,s')\in an(m,n') & \mbox{def (\ref{eq:hatALT})} \\
    \end{array}
\]
The argument is essentially the same for all cases in Figs.~\ref{fig:RVClo}
and~\ref{fig:RVCjo}.
\end{proof}

\begin{remark}\upshape
Lemma~\ref{lem:state-to-store} is a strict implication,
because $((n,n'),(s,s')) \in R$ does not imply anything about the value of $pc$ in $s$ and $s'$.
However, 
%% If $((n,n'),(s,s')) \in R$ 
%% then so is $((n,n'),(\update{s}{pc}{i},\update{s'}{pc}{j}))$,
%% for any integers $i,j$, by $R$ independent from $pc$.
%% So in particular 
%% $((n,n'),(\update{s}{pc}{n},\update{s'}{pc}{n'}))$ is in $R$.
%% Hence 
%% $(\update{s}{pc}{n},\update{s'}{pc}{n'}))$ is in 
%% $\encode{R}\land\bothF{\tpc n\sep\tpc n'}$.
we can recover $R$ from $\tilde{R}$.
Given any store relation $\S$, 
define $f(\S) 
= \{ ((s(pc),s'(pc)),(\update{s}{pc}{n},\update{s'}{pc}{n'})) \mid
     (s,s')\in \S, n\in\nat, n'\in\nat \}$.
Then $f(\encode{R})=R$ provided $R$ is independent from $pc$.
Perhaps this can be used to show the encoded VCs are actually equivalent to
the primary VCs.  But we have no motivation to work out the details.
\qed\end{remark}

\section{Examples: Obtaining a logic proof from an IAM proof}\label{sec:examples}

This section works out examples that are instances of 
Theorems~\ref{thm:FloydComplete} and~\ref{thm:acomplete}.
This serves to expose the key ideas, but some readers may prefer to 
skip to Sect.~\ref{sec:KATnf}.

\subsection{Example $c0$ unary correctness}\label{sec:c0unary}

% TODO these paragraphs can be made more clear

For the running example command $c0$ in Sect.~\ref{sec:unary},
Example~\ref{ex:c0an} considers  a schematic annotation $an$ of the automaton $\aut(c0,6)$ for a spec $\spec{P}{Q}$.
Assuming the annotation is valid, the judgment is valid, so by completeness of HL there is a proof. 
In this section we sketch what can be seen as an alternate way to prove completeness of HL,
instantiated to this example.
We present the result for two reasons.
First, it makes a connection between automata-based verification methods and HL.
Second, and more importantly, it serves to illustrate several ingredients 
in the main result of the paper.  
These ingredients are also used in the relational examples of Sect.~\ref{sec:lockstep} and Sect.~\ref{sec:condEg}.

We will define a command $d0$ that is in ``automaton normal form'', 
roughly equivalent to the example $c0$,
and then adapt the annotation for reasoning about $d0$.

Let $pc$ be a fresh variable. Specifically, $pc$ must not occur in $c0$, 
and $P$, $Q$, and each of the given predicates $an(i)$ should be independent from $pc$.  
Define the following program obtained from $c0$ by adding assignments to $pc$ corresponding to labels.
\[ 
\begin{array}{l}
\graybox{$c0^+:$} \quad
 \spc 1; \lassg{1}{x}{y} ;
        \spc 2; \keyw{do}^2 x > 0 \gcto 
                            \spc 3; \!\!\!\begin{array}[t]{l}
                            \keyw{if}^3  x\mod 2 = 0 \gcto \spc 4;\lassg{4}{x}{x-1} \\
                            \gcsep x\mod 2 \neq 0 \gcto \spc 5;\lassg{5}{x}{x-2} 
                             ~ \keyw{fi} ~ ; \spc 2 ~ \keyw{od} 
                             \,;\; \spc 6
                             \end{array}
                             \end{array}
 \]
On termination, $c0^+$ sets $pc$ to $6$, just as the automaton $\aut(c,6)$ of the program does.
This is an instance of a general construction defined later: 
$c0^+ = \addPC(c0) ; \spc 6$.  

A key point is that $c0^+$ is equivalent to another command, $\spc 1;d0$,
that is expressed in a way that reflects the transitions of $\aut(c0,6)$.
As discussed later, we can prove that 
\begin{equation}\label{eq:c0d0}
\spc 1; d0 \kateq c0^+
\end{equation}
where $d0$ is defined by 
\[ \graybox{d0:}\quad \keyw{do} 
\begin{array}[t]{l}
\tpc 1 \gcto x:=y; \spc 2 \\
      \gcsep \tpc 2 \land x > 0 \gcto \spc 3 \\
      \gcsep \tpc 2 \land x \ngtr 0 \gcto \spc 6  \\  
      \gcsep \tpc 3 \land x\mod 2=0 \gcto \spc 4 \\ 
      \gcsep \tpc 3 \land x\mod 2\neq 0 \gcto \spc 5 \\
      \gcsep \tpc 4 \gcto \lassg{}{x}{x-1}; \spc 2 \\ 
      \gcsep \tpc 5 \gcto \lassg{}{x}{x-2}; \spc 2  
      \quad \keyw{od}
      \end{array}
\]
We call $\spc 1;d0$ the \dt{automaton normal form} of $c0^+$.
Later we introduce a relation between commands and normal forms, in terms of which we have $\norm{c0}{6}{gcs}$ where $gcs$ is the body of $d0$.

Note that $\enab(gcs)$ is 
\[ \tpc 1 \lor (\tpc 2 \land x > 0) \lor (\tpc 2 \land x \ngtr 0) \lor (\tpc 3 \land x\mod 2=0) \lor (\tpc 3 \land x\mod 2\neq 0) \lor \tpc 4 \lor \tpc 5 \]
This simplifies to $1 \leq pc < 6$ which can also be written
$\tpc 1 \lor \tpc 2 \lor \ldots \lor \tpc 5$.

In HL+ we prove $c0: \spec{P}{Q}$ as follows.
Define $I$, using the given annotation $an$, as follows:
\[ I: \quad 1 \leq pc \leq 6 \land (\land i : 0\leq i\leq 6 : pc=i \imp an(i)) \]
Later we will prove  $c: \spec{e\land I}{I}$ for each guarded command $e\gcto c$ in $d0$.
Then by rule \rn{Do}, and \rn{Conseq} to simplify the enabling condition, we get
\[ d0: \spec{I}{I \land \neg(1\leq pc < 6)} \]
Now $I \land \neg(1\leq pc < 6)$ implies $I\land pc=6$ which implies $Q$
(recall that we assume $an(6)= Q$).  
Moreover $pc=1\land P$ implies $I$ (using the assumption that $an(1)=P$).
So by the consequence rule, from the above we get 
\[ d0: \spec{pc=1\land P}{Q} \]
Using the assignment rule we get $pc:=1 : \spec{P}{pc=1\land P}$ (recall the assumption that $pc$ is fresh for $P$). 
So by the sequence rule we get 
\[ \spc 1; d0 : \spec{P}{Q} \]
Now using (\ref{eq:c0d0}) and 
%We will give axioms for  $\kateq$ which suffice to show that 
%\begin{equation}\label{eq:c0nf}  \spc 1; d0  \kateq c0^+ 
%\end{equation}
%So by 
the \rn{Rewrite} rule we have 
\[ c0^+ : \spec{P}{Q} \]
By the choice of $pc$ we have $\ghost(pc,c0^+)$
and $P,Q$ are independent from $pc$, 
so by rule \rn{Ghost} we have 
\[ \erase(pc,c0^+) : \spec{P}{Q} \]
Now $\erase(pc,c0^+)$ is similar to $c0$ but different: it looks like
\[ 
 \lskipc{}; \lassg{1}{x}{y} ;
        \lskipc{}; \keyw{do}^2 x > 0 \gcto 
                            \lskipc{}; \keyw{if}^3 \ldots ~ \keyw{fi} ~ \keyw{od};\lskipc{} 
\]
But we can show that
\[ \erase(pc,c0^+) \kateq c0 \]
(essentially because $\skipc;d\kateq d$ and $d;\skipc\kateq d$,
and the fact that $\kateq$ is a congruence).
So using rule \rn{Rewrite} a second time, we finally arrive at our goal:
\[ c0 : \spec{P}{Q} \]
To complete the above proof, we need to discharge the proof obligations for the use of the \rn{Do} rule.  That is, for each $e\gcto c$ in the body of $d0$ 
we need $c:\spec{I\land e}{I}$.  

For the case \graybox{$\tpc 1 \gcto x:=y;\spc 2$} in $d0$, the proof is as follows.
We have $an(1)\imp\subst{an(2)}{x}{y}$ by a VC in Eqn.~(\ref{eq:VCc0}).
Using rules \rn{Asgn} and \rn{Conseq} we can prove 
$x:=y : \spec{an(1)}{an(2)}$.
Hence by \rn{Conseq} we have
$ x:=y : \spec{\tpc 1 \land an(1)}{an(2)} $.
By rule \rn{Asgn} and independence of $an(2)$ from $pc$ we  have
$pc:=2 : \spec{an(2)}{\tpc 2 \land an(2)}$
hence by \rn{Seq} we have
$x:=y; \spc 2 : \spec{\tpc 1\land an(1)}{\tpc 2 \land an(2)}$
hence by \rn{Conseq} and definition of $I$ we get
$x:=y; \spc 2  : \spec{I\land \tpc 1}{I}$ as required.

For the case \graybox{$\tpc 2 \land x > 0 \gcto \spc 3$} the proof is as follows.
By assign rule and $pc$ not in $FV(an(3))$ and $pc$ different from $x$, we have
$\spc 3 : \spec{3=3 \land an(3) }{ \tpc 3 \land an(3) }$.
(Note that $\subst{(\tpc 3)}{pc}{3}$ is $3=3$.)
Hence by consequence rule using the VC $x>0 \land an(2) \imp an(3)$ 
in Eqn.~(\ref{eq:VCc0}), 
we get $\spc 3 : \spec{x>0 \land an(2)}{\tpc 3 \land an(3)}$.
Hence by consequence using the $\tpc 2 \land I \imp an(2)$  and  $\tpc 3 \land an(3) \imp I$
(from definition of $I$), we get the required 
$\spc 3 : \spec{I\land \tpc 2 \land x > 0}{I}$.

The remaining cases are:
\graybox{$\tpc 2 \land x \ngtr 0 \gcto \spc 6$},
%% ----------------------
%% ?2 /\ not(x > 0) -> !6 
%% ----------------------
%% By assign rule and pc not in FV(6) have
%% pc:=6 : pc=2 /\ not(x>0) /\ an(6) ~> pc=6 /\ an(6)
%% hence by conseq, using VC  not(x>0) /\ an(2) => an(6)  have 
%% pc:=6 : pc=2 /\ not(x>0) /\ an(2) ~> pc=6 /\ an(6)
%% hence by conseq and def I have 
%% pc:=6 : pc=2 /\ not(x>0) /\ I ~> I
\graybox{$\tpc 3 \land x\mod 2=0 \gcto \spc 4$},
%% -----------------
%% ?3 /\ x%2=0 -> !4
%% -----------------
%% By assign rule and pc not in FV(3) have 
%% pc:=4 : pc=3 /\ x%2=0 /\ an(4) ~> pc=4 /\ an(4) 
%% hence by conseq, using VC  x%2=0 /\ an(3) => an(4)  have
%% pc:=4 : pc=3 /\ x%2=0 /\ an(3) ~> pc=4 /\ an(4) 
%% hence by conseq and def I have 
%% pc:=4 : pc=3 /\ x%2=0 /\ I ~> I
\graybox{$\tpc 3 \land x\mod 2\neq 0 \gcto \spc 5$},
\graybox{$\tpc 4 \gcto x:=x-1; \spc 2$}, and
%% ----------------
%% ?4 -> x:=x-1; !2 
%% ----------------
%% By VC have
%% x:=x-1 : an(4) ~> an(2)  
%% hence by conseq have
%% (a) x:=x-1 : pc=4 /\ an(4) ~> an(2) 
%% By assign and pc not in FV(an(2)) have
%% pc:=2 : an(2) ~> pc=2 /\ an(2)
%% Hence by seq using (a) have 
%% x:=x-1; pc:=2 : pc=4 /\ an(4) ~> pc=2 /\ an(2) 
%% hence by conseq and def I have 
%% x:=x-1; pc:=2 : pc=4 /\ I ~> I
\graybox{$\tpc 5 \gcto x:=x-2; \spc 2$}.
They are all proved similarly.

In summary, we constructed a HL+ proof of $c0 : \spec{P}{Q}$ in a way which mimics the automaton based proof.  
The main points of the example are that 
(i) a program is equivalent to one that corresponds to a transition system representation of the program,  
(ii) for such a program, the annotation gives rise to a loop invariant 
which is a boolean combination of the annotation's assertions $an(i)$,
and (iii) the proof obligations for the commands in the automaton normal form body 
correspond closely to the VCs for the annotated automaton.

\begin{remark}\upshape 
One might think we can avoid the last \rn{Rewrite} step above, by defining $\erase$ to remove gratuitous skips.  But one must be careful to ensure that the resulting program is well formed.  For example, we have $\ghost(x,\lassg{}{x}{x})$ and
$\erase(x,\lassg{}{x}{x})$ is $\lskipc{}$; we do not want to eliminate this skip, nor
the one in $\erase(x,\lifgc{}{y>0\gcto \lassg{}{x}{1}\gcsep\ldots})$.
\qed
\end{remark}

\subsection{A lockstep alignment example}\label{sec:lockstep}

In this section we prove a relational judgment using an alignment that runs the automata in lockstep.  Then we show how the annotation can be used as basis for a RHL+ proof.

To streamline notation in the following, we use primed variables to indicate 
their value in the right state, so $x=x'$ abbreviates $\leftF{x}=\rightF{x}$.

We show $c0\sep c0:\rspec{y=y'}{x=x'}$ for the running example (from Sect.~\ref{sec:unary}), repeated here:
\[ \graybox{c0:} \quad
        \lassg{1}{x}{y} ;
        \keyw{do}^2 x > 0 \gcto 
                            \begin{array}[t]{l}
                            \keyw{if}^3 \: x\mod 2 = 0 \gcto \lassg{4}{x}{x-1} \\
                            \gcsep \: x\mod 2 \neq 0 \gcto \lassg{5}{x}{x-2} 
                             ~ \keyw{fi} ~ \keyw{od}
                             \end{array} 
\]
We choose $6$ as final label on both sides,
and use a lockstep alignment automaton,
with $L$ and $R$ false and for $J$ the condition
that the control is in $\{(i,i) \mid 1\leq i \leq 6\}$,
i.e., $J:= [1|1]\lor[2|2]\lor\ldots\lor[5|5]$.
We consider $\aprod(\aut(c0,6),aut(c0,6),\mathit{false},\mathit{false},J)$.
The automaton is manifestly $(y=y')$-adequate:
$J\lor [6|6]$ is preserved by $\biTrans$ for the alignment automaton.
Note that it would remain manifestly $true$-adequate if we included
$[6|6]$ as a disjunct of $J$,
but that would go against Assumption~1 which says we only consider live $(L,R,J)$.

Here is the annotation:
\[\begin{array}{lcl}
an(1,1) &=& (y=y') \\
an(i,i) &=& (x=x') \quad\mbox{for $i$ such that $2\leq i\leq 6$ }\\
an(i,j) &=& \mathit{false} \quad\mbox{for $i,j$ in $1..6$ with $i\neq j$}
\end{array}\] 
Here are the nontrivial VCs:
\begin{equation}\label{eq:VCrel}\begin{array}{l}
%\models x:=1\sep x:=1 : \rspec{an(1,1)}{an(2,2)} \\       
an(1,1)\land\bothF{\tpc 1}\imp \subst{an(2,2)}{x|x'}{y|y} \\       
an(2,2)\land\bothF{\tpc 2} \land x > 0 \land x' > 0 \imp an(3,3) \\ 
an(2,2)\land\bothF{\tpc 2} \land x\ngtr 0 \land x'\ngtr 0 \imp  an(6,6) \\
an(3,3)\land\bothF{\tpc 3} \land x\mod 2=0 \land x'\mod 2=0 \imp an(4,4) \\
an(3,3)\land\bothF{\tpc 3} \land x\mod 2\neq 0 \land x'\mod 2\neq 0 \imp an(5,5) \\
%\models x:=x-1\sep x:=x-1 : \rspec{an(4,4)}{an(2,2)}     \\ 
an(4,4)\land\bothF{\tpc 4} \imp \subst{an(2,2)}{x|x'}{x-1|x'-1}     \\ 
%\models x:=x-2\sep x:=x-2 : \rspec{an(5,5)}{an(2,2)}     
an(5,5)\land\bothF{\tpc 5}\imp \subst{an(2,2)}{x|x'}{x-2|x'-2}
\end{array}
\end{equation}
All the other VCs have false antecedent or precondition.
For example, one VC is 
$an(2,2)\land\bothF{\tpc 2} \land x > 0 \land x' < 0 \imp an(3,6) $
but $an(2,2)$ says $x=x'$ so the antecedent simplifies to false.
The annotation is valid.

To construct a proof of $c0\sep c0:\rspec{true}{x=x'}$
in RHL+ we use the normal form in (\ref{eq:c0d0}),
i.e., $\spc 1; d0$ where 
\[ d0: \quad \keyw{do} 
\begin{array}[t]{l}
      \tpc 1 \gcto x:=y; \spc 2 \\
      \gcsep \tpc 2 \land x > 0 \gcto \spc 3 \\
      \gcsep \tpc 2 \land x \ngtr 0 \gcto \spc 6  \\  
      \gcsep \tpc 3 \land x\mod 2=0 \gcto \spc 4 \\ 
      \gcsep \tpc 3 \land x\mod 2\neq 0 \gcto \spc 5 \\
      \gcsep \tpc 4 \gcto \lassg{}{x}{x-1}; \spc 2 \\ 
      \gcsep \tpc 5 \gcto \lassg{}{x}{x-2}; \spc 2  
      \quad \keyw{od}
      \end{array}
\]
We will use rule \rn{dDo} instantiated with $\Q,\Lrel,\R:=\S,\mathit{false},\mathit{false}$,
with the following as loop invariant.
\[\S: \quad
\begin{array}[t]{l}
  (pc=1=pc' \imp true) \land \\
  \quant{\land}{i}{2\leq i\leq 6}{pc=i=pc' \imp x=x'} \land \\
  1\leq pc \leq 6 \land pc=pc'
  \end{array}
\]
Note that $1\leq pc \leq 6 \land pc=pc'$ is equivalent to
$\encode{J}\lor \bothF{\tpc 6}$. 

Writing $gcs$ for the body of $d0$, 
the side condition of \rn{dDo} is thus 
\[ \S\imp (\leftF{\enab(gcs)} = \rightF{\enab(gcs)}) 
        \lor (\mathit{false} \land \leftF{\enab(gcs)})
        \lor (\mathit{false} \land \rightF{\enab(gcs)}) \]
which is valid
by $\leftF{pc}=\rightF{pc}$, i.e., $pc=pc'$, from $\S$.  
There are three sets of premises in \rn{dDo}.  The first set relates a guarded command on the left to $\skipc$ on the right, with a precondition including the left alignment condition which here is $\mathit{false}$, so these premises are immediate using rule \rn{rFalse}.
The second set of premises is similar, with $\skipc$ on the left, and again precondition $\mathit{false}$.

What remains is the third set of premises, which require us to relate each guarded command in $d0$, on the left, with each guarded command of $d0$ on the right, under precondition $\encode{J}$.
There are $7\times 7$ such premises, 
each of the form $d\sep d': \rspec{\leftF{e}\land\rightF{e'}\land\S}{\S}$ 
for $e\gcto d$ and $e'\gcto d'$ in $d0$.

Here is one case: for the guarded commands
$ \tpc 1 \gcto x:=y; \spc 2$ and 
$ \tpc 2 \land x > 0 \gcto \spc 3$ we must show 
\[ x:=y;\spc 2 \sep \spc 3 : \rspec{\leftF{\tpc 1}\land\rightF{\tpc 2\land x>0}\land \S}{\S} \]
The precondition has a contradiction:
$\leftF{\tpc 1}$ and $\rightF{\tpc 2}$ contradict the conjunct $pc=pc'$ of $\S$.
So the displayed judgment can be proved using \rn{rConseq} and \rn{rFalse}.

In fact all cases where control is at different places in the two programs, 
or different branches are being taken, can be proved using \rn{rFalse}.
Some of those cases contradict $pc=pc'$, as above.
The other cases contradict $x=x'$, for example the case
of $\tpc 2 \land x > 0 \gcto \spc 3$ versus 
$\tpc 2 \land x \ngtr 0 \gcto \spc 6$.

What remain are the cases where control is at the same point on both sides and the same branch (of the original program) is being taken.  We consider these in turn.

\begin{itemize}
\item $\tpc y \gcto x:=y; \spc 2$.  For this we must show 
\begin{equation}\label{eq:A}
 x:=y; \spc 2 \mid x:=y; \spc 2 : \rspec{\leftF{\tpc 1}\land\rightF{\tpc 1}\land\S}{\S} 
\end{equation}
By a VC in (\ref{eq:VCrel}) we have 
$an(1,1)\land\bothF{\tpc 1}\imp\subst{an(2,2)}{x|x'}{y|y}$
so using \rn{dAsgn} and then \rn{rConseq} we get 
$x:=y \sep x':=y : \rspec{an(1,1)\land\bothF{\tpc 1}}{an(2,2)}$.
By \rn{dAsgn} and $pc$ not free in $an(2,2)$ we have 
\[ \begin{array}{l}
    \spc 2 \sep \spc 2 : \rspec{\leftF{2=2}\land\rightF{2=2}\land  an(2,2)  }{\leftF{\tpc 2}\land\rightF{\tpc 2}\land an(2,2)}
    \end{array}
 \]
hence by \rn{rConseq} we have
\[ \spc 2 \sep \spc 2 : \rspec{an(2,2)}{\leftF{\tpc 2}\land\rightF{\tpc 2}\land an(2,2)} \]
So by \rn{dSeq} we have 
\[ \begin{array}{l}
   x:=y; \spc 2 \sep x:=y; \spc 2 : \rspec{an(1,1)\land\bothF{\tpc 1}}{ \leftF{\tpc 2}\land\rightF{\tpc 2}\land an(2,2)} 
\end{array}
\] 
so by \rn{rConseq} and definition of $\S$ we have (\ref{eq:A}).

\item $\tpc 2 \land x > 0 \gcto \spc 3$.
By \rn{dAsgn}, using that $pc$ is distinct from $x$ and is not free in $an(3,3)$, we have
\[
\begin{array}{l}
 \spc 3 \sep \spc 3 : \rspec{\bothF{3=3 } \land an(3,3) }{ \bothF{\tpc 3} \land an(3,3)} 
 \end{array}
\]
so by \rn{rConseq} we have 
\[\begin{array}{l}
 \spc 3 \sep \spc 3 : \rspec{an(3,3) }{ \bothF{\tpc 3}\land an(3,3)} 
 \end{array}\]
Using the VC $an(2,2)\land\bothF{\tpc 2} \land \bothF{x>0}\imp an(3,3)$
with \rn{rConseq} we get 
\[\begin{array}{l}
 \spc 3 \sep \spc 3 : \rspec{\bothF{x > 0}\land an(2,2)\land\bothF{\tpc 2} }{ \bothF{\tpc 3}\land an(3,3)} 
\end{array}\]
From the definition of $\S$ we have
$\bothF{\tpc 2}\land S \imp an(2,2)$ and 
$\bothF{\tpc 3}\land an(3,3) \imp \S$
so by \rn{rConseq} we get what is required:
\[ \spc 3 \sep \spc 3 : \rspec{\bothF{\tpc 2 \land x > 0}\land \S }{\S}
\]

\item The remaining five cases are proved similarly; the cases are 
$\tpc 2 \land x \ngtr 0 \gcto \spc 6$,
$\tpc 3 \land x\mod 2=0 \gcto \spc 4$,
$\tpc 3 \land x\mod 2\neq 0 \gcto \spc 5$,
$\tpc 4 \gcto \lassg{}{x}{x-1}; \spc 2 $, and 
$\tpc 5 \gcto \lassg{}{x}{x-2}; \spc 2  $
\end{itemize}
Having established the premises and side condition of rule \rn{dDo}, we infer its conclusion:
\[ d0\sep d0: \rspec{\S}{\S\land\neg\leftF{\enab(gcs)}\land\neg\rightF{\enab(gcs)} } \]
As in the proof for unary correctness of $d0$, 
we have that $\enab(gcs)$ is equivalent to $1\leq pc <6$,
so in conjunction with $\S$ we can apply \rn{rConseq} to get 
\[ d0\sep d0: \rspec{\S}{\S\land \bothF{\tpc 6}} \]
Another use of \rn{rConseq} and the definition of $\S$ yields
\[ d0\sep d0: \rspec{\S}{x=x'} \]   
Yet another use of \rn{rConseq} and definition of $\S$ yields 
\[ d0\sep d0: \rspec{\bothF{\tpc 1}}{x=x'} \]   
By \rn{dAsgn} and \rn{rConseq} we have
$\spc 1 \sep \spc 1 : \rspec{true}{\leftF{\tpc 1}\land\rightF{\tpc 1}}$.
Hence by \rn{dSeq} we have 
\[ \spc 1; d0 \Sep \spc 1; d0 :  \rspec{true}{x=x'} \]
Now, using the equivalence (\ref{eq:c0d0}) and rule \rn{rRewrite} we get 
\[ c0^+ \sep c0^+ :  \rspec{true}{x=x'} \]
Next, $pc$ has the ghost property so by \rn{rGhost} we get
\[ \erase(pc,c0^+) \sep \erase(pc,c0^+) :  \rspec{true}{x=x'} \]
As noted in the unary proof (Sect.~\ref{sec:c0unary}) we have $ \erase(pc,c0^+) \kateq c0$,
so by \rn{rRewrite} we get 
\[ c0 \sep c0 :  \rspec{true}{x=x'} \]

\subsection{A conditionally aligned example}\label{sec:condEg}

\begin{figure*}
\begin{footnotesize}
\( 
\begin{array}{ll}
%% c4: & 
%% \lassg{1}{y}{x}; \lassg{2}{z}{24}; \lassg{3}{w}{0}; 
%%    \lwhilec{4}{y\neq 4}{
%%       \lifc{5}{w \mod 2 = 0}{
%%         \lassg{6}{z}{z*y}; \lassg{7}{y}{y-1} }{\lskipc{8}}; \lassg{9}{w}{w+1} }
%% \\[.3ex]
%% c5: & 
%% \lassg{1}{y}{x}; \lassg{2}{z}{16}; \lassg{3}{w}{0}; 
%%     \lwhilec{4}{y\neq 4}{ 
%%           \lifc{5}{w \mod 3 = 0}{
%%            \lassg{6}{z}{z*2}; \lassg{7}{y}{y-1}}{\lskipc{8}}; \lassg{9}{w}{w+1} }
%% \\[1ex]
c4: & 
\lassg{1}{y}{x}; \lassg{2}{z}{24}; \lassg{3}{w}{0}; 
\ldogc{4}{
   y\neq 4 \gcto
      \lifgc{5}{ w \mod 2 = 0 \gcto \lassg{6}{z}{z*y}; \lassg{7}{y}{y-1} 
           \gcsep w\mod 2 \neq 0 \gcto \lskipc{8}}
      ; \lassg{9}{w}{w+1} }
\\[.3ex]
c5: & 
\lassg{1}{y}{x}; \lassg{2}{z}{16}; \lassg{3}{w}{0}; 
    \ldogc{4}{y\neq 4 \gcto 
          \lifgc{5}{w \mod 3 = 0 \gcto \lassg{6}{z}{z*2}; \lassg{7}{y}{y-1}
              \gcsep w \mod 3 \neq 0 \gcto \lskipc{8}}
    ; \lassg{9}{w}{w+1} }
\end{array}
\)
\end{footnotesize}
\caption{Commands c4 and c5 adapted from~\cite{NagasamudramN21}}
\label{fig:c4c5}
\end{figure*}

We consider two example programs used in prior work to illustrate the conditionally aligned loop rule (our \rn{dDo}).
We use them to illustrate a non-lockstep alignment automaton
used in an IAM proof and then in the corresponding RHL+ proof.
The two commands are in Fig.~\ref{fig:c4c5}.

Both loops maintain $y\geq 4$ as invariant.
The loop in $c4$ maintains the invariants  $x! \cdot 4! = z \cdot y!$ and 
the loop in $c5$ maintains $2^x\cdot 2^4 = z \cdot 2^y$, owing to the respective 
initializations of $z=4!=24$ and $z=2^4=16$.
(Postfix $!$ means factorial as usual, and is not to be confused with the prefix abbreviation $\spc n$ for $\lassg{}{pc}{n}$.)
The reader may enjoy to use the proof rules to show 
\[ c4 \sep c5 : \rspec{x = x'\land x > 3}{z>z'} \]
using loop alignment conditions $\Lrel := \leftF{ w\mod 2\neq 0 }$ and $\R := \rightF{  w'\mod 3\neq 0 }$
and rule \rn{dDo}. 
But we give an automaton and then verify the program in automaton normal form.
The alignment enables use of annotations with simple predicates, avoiding any reference 
to the factorial or exponential functions.

\begin{remark}\upshape\label{rem:moduli}
Here is an interesting variation on the example.  Replace the moduli 2 (in $c4$) and 3 (in $c5$) by  variables $u$ and $v$ respectively, that are not modified by the program.
Conjoin to the precondition that $1 < u$ and $1 < v'$.
The proofs in this section go through essentially unchanged.
But the variation goes beyond what can be handled by Churchill et al~\cite{ChurchillP0A19}
and similar techniques that are restricted to linear arithmetic in alignment conditions
and in relational invariants.
\qed\end{remark}

We prove 
\begin{equation}\label{eq:specc4c5}
\aut(c4,0) \sep \aut(c5,0) : \rspec{x = x'\land x > 3}{z>z'} 
\end{equation}
using $0$ as final label.

Define $L,R,J$ as follows; recall that 
$[i|j]$ is  set of automaton states where control is at $(i,j)$.
\[ \begin{array}{l}
J:\quad [1|1]\lor[2|2]\lor[3|3]\lor([4|4]\land\leftF{ w\mod 2 = 0}\land \rightF{ w'\mod 3 = 0})\lor  
[5|5]\lor[6|6]\lor[7|7]\lor[9|9]
\\
L: \quad
([4|4]\land\leftF{ w\mod 2\neq 0})
\lor [5|4]
\lor [8|4]
\lor [9|4]
\\
R: \quad
([4|4]\land\rightF{w'\mod 3\neq 0})
\lor [4|5]
\lor [4|8]
\lor [4|9]
\end{array}
\]
Note that $(L,R,J)$ is live,
because assignments are always enabled and at $(5,5)$ the conditional is total.
It is straightforward but tedious to check that the following
is an inductive $(x=x')$-invariant: $J\lor L\lor R\lor [0|0]$.
So the alignment automaton 
$\aprod(\aut(c4,0),\aut(c5,0),L,R,J)$ is manifestly $(x=x')$-adequate.  

\begin{figure}
\(  \begin{array}{ll}
(n,m) & an(n,m) \\\hline
(1,1) & x=x'\land x>3 \qquad\mbox{precondition}\\
(2,2)  & y=y' \land y>3 \\
(3,3)  & \S \\ %y=y' \land y>3 \land z>z'>0 \\
(4,4)  & \S \\
(5,5)  & \S \land y>4 \land w\mod 2 = 0 = w'\mod 3 \\
(6,6)  & \S \land y>4 \land w\mod 2 = 0 = w'\mod 3 \\ 
(7,7)  & \S \land y>4 \land w\mod 2 = 0 = w'\mod 3 \\ 
(9,9)  & \S \\ 
(5,4)  & \S \land w\mod 2 \neq 0 \\
(8,4)  & \S   \\
(9,4)  & \S   \\
(4,5)  & \S \land w'\mod 3\neq 0  \\
(4,8)  & \S    \\
(4,9)  & \S   \\
(0,0) & z>z' \qquad\mbox{postcondition} \\
(\_,\_) & \mathit{false} \quad\mbox{in all other cases}
\end{array}
\)
\caption{Annotation for $c4,c5$ example}\label{fig:c4c5an}
\end{figure}

Using the predicate
\[ \S: \qquad y=y' \land y>3 \land z>z'>0 \]
we define annotation $an$ in Fig.~\ref{fig:c4c5an}.

The automaton has $10^2$ control points $(i,j)$, and thus $10^4$ pairs of control points
for which VCs must hold.
Most are infeasible of course.
For example, consider the VC for $((5,5),(8,8))$.
The automaton only transitions from $(5,5)$ to $(8,8)$ when 
$w\mod 2\neq 0$ and $w'\mod 3\neq 0$ but this contradicts $an(5,5)$ so the VC holds vacuously.
As another example, the automaton can transition from $(5,4)$ to $(6,4)$ when $w\mod 2 = 0$,
but this contradicts $an(5,4)$.

Fig.~\ref{fig:VC45} shows the VCs where the automaton has transitions from states that satisfy the annotation.  
They all follow by propositional reasoning and a bit of linear inequality reasoning.
So we have an IAM proof the programs satisfy their spec (\ref{eq:specc4c5})
in accord with Corollary~\ref{cor:relIAM}. 

\begin{figure*}
\begin{small}
\(\begin{array}{lll}
((n,n'),(m,m'))\hspace*{-3em} & \qquad \mbox{VC for } ((n,n'),(m,m')) \\\hline

((1,1),(2,2)) &
J\land[1|1]\land x=x'\land x>3 \imp \subst{(y=y'\land y>3)}{y|y'}{x|x'}
\\
((2,2),(3,3)) &
J\land[2|2]\land y=y' \land y>3 \imp \subst{\S}{z|z'}{24|16}
\\
((3,3),(4,4)) &
J\land[3|3]\land \S \imp \subst{\S}{w|w'}{0|0}
\\
((4,4),(5,5)) &
J\land[4|4]\land \S\land y\neq 4 \land y'\neq 4 \imp \S\land y>4\land w\mod 2=0=w'\mod 3 
& \mbox{\hspace*{-8em}(proved using $J$)}
\\
((4,4),(0,0)) &
J\land[4|4]\land \S\imp z>z'
\\
((5,5),(6,6)) &
J\land[5|5]\land w\mod 2=0=w'\mod 3 \land\S\land y>4 \imp \S\land y>4\land w\mod 2=0=w'\mod 3 
\\
((6,6),(7,7)) &
J\land[6|6]\land w\mod 2=0=w'\mod 3 \land\S\land y>4 \imp 
    \subst{(\S\land y>4\land w\mod 2=0=w'\mod 3)}{z|z'}{z*y|z'*2}
\\
((7,7),(9,9)) &
J\land[7|7]\land w\mod 2=0=w'\mod 3 \land\S\land y>4 \imp 
    \subst{\S}{y|y'}{y-1|y'-1}
\\
((9,9),(4,4)) &
J\land[9|9]\land\S \imp \subst{\S}{w|w'}{w+1|w'+1}
\\
((4,4),(5,4)) &
L\land[4|4]\land \S\land y\neq 4 \imp \S \land w\mod 2 \neq 0
\\
((5,4),(8,4)) &
L\land[5|4]\land \S\land w\mod 2\neq 0 \imp \S
\\
((8,4),(9,4)) &
L\land[8|4]\land \S\imp \S
\\
((9,4),(4,4)) &
L\land[9|4]\land \S\imp \S
\\
((4,4),(5,4)) &
R\land[4|4]\land \S\land y'\neq 4 \imp \S \land w'\mod 3 \neq 0
\\
((4,5),(4,8)) &
R\land[4|5]\land \S\land w'\mod 3\neq 0 \imp \S
\\
((4,8),(4,9)) &
R\land[4|8]\land \S\imp \S
\\
((4,9),(4,4)) &
R\land[4|9]\land \S\imp \S
\end{array}
\)
\end{small}
\vspace*{-1ex}
\caption{The nontrivial VCs for 
$\aprod(\aut(c4,0),\aut(c5,0,L,R,J)$}
\label{fig:VC45}
\end{figure*}

For a deductive proof in RHL+, 
let $pc$ be a fresh variable, i.e., one distinct from the variables of $c4,c5$ since these are the only variables that occur in the annotation or in the alignment conditions.
Fig.~\ref{fig:Jpc} spells out the $pc$-encoded store relations $\encode{L},\encode{R},\encode{J}$, that express the alignment conditions $L,R,J$ as predicates on $pc$
according to Def.~\ref{def:pc:encode}.

\begin{figure*}
\[ \begin{array}{l}
\encode{J}:
\quad 
\bothF{\tpc 1}\lor\bothF{\tpc 2}\lor\bothF{\tpc 3}\lor(\bothF{\tpc 4}\land\leftF{ w\mod 2 = 0}\land \rightF{ w'\mod 3 = 0})\lor \bothF{\tpc 5}\lor\bothF{\tpc 6}\lor\bothF{\tpc 7}\lor\bothF{\tpc 9}
\\
\encode{L}: 
\quad
(\bothF{\tpc 4}\land\leftF{ w\mod 2\neq 0})
\lor (\bothF{\tpc 5\sep\tpc 4})
\lor (\bothF{\tpc 8\sep\tpc 4}) 
\lor (\bothF{\tpc 9\sep\tpc 4})
\\
\encode{R}: 
\quad
(\bothF{\tpc 4}\land\rightF{ w\mod 3\neq 0})
\lor (\bothF{\tpc 4\sep\tpc 5})
\lor (\bothF{\tpc 4\sep\tpc 8}) 
\lor (\bothF{\tpc 4\sep\tpc 9})
\end{array}
\]
\vspace*{-1ex}
\caption{Store relations $\encode{L},\encode{R},\encode{J}$ for $c4,c5$}\label{fig:Jpc}
\end{figure*}

We aim to prove $\rspec{x=x'\land x>3}{z>z'}$ for the normal forms of $c4$ and $c5$,
instantiating rule \rn{dDo} 
with $\Lrel:=\encode{L}$, $\R:=\encode{R}$, and invariant
$\Q$ defined by
\[ \Q: \quad \S\land(\encode{L}\lor\encode{R}\lor\encode{J}\lor\bothF{\tpc 0}) \]
In the unary example of Sect.~\ref{sec:c0unary},
the invariant $I$ includes a conjunct that constrains the range of $pc$.
In $\Q$ the range on both sides is constrained by 
$\encode{L}\lor\encode{R}\lor\encode{J}\lor\bothF{\tpc 0}$.

The normal form of $c4$ is $\spc 1; d4$ where
\[ \graybox{d4:}\quad \keyw{do} 
\begin{array}[t]{l}
       \tpc 1 \gcto y:=x; \spc 2 \\
\gcsep \tpc 2 \gcto z:=24; \spc 3 \\
\gcsep \tpc 3 \gcto w:=0; \spc 4 \\
\gcsep \tpc 4 \land y\neq 0 \gcto \spc 5 \\
\gcsep \tpc 4 \land y=0 \gcto \spc 0 \\
\gcsep \tpc 5 \land w\mod 2 = 0 \gcto \spc 6 \\
\gcsep \tpc 5 \land w\mod 2 \neq 0 \gcto \spc 8 \\
\gcsep \tpc 6 \gcto z:=z*y; \spc 7 \\
\gcsep \tpc 7 \gcto y:=y-1; \spc 9 \\
\gcsep \tpc 8 \gcto \spc 9 \\
\gcsep \tpc 9 \gcto w:=w+1; \spc 4
\quad \keyw{od}
\end{array}
\]
The normal form of $c5$ is $\spc 1; d5$ where
\[ \graybox{d5:}\quad \keyw{do} 
\begin{array}[t]{l}
       \tpc 1 \gcto y:=x; \spc 2 \\
\gcsep \tpc 2 \gcto z:=16; \spc 3 \\
\gcsep \tpc 3 \gcto w:=0; \spc 4 \\
\gcsep \tpc 4 \land y\neq 0 \gcto \spc 5 \\
\gcsep \tpc 4 \land y=0 \gcto \spc 0 \\
\gcsep \tpc 5 \land w\mod 3 = 0 \gcto \spc 6 \\
\gcsep \tpc 5 \land w\mod 3 \neq 0 \gcto \spc 8 \\
\gcsep \tpc 6 \gcto z:=z*2; \spc 7 \\
\gcsep \tpc 7 \gcto y:=y-1; \spc 9 \\
\gcsep \tpc 8 \gcto \spc 9 \\
\gcsep \tpc 9 \gcto w:=w+1; \spc 4
\quad \keyw{od}
\end{array}
\]
As an abbreviation, we write $\enab(d4)$ for $\enab$ applied to the body of $d4$,
and likewise for $\enab(d5)$.
The conditions $\enab(d4)$ and $\enab(d5)$ are both equivalent to $1\leq pc\leq 9$.

We aim to instantiate rule \rn{dDo} by 
$\Q:=\Q$, $\Lrel:=\encode{L}$, and $\R:=\encode{R}$.
The side condition of \rn{dDo} is
\begin{equation}\label{eq:side-c4c5}
 \Q \imp (\leftF{\enab(d4)} = \rightF{\enab(d5)})
         \lor (\encode{L}\land\leftF{\enab(d4)})
         \lor (\encode{R}\land\rightF{\enab(d5)}) 
\end{equation}
To prove it, distribute in $\Q$ so the antecedent of the side condition has the form
\[ (\S\land\encode{L})\lor(\S\land\encode{R})\lor(\S\land\encode{J})\lor(\S\land\bothF{\tpc 0}) \]
We show each disjunct implies the consequent of the side condition (\ref{eq:side-c4c5}).
\begin{itemize}
\item $\S\land\encode{L}$ implies $\encode{L}\land\leftF{\enab(d4)}$ using the definition of $\encode{L}$.
\item $\S\land\encode{R}$ implies $\encode{R}\land\rightF{\enab(d5)}$ using the definition of $\encode{R}$.
\item $\S\land\encode{J}$ implies both $\leftF{\enab(d4)}$ and $\rightF{\enab(d5)}$ are true
so we have $\leftF{\enab(d4)} = \rightF{\enab(d5)}$.
\item $\S\land\bothF{\tpc 0}$ implies both $\leftF{\enab(d4)}$ and $\rightF{\enab(d5)}$ are false, 
so we have $\leftF{\enab(d4)} = \rightF{\enab(d5)}$.
\end{itemize}

As for the premises of \rn{dDo}, note that the normal form bodies in $d4$ and $d5$ each have 11 cases.  So for \rn{dDo} there are 11 left-only premises, 11 right-only, and 121 joint premises.  

%% \dn{Stray note: the conditions $\encode{L},\encode{R},\encode{J}$ are almost disjoint; the only overlap
%% is between $\encode{L}$ and $\encode{R}$ for states that satisfy $\bothF{\tpc 4}\land\leftF{w\mod 2\neq 0}\land\rightF{w'\mod 3\neq 0}$.
%% }

Let's do the left-only premises first --- selected cases of guarded commands in $d4$.
\begin{itemize}
\item for $\tpc 1\gcto y:=x;\spc 2$  we must show
\[ y:=x;\spc 2\sep\skipc : \rspec{\Q\land \leftF{\tpc 1}\land\encode{L}}{\Q} \]
Since $\leftF{\tpc 1}$ contradicts $\encode{L}$, we get the judgment by \rn{rConseq} and \rn{rFalse}.

\item for $\tpc 4\land y\neq 0\gcto \spc 5$ we must show
\[ \spc 5\sep\skipc : \rspec{\Q\land \leftF{\tpc 4\land y\neq 0}\land\encode{L}}{\Q} \]
By \rn{AsgnSkip} we have 
\( \spc 5\sep\skipc : \rspec{\subst{\Q}{pc|}{5|}}{\Q} \),
and we complete the proof using \rn{rConseq}, because 
we can prove
$\Q\land \leftF{\tpc 4\land y\neq 0}\land\encode{L} \imp \subst{\Q}{pc|}{5|}$.
To do so, note that $\subst{\Q}{pc|}{5|}$ is equivalent to 
\[ (\S\land \subst{\encode{L}}{pc|}{5|})\lor
 (\S\land \subst{\encode{R}}{pc|}{5|})\lor
 (\S\land \subst{\encode{J}}{pc|}{5|})\lor
 (\S\land \subst{\bothF{\tpc 0}}{pc|}{5|}) \]
We show $\S\land \subst{\encode{L}}{pc|}{5|}$ follows from the antecedent
$\Q\land \leftF{\tpc 4\land y\neq 0}\land\encode{L}$.
For $\S$, it is immediate from $\Q$.  
Now $\subst{\encode{L}}{pc|}{5|}$ is
\[ 
(\leftF{5=4}\land\rightF{\tpc 4}\land\leftF{ w\mod 2\neq 0})
\lor (\leftF{5=5}\land\rightF{\tpc 4})
\lor (\leftF{5=8}\land\rightF{\tpc 4}) 
\lor (\leftF{5=9}\land\rightF{\tpc 4})
\]
And we have that $\leftF{5=5}\land\rightF{\tpc 4}$ follows from 
$\Q\land \leftF{\tpc 4}\land\encode{L}$,
because $\encode{L}$ and $\leftF{\tpc 4}$ imply $\rightF{\tpc 4}$.

\end{itemize}

Turning to the joint premises,
note that the antecedents  include $\lnot\Lrel\land\lnot\R$,
which together with $\encode{L}\lor\encode{R}\lor\encode{J}\lor\bothF{\tpc 0}$ from the invariant 
implies $\encode{J}\lor\bothF{\tpc 0}$.  

For each $e\gcto c$ in $d4$ and $e'\gcto c'$ in $d5$,
the premise to be established is 
\[ c|c':\rspec{ \Q\land\leftF{e}\land\rightF{e'}\land\lnot\encode{L}\land\lnot\encode{R} }
              { \Q }
\]
which by \rn{rConseq} follows from
\[ c|c':\rspec{ \Q\land\leftF{e}\land\rightF{e'}\land(\encode{J}\lor\bothF{\tpc 0}) }
              { \Q }
\]
which follows by \rn{rDisj} from 
\[ c|c':\rspec{ \Q\land\leftF{e}\land\rightF{e'}\land\bothF{\tpc 0} }
              { \Q }
\]
and 
\begin{equation}\label{eq:c4c5joint}
 c|c':\rspec{ \Q\land\leftF{e}\land\rightF{e'}\land\encode{J} }
              { \Q }
\end{equation}
The first of these then follows using \rn{rConseq} and \rn{rFalse} because the guards $e$ in $d4$, and $e'$ in $d5$, all contradict $\tpc 0$.
This leaves us with the 121 cases of the form (\ref{eq:c4c5joint}).
Many of those cases go by \rn{rConseq} and \rn{rFalse}, for example
the case for $\tpc 1 \gcto y:=x; \spc 2$ and $\tpc 2 \gcto z:=16; \spc 3$
has precondition $\leftF{\tpc 1}\land\rightF{\tpc 2}$ for the guards, which contradicts $\encode{J}$.

The remaining cases correspond to the $J$-cases of the VCs in Fig.~\ref{fig:VC45}.
We consider one case.
\begin{itemize}\item 
$\tpc 6 \gcto \lassg{}{z}{z*y};\spc 7
\Sep
\tpc 6 \gcto \lassg{}{z}{z*2};\spc 7$.
We must show 
\[ \lassg{}{z}{z*y};\spc 7 \Sep \lassg{}{z}{z*2};\spc 7 
: \; \rspec{\Q\land\encode{J}\land\bothF{6}}{\Q} \]
We get $\spc 7\sep\spc 7: \rspec{an(7,7)}{an(7,7)\land\bothF{\tpc 7}}$
using \rn{rAsgn} and \rn{rConseq}, noting that $an(7,7)$ is independent from $pc$.
By definition of $\Q$ we have 
$an(7,7)\land\bothF{\tpc 7} \imp \Q$ so by \rn{rConseq} we get 
$\spc 7\sep\spc 7: \rspec{an(7,7)}{\Q}$.
By \rn{rAsgn} we get 
$\lassg{}{z}{z*y} \sep \lassg{}{z}{z*2}
: \rspec{\subst{an(7,7)}{z|z}{z*y|z*2}}{an(7,7)}$.
By a $pc$-encoded VC we have 
$\encode{J}\land\bothF{\tpc 6}\land an(6,6) \imp \subst{an(7,7)}{z|z}{z*y|z*2}$.
So by \rn{rConseq} we get 
$\lassg{}{z}{z*y} \sep \lassg{}{z}{z*2}
: \rspec{\encode{J}\land\bothF{\tpc 6}\land an(6,6)}{an(7,7)}$.
So by \rn{dSeq} we get
$\lassg{}{z}{z*y};\spc 7 \Sep \lassg{}{z}{z*2};\spc 7 
: \; \rspec{\Q\land\encode{J}\land\bothF{6}}{\Q}$.
\end{itemize}
Once all the premises for \rn{dDo} are proved, we get its conclusion: 
$d4\sep d5: \rspec{\Q}{\Q\land \neg\leftF{\enab(gcs_4)}\land\neg\rightF{\enab(gcs_5)}}$
where $gcs_4$ and $gcs_5$ are the bodies of $d4$ and $d5$. 

By steps like those in Sect.~\ref{sec:lockstep},
we can proceed to get 
$\spc 1;d4 \sep \spc 1;d5 : \rspec{x = x'\land x > 3}{z>z'}$
and continue using using \rn{rRewrite} and \rn{rGhost} to conclude
the proof of 
$ c4 \sep c5 : \rspec{x = x'\land x > 3}{z>z'}$.

\section{Automaton normal form reduction and KAT}\label{sec:KATnf}

Sect.~\ref{sec:embed} defines the representation of commands by expressions in 
Kleene algebra with tests (KAT), and defines command equivalence in terms of KAT equality under certain hypotheses. 
Sect.~\ref{sec:nf} defines the automaton normal form and the reduction of 
a command to its normal form.
Sect.~\ref{sec:KATequiv} presents the particular hypotheses for command equivalence 
that are needed for the normal form theorem.
Sect.~\ref{sec:nfthm} proves that any command, once instrumented with assignments 
to a fresh $pc$ variable, is equivalent to its normal form.
Sect.~\ref{sec:alt} discusses design alternatives for the formalization of command equivalence.  

\subsection{Embedding programs in a KAT}\label{sec:embed}

\newcommand\kK{\mathbb{K}}
\newcommand\kB{\mathbb{B}}
\newcommand\kdot{\mathbin{;}}
\newcommand\kstar{*}
\newcommand\kplus{+}
\newcommand\kneg{\neg}
\newcommand\kNeg[1]{\overline{#1}}
\newcommand\kone{1}
\newcommand\kzero{0}

\begin{definition}\label{def:KAT}
A \dt{Kleene algebra with tests}~\cite{Kozen97} (\dt{KAT}) is a structure $(\kK,\kB,\kplus,\kdot,\kstar,\kneg,\kone,\kzero)$
such that
\begin{itemize}
\item  $\kK$ is a set and $\kB\subseteq \kK$ (elements of $\kB$ are called \dt{tests}).
\item  $\kB$ contains $\kone$ and $\kzero$, and is closed under the operations $\kplus,\kdot,\kneg$,
       and these satisfy the laws of Boolean algebra, with $\kone$ as true
and $\kdot$ as conjunction.
\item $\kK$ is an idempotent semiring and the 
following hold for all $x,y,z$ in $\kK$.
\[\begin{array}{lcl}
\kone \kplus x\kdot x^\kstar &=& x^\kstar \\
\kone \kplus x^\kstar\kdot x &=& x^\kstar \\
y \kplus x\kdot z \leq z &\imp& x^\kstar\kdot y \leq z \\
y \kplus z\kdot x \leq z &\imp& y\kdot x^\kstar \leq z 
\end{array}
\]
\end{itemize}
The ordering $\leq$ is defined by \graybox{$x\leq y$} iff $x+y=y$.
The sequence/conjunction operator $;$ binds tighter than $+$.
\end{definition}

In a \dt{relational model}~\cite{Kozen1996}, $\kK$ is some set of relations on some set $\Sigma$,
with $\kzero$ the empty relation, 
$\kone$ the identity relation on $\Sigma$,
$+$ union of relations,
$\kdot$ relational composition, and $*$ reflexive-transitive closure. 
Moreover  $\kB$ is a set of \dt{coreflexives}, i.e., subsets of the identity relation $\kone$,
and $\kneg$ is complement with respect to $\kone$.
The set $\kK$ needs to be closed under these operations but need not be the full set of relations on $\Sigma$; likewise, $\kB$ need not be all coreflexive relations.
In a relational model, $\leq$ is set inclusion.

The \dt{relational model for GCL}, denoted $\relKAT$, is the relational model comprising all relations on $\Sigma$ where $\Sigma$ is the set of variable stores.\footnote{We could as well choose $\kB$ to be the least set of relations that
contains the coreflexive  $\{(s,s) \mid \means{e}(s) = true \}$ for every primitive boolean expression $e$, 
and is closed under $\kdot,\kplus,\kneg$.
And choose $\kK$ to be the least set of relations that contains $\kB$,
contains $\means{\lassg{}{x}{e}}$ for every assignment $\lassg{}{x}{e}$,
and is closed under the operations.  But the ``full'' model is simpler to formulate in Coq as we have done.
}

In this paper we work with $\relKAT$ and also with equational proofs with equational hypotheses.  The latter are formalized in terms of KAT expressions, together with interpretations that map KAT expressions to a model.  But we also need GCL syntax and its translation to KAT expressions.
Our formalization is meant to be precise but streamlined.

KAT expressions are usually defined with respect to given finite sets of primitive tests and actions~\cite[Sect.\ 2.3]{Kozen1996}.
We only need the special case where the actions are assignment commands and the tests are primitive boolean expressions of GCL; of course there are infinitely many of these.
Define the \dt{KAT expressions} by this grammar.
(Recall that $bprim$ stands for primitive boolean expressions in GCL.) 
\begin{equation}\label{eq:KATexp}
\KE ::= \underline{bprim} \mid \underline{\lassg{}{x}{e}} \mid
\KE + \KE \mid \KE\mathbin{;}\KE \mid \KE^* \mid \neg \KE \mid 1 \mid 0 
\end{equation}
The primitive KAT expression forms $\underline{bprim}$ and $\underline{\lassg{}{x}{e}}$ are written this way to avoid confusion with their counterparts in GCL syntax.

The \dt{GCL-to-KAT translation} \graybox{$\mkt{-}$} 
maps GCL commands and boolean expressions to KAT expressions.
For boolean expressions it is defined by:
\[
\begin{array}{lcl}
\mkt{bprim} &\eqdef& \underline{bprim} \\ 
\mkt{e\land e'} &\eqdef & \mkt{e}\kdot\mkt{e'} \\
\mkt{e\lor e'} &\eqdef & \mkt{e}\kplus\mkt{e'} \\
\mkt{\neg e} &\eqdef & \kneg\mkt{e}
\end{array}\]
For commands, the translation is defined by mutual recursion with the definition of $\mkt{gcs}$,
as follows.\footnote{Using Dijkstra's notation   
   $\quant{+}{dummies}{range}{term}$~\cite{DijkstraScholten}
   rather than the more common $\Sigma$ notation for finite sums.}
\begin{equation}\label{eq:def:mkt}
\begin{array}{l@{\;}c@{\;}l}
%\mkt{\lassg{}{x}{e}} &  & \mbox{is given} \\
\mkt{x:=e} & \eqdef  & \underline{x:=e} \\
\mkt{\skipc} & \eqdef  & 1 \\
\mkt{c;d} & \eqdef  & \mkt{c}\kdot\mkt{d} \\
\mkt{\lifgc{}{gcs}} & \eqdef  & \mkt{gcs} \\
\mkt{\ldogc{}{gcs}} & \eqdef  & \mkt{gcs}^\kstar \kdot \kneg\mkt{\enab(gcs)}  \\[1ex]
\mkt{gcs} & \eqdef  & \quant{\kplus}{e,c}{(e\gcto c) \in gcs}{\mkt{e}\kdot\mkt{c}} 
\end{array} 
\end{equation}
This is an adaptation of the well known translation of imperative programs into a KAT~\cite{Kozen97}.

% FUTURE: there's a reason Dexter didn't use this old-fashioned superscript notation for models
% and interpretations.  Let's change.

Let $\anyKAT$ be a model $(\kK^\anyKAT,\kB^\anyKAT,\kplus^\anyKAT,\kdot^\anyKAT,\kstar^\anyKAT,\kneg^\anyKAT,\kone^\anyKAT,\kzero^\anyKAT)$.
Given interpretations $\underline{bprim}^\anyKAT$ and $\underline{x:=e}^\anyKAT$ 
in $\anyKAT$ for each primitive KAT expression,
we get an interpretation $\KE^\anyKAT$ for every $\KE$, 
defined homomorphically as usual; 
for example: 
\[ %(\KE^\kstar)^\anyKAT \eqdef (\KE^\anyKAT)^{\kstar^\anyKAT} 
(\KE_0 \kdot \KE_1)^\anyKAT \eqdef \KE_0^\anyKAT \,\kdot^\anyKAT\,\KE_1^\anyKAT 
\qquad
(\KE_0+\KE_1)^\anyKAT \eqdef \KE_0^\anyKAT \,+^\anyKAT\,\KE_1^\anyKAT \]
Henceforth we never write $+^\anyKAT$ because the reader can infer whether the symbol
$+$ is meant as syntax in a KAT expression or as an operation of a particular KAT.

Consider the relational model $\relKAT$ for GCL.
For any primitive boolean expression $bprim$ in GCL,
we define an interpretation in $\relKAT$ for the KAT expression $\underline{bprim}$.
We also define an interpretation for every KAT expression $\underline{x:=e}$.
The definitions are 
$\underline{bprim}^\relKAT \eqdef \{(s,s) \mid \means{bprim}(s) = true \}$
and 
$\underline{\lassg{}{x}{e}}^\relKAT  \eqdef  \means{\lassg{}{x}{e}}$.
As a consequence, we have 
\[ %begin{equation}\label{eq:def:priminterp}
\mkt{bprim}^\relKAT = \{(s,s) \mid \means{bprim}(s) = true \} 
\qquad
\mkt{\lassg{}{x}{e}}^\relKAT  = \means{\lassg{}{x}{e}} 
\]
For a command translated to a KAT expression, 
the interpretation in the relational model $\relKAT$ coincides 
with the denotational semantics of GCL.

\begin{lemma}\label{lem:correctInterp} % named after the Coq lemma 
\upshape
For all commands $c$ we have $\mkt{c}^\relKAT = \means{c}$.
For all boolean expressions $e$ we have 
$\mkt{e}^\relKAT = \{(s,s) \mid \means{e}(s) = true \}$.
\end{lemma}
The proof is by induction on $e$ and then on $c$.
The base cases are by definition.
A secondary induction is used for the loop case.  
%FUTURE
%% \dn{One could do this by defining $*$ in the relational model as union of iterates,
%% and using union-of-iterates characterization of $\means{-}$, 
%% or use big-step style definition of $\means{-}$ together with rule induction def of transitive closure for $*$.  What's something succinct to say here, if we need to say anything?}

\paragraph{KAT provability}

Equivalence of commands will be defined in terms of KAT provability.

\begin{definition}[KAT provability] %\label{def:KATprov}
Let $H$ be a set of equations between KAT expressions.
We write \graybox{$H \proves \KE_0 = \KE_1$} to say $\KE_0=\KE_1$ 
is provable by propositional and equational reasoning from $H$ plus 
the axioms of KAT (Def.~\ref{def:KAT}).
\end{definition}
Propositional reasoning is needed for the $*$-induction rules,
i.e., the two implications in Def.~\ref{def:KAT}.
The important fact is that for any model $\anyKAT$ of $KAT$, 
if $H \proves \KE_0 = \KE_1$ and the equations in $H$ hold when interpreted in $\anyKAT$,
then $\KE_0^\anyKAT = \KE_1^\anyKAT$.

Compared with the program notation, KAT expressions serve to put expressions, commands, and guarded commands on the same footing. This is convenient for formulating results such as the following, which uses the empty set of hypotheses.

\begin{lemma}\label{lem:enabmkt}
\upshape
For any $gcs$ we have 
$\proves \mkt{gcs} = \mkt{\enab(gcs)} \kdot \mkt{gcs}$.
\end{lemma}
\begin{proof}
By unfolding definitions, using distributivity (of $;$ over $+$), and using that $\mkt{e_0};\mkt{e_1}\leq\mkt{e_0}$
where $e_0$ and $e_1$ are guards of different guarded commands in $gcs$.
\end{proof}
For example,
$\enab(x>0\gcto y:=1 \gcsep x<z\gcto y:=0) = x>0 \lor x < z$
and the lemma says 
$\mkt{x>0\gcto y:=1 \gcsep x<z\gcto y:=0} = 
\mkt{x>0 \lor x < z}\kdot\mkt{x>0\gcto y:=1 \gcsep x<z\gcto y:=0}$.

\paragraph{Correctness in KAT}

For specifications where the pre- and post-condition are expressible in the KAT,
one can express correctness judgments;
one form is $pre;code\leq pre;code;post$.
This is equivalent to 
$pre;code;\neg post = 0$
and also to   
$pre;code = pre;code;post$.

On this basis, one can dispense with HL in favor of equational reasoning in KAT~\cite{Kozen00}.
In our setting, for expressions $e_0$ and $e_1$ we have that
$ \models c:\spec{e_0}{e_1}$
is equivalent to the equation 
\[ \mkt{e_0}\kdot\mkt{c} \leq \mkt{e_0}\kdot\mkt{c}\kdot\mkt{e_1} \]
being true in $\relKAT$, i.e.,
\(    \mkt{e_0}^\relKAT\kdot\mkt{c}^\relKAT \leq \mkt{e_0}^\relKAT\kdot\mkt{c}^\relKAT\kdot\mkt{e_1}^\relKAT \).

In this paper we do not use the KAT formulation for correctness judgments in general; indeed it would be difficult to reconcile with the shallow embedding of store predicates.
However, we will use this encoding of correctness for a limited purpose, namely certain
axioms about assignment, where the pre- and post-condition will be expressions in GCL.

\paragraph{Command equivalence}

The idea for the equivalence condition $c\kateq d$ in rules \rn{Rewrite} and \rn{rRewrite}  
is that it should mean $H\proves \mkt{c}=\mkt{d}$ for a suitable set of hypotheses
that axiomatize the semantics of some primitive boolean expressions and commands.
A typical axiom for an assignment is a correctness equation like
$\mkt{x\geq 0}\kdot\mkt{x:=x+1} = \mkt{x\geq 0}\kdot\mkt{x:=x+1}\kdot\mkt{x>0}$.
A typical axiom for tests is 
%$\mkt{x<y}\leq\mkt{x\leq y}$.
$\mkt{odd(x)}= \neg\mkt{even(x)}$.
As an example, using the latter equation and KAT laws we have for any $c,d$ that 
%\[ \lifgc{}{ odd(x) \gcto c \gcsep \neg odd(x) \gcto d } \quad\kateq\quad
%\lifgc{}{ even(x) \gcto d \gcsep \neg even(x) \gcto c } \]
%That is, we have
\[ \mkt{odd(x)} = \neg\mkt{even(x)} \;\proves\;
%\lifgc{}{ odd(x) \gcto c \gcsep \neg odd(x) \gcto d } \kateq
%\lifgc{}{ even(x) \gcto d \gcsep \neg even(x) \gcto c } \]
\mkt{\lifgc{}{ odd(x) \gcto c \gcsep \neg odd(x) \gcto d }} =
\mkt{\lifgc{}{ even(x) \gcto d \gcsep \neg even(x) \gcto c }} \]
The equivalence is true schematically and does not depend on anything about $c$ or $d$.  

To prove our main results we only need a few axioms that will be discussed later.
But for the sake of a straightforward deductive system we 
formulate equivalence in terms of a generic collection of hypotheses.

\begin{definition}\label{def:Hyp}
Define \graybox{$\Hyp$} to be the set of equations on KAT expressions 
given by these two cases:
\begin{itemize}
\item the equation $\mkt{e}=0$ for every boolean expression $e$
such that $e \imp\mathit{false}$ is valid
\item the equation $\mkt{e_0};\mkt{x:=e};\kneg\mkt{e_1}=\kzero$
for all assignments $x:=e$
and boolean expressions $e_0$, $e_1$ 
such that $e_0\imp \subst{e_1}{x}{e}$ is valid.
\end{itemize}
\end{definition}
Validity here refers to the program semantics (Sect.~\ref{sec:unary}).
In light of Lemma~\ref{lem:correctInterp},
the set $\Hyp$ can characterized as follows:\footnote{Using program semantics
and in particular the fact that $\subst{e_1}{x}{e}$ gives the weakest 
precondition for assignment.}
\begin{itemize}
\item $\mkt{e}=0$ for every boolean expression $e$
such that $\mkt{e}^\relKAT = \emptyset$
\item  $\mkt{e_0};\mkt{x:=e};\kneg\mkt{e_1}=\kzero$
for every $x,e,e_0,e_1$ such that 
$\mkt{e_0}^\relKAT;\mkt{x:=e}^\relKAT;\kneg\mkt{e_1}^\relKAT=\emptyset$.
\end{itemize}

\begin{definition}[command equivalence]\label{def:kateq}
Define $\kateq$ by 
\( \graybox{$c \kateq d$} \eqdef \Hyp \proves \mkt{c} = \mkt{d} \). 
\end{definition}
As an example, we have $\lskipc{};\lskipc{} \kateq \lskipc{}$.
Another example is the equivalence
\[ \ldogc{}{ e_0 \gcto c } \kateq 
   \ldogc{}{ e_0 \gcto c;\ldogc{}{e_0\land e_1\gcto c}} \]
(which holds for any $c,e_0,e_1$) used in the loop tiling example of~\cite{BNN16}.

\begin{lemma}\label{lem:Hyp}
\upshape
Every equation in $\Hyp$ holds in $\relKAT$.
\end{lemma}
This is an easy consequence of the definition of $\Hyp$.

\paragraph{Soundness of HL+ and RHL+}

Having defined $\kateq$, 
the key ingredient of the rules \rn{Rewrite} and \rn{rRewrite},
we have completed the definition of HL+ and RHL+.
Finally we can state that the proof rules are sound in the sense that they infer valid judgments from valid premises.

\begin{lemma} %\label{lem:rewriteSound}
\upshape
The rules \rn{Rewrite} in Fig.~\ref{fig:HLplus} 
and \rn{rRewrite} in Fig.~\ref{fig:RHL} are sound.
\end{lemma}
\begin{proof}
For \rn{Rewrite}, suppose $c \kateq d$ holds.
That means $\Hyp\proves\mkt{c}=\mkt{d}$ so
$\mkt{c}^\anyKAT = \mkt{d}^\anyKAT$ in any model $\anyKAT$ that satisfies
the equations $\Hyp$.  In particular, by Lemma~\ref{lem:Hyp} we have
$\mkt{c}^\relKAT=\mkt{d}^\relKAT$,
hence $\means{c}=\means{d}$ by Lemma~\ref{lem:correctInterp}.
Suppose the premise of \rn{Rewrite} holds,
i.e., $\models c: \spec{P}{Q}$.
This is a condition on $\means{c}$ ---see (\ref{eq:valid})--- so we have $\models d: \spec{P}{Q}$.

The proof of \rn{rRewrite} is similar.
\end{proof}

Soundness proofs for the other rules is straightforward.

\begin{proposition}\label{prop:HLsound}
\upshape
All the rules of HL+ are sound.
\end{proposition}

\begin{proposition}\label{prop:RHLsound}
\upshape
All the rules of RHL+ (Fig.~\ref{fig:RHL}) are sound.
\end{proposition}

\subsection{Normal form reduction}\label{sec:nf}

\paragraph{Instrumenting a program with a program counter}

Choose a variable name $pc$.
We define $\addPC$, a function from commands to commands,
that adds $pc$ as in the example $c0^+$.
The definition of $\addPC$ is by structural recursion, with mutually recursive 
helpers $\addPC_0$ and $\addPC_1$ that map over guarded commands, and is written using the
abbreviations in (\ref{eq:pcabbrev}).
\[\begin{array}{lcl}
\addPC(\lskipc{n}) &\eqdef& \spc  n \,; \lskipc{n} \\
\addPC(\lassg{n}{x}{e}) &\eqdef& \spc  n \,; \lassg{n}{x}{e} \\
\addPC(c;d) &\eqdef& \addPC(c) \,; \addPC(d) \\
\addPC(\lifgc{n}{gcs}) &\eqdef& \spc n \,; \lifgc{n}{(\addPC_0(gcs))} \\
\addPC(\ldogc{n}{gcs}) &\eqdef& \spc n \,; \ldogc{n}{(\addPC_1(n,gcs))} \\[1ex]
\addPC_0(e\gcto c) &\eqdef& e\gcto \addPC(c) \\[.5ex]
\addPC_0(e\gcto c \gcsep gcs) &\eqdef& e\gcto \addPC(c) \gcsep \addPC_0(gcs)\\[1ex]
\addPC_1(n,e\gcto c) &\eqdef& e\gcto \addPC(c);\spc n \\[.5ex]
\addPC_1(n,e\gcto c \gcsep gcs) &\eqdef& e\gcto \addPC(c);\spc n \gcsep \addPC_1(n,gcs)
\end{array}\]
Note that for \keyw{if} commands, $\addPC_0$ maps $\addPC$ over the guarded commands,
and for \keyw{do}, $\addPC_1$ additionally adds a trailing assignment to set $pc$ to 
the loop label.

\begin{lemma}\label{lem:addPC}
\upshape
For any $c$, if $pc$ does not occur in $c$,
then we have $\ghost(pc,\addPC(c))$.
Moreover $ \erase(pc,\addPC(c)) \kateq c$.
\end{lemma}
To show $\erase(pc,\addPC(c)) \kateq c$, 
we must show $\Hyp\proves\mkt{\erase(pc,\addPC(c))} = \mkt{c}$. 
In fact we can show $\proves\mkt{\erase(pc,\addPC(c))} = \mkt{c}$. 
To do so, go by induction on $c$, using that $\mkt{\skipc} = \kone$
and $\kone$ is the unit of sequence.
The proof of $\ghost(pc,\addPC(c))$ is also straightforward.

\paragraph{Normal forms} 

We define the ternary relation $\norm{c}{m}{gcs}$
to relate a command and a label to the gcs that will be the body of its normal form;
see Fig.~\ref{fig:norm}.
To be precise, the relation depends on a choice of $pc$ variable but we leave this implicit,
as it is in the $\tpc \_$ and $\spc \_$ notations of (\ref{eq:pcabbrev}).
We can write ``$\norm{c}{m}{gcs}$ (for $pc$)'' to make the choice explicit.

There are two kinds of guarded commands in a normal form body:
\begin{description}
\item[(test nf)] \quad $\tpc n\land e \gcto \spc m$
\quad\mbox{i.e., $pc = n \land e \gcto \lassg{}{pc}{m}$}
\item[(set nf)] \quad $\tpc n\gcto \lassg{}{x}{e};\spc m$
\quad\mbox{i.e., $pc = n \gcto \lassg{}{x}{e}; \lassg{}{pc}{m}$}
\end{description}
In the normal form for skip, we write $\tpc n\gcto \spc f$ to abbreviate 
$\tpc n\land true\gcto \spc f$ which is a test nf.

An example is $\norm{\lassg{4}{x}{x-1}}{2}{(\tpc 4\gcto x:=x-1;\spc 2)}$.
For the running example, we have $\norm{c0}{6}{gcs}$ where $gcs$ is the body of $d0$
which is defined in Sect.~\ref{sec:c0unary}.

For if-commands the rule in Fig.~\ref{fig:norm} only handles the case of two guarded commands.
For do, there is a rule for one guarded command and a rule for two.
To handle any number of guarded commands,
the general rules are given in the appendix. 
But these special cases are easier to read because they avoid the need to map over guarded command lists. % and catenate normal form bodies. 

\begin{definition}
The \dt{automaton normal form} of a command $c$ with $\lab(c)=n$ and label $f\notin\labs(c)$, for chosen variable $pc$, is 
\[ \spc  n ; \ldogc{}{ gcs } 
\quad\mbox{where } \normSmall{c}{f}{gcs} 
\]
\end{definition}
For brevity we say ``normal form'', as this is the only normal form of interest in this paper.
As with the notations $\tpc n$ and $\spc n$, the notation 
$ \normSmall{c}{f}{gcs} $
does not make explicit the dependence on chosen variable $pc$.

\begin{figure*}
\begin{small}
\begin{mathpar}

\inferrule{}{
\norm{\lskipc{n}}{f}{
(\tpc n\gcto \spc f)}
}

\inferrule{}{
\norm{\lassg{n}{x}{e}}{f}{ 
(\tpc  n \gcto \lassg{}{x}{e}; \spc  f)}
}

\inferrule{
\norm{c}{\lab(d)}{gcs_0} \\
\norm{d}{f}{gcs_1}
}{
\norm{c;d}{f}{gcs_0\gcsep gcs_1}
}

\inferrule{
\norm{c_0}{f}{gcs_0} \\ \norm{c_1}{f}{gcs_1} 
}{ 
\norm{\lifgc{n}{e_0\gcto c_0 \gcsep e_1\gcto c_1}}{f}{
  \tpc  n \land e_0 \gcto \spc \lab(c_0)  \gcsep
  \tpc  n \land e_1 \gcto \spc \lab(c_1)  \gcsep
  gcs_0 \gcsep gcs_1
}}

\inferrule{
\norm{c}{n}{gcs}
}{
\norm{\ldogc{n}{e\gcto c}}{f}{
  \tpc  n\land e \gcto \spc \lab(c)  \gcsep
  \tpc  n\land\neg e \gcto \spc f \gcsep gcs 
}}

\inferrule{
\norm{c_0}{n}{gcs_0} \\ \norm{c_1}{n}{gcs_1} 
}{
\norm{\ldogc{n}{e_0\gcto c_0 \gcsep e_1\gcto c_1}}{f}{
  \tpc  n\land e_0 \gcto \spc \lab(c_0)  \gcsep
  \tpc  n\land e_1 \gcto \spc \lab(c_1)  \gcsep
  \tpc  n\land\neg (e_0 \lor e_1) \gcto \spc f \gcsep
  gcs_0 \gcsep gcs_1
}}

\end{mathpar}
\end{small}
%\vspace*{-5ex}
\caption{Normal form bodies
(see Fig.~\ref{fig:norm:general} for the if/do general cases)}\label{fig:norm}
\end{figure*}

A straightforward structural induction on $c$ proves the following:
\begin{lemma}\label{lem:normExists} % Coq
\upshape
For all $c$ and $f$, there is some $gcs$ with $\normSmall{c}{f}{gcs}$.
\end{lemma}
This does not require $c$ to be $\ok$, but the normal form is only useful if $\okf(c,f)$.
Note that that $gcs$ is uniquely determined by $c$ and $f$, so we could formulate the 
definition as a recursive function, but the relational formulation is convenient for proofs.
And we do not exploit uniqueness.

\paragraph{Enumerating the cases of a normal form}

As remarked preceding Lemma~\ref{lem:VCprog},
there are five kinds of transition in the small step semantics.
Each gives rise to a corresponding form of guarded command in the normal form.

\begin{lemma}\label{lem:nfCases}[guarded commands of a normal form]
\upshape
Suppose $\okf(c,f)$ and $\norm{c}{f}{gcs}$.
Every guarded command in $gcs$ has one of these five forms:
\begin{itemize}
\item $\tpc k\gcto \spc m$, 
for some $k,m$ such that $\sub(k,c)$ is $\lskipc{k}$ and $m=\fsuc(k,c,f)$.
\item $\tpc k\gcto \lassg{}{x}{e};\spc m$,
for some $k,m,x,e$ such that $\sub(k,c)$ is $\lassg{k}{x}{e}$ 
and $m = \fsuc(k,c,f)$.

\item $\tpc k\land e\gcto \spc m$,
for some $k,m,e,d,gcs_0$ such that $\sub(k,c)$ is $\lifgc{k}{gcs_0}$ 
and $m=\lab(d)$ where $e\gcto d$ is in $gcs_0$.

\item $\tpc k\land e\gcto \spc m$, 
for some $k,m,e,d,gcs_0$ such that $\sub(k,c)$ is $\ldogc{k}{gcs_0}$
and $m=\lab(d)$ where $e\gcto d$ is in $gcs_0$.

\item $\tpc k\land \neg\enab(gcs_0)\gcto \spc m$, 
for some $k,m,gcs_0$ such that $\sub(k,c)$ is 
$\ldogc{k}{gcs_0}$ and $m=\fsuc(k,c,f)$.
\end{itemize}
\end{lemma}

\subsection{Axioms for normal form equivalence}\label{sec:KATequiv}

For the sake of straightforward presentation, command equivalence has been 
formulated in terms of a single fix set of hypotheses that axiomatize simple
assignments and boolean expressions (Defs.~\ref{def:Hyp} and~\ref{def:kateq}).
However, many useful equivalences require no such hypotheses.
Only a few specific hypotheses are needed to prove the normal form theorem,
and in this section we spell those out.

Observe that  if $\norm{c}{f}{gcs}$ then there is a finite set of instances of this relation that supports the fact, namely the normal forms of subprograms of $c$.
This enables us to define a set of axioms that are useful for reasoning about 
a given program $c$ and its normal form.  The definition is parameterized on $c$ and on the choice of a $pc$ variable and final label.

\begin{definition}\label{def:nfax}
The \dt{normal form axioms}, $\nfax(pc,c,f)$, 
comprises the following set of equations.
\begin{description}
\item[(diffTest)] 
$\mkt{ \tpc i } \kdot \mkt{ \tpc j } = \kzero$  for 
$i$ and $j$ in $\labs(c)\union\{f\}$ such that $i\neq j$
\item[(setTest)]
$\mkt{ \spc  i }\kdot\mkt{ \tpc  i } = \mkt{ \spc  i }$ for $i$ in $\labs(c)\union\{f\}$ 
 
%% \item[testSet]  NOT USED 
%% $\mkt{ \tpc  i }\kdot\mkt{ \spc  i } = \mkt{ \tpc  i }$ for $i$ in $\labs(c)\union\{f\}$ 
%% %% %  ; setConst : (* !i ; !j = !j *) 
%% %% \mkt{ \spc  i } \kdot \mkt{ \spc  j } & = & \mkt{ \spc  j } % \quad\mbox{(for all $i,j$)} 
\item[(totIf)]
$\kneg\mkt{\enab(gcs)} = 0$ for every $\lifgc{}{gcs}$ that occurs in $c$

\item[(testCommuteAsgn)] 
$\mkt{\tpc i} ; \mkt{\lassg{}{x}{e}} = \mkt{\lassg{}{x}{e}} ; \mkt{\tpc i}$
for every $\lassg{}{x}{e}$ in $c$ such that  $x\nequiv pc$,
and every $i$ in $\labs(c)\union\{f\}$ 
\end{description}
\end{definition}

\begin{lemma}\label{lem:nfaxHyp}
\upshape
For any $pc,c,f$, the equations $\nfax(pc,c,f)$ 
follow by KAT reasoning from the equations of $\Hyp$.
\end{lemma}
\begin{proof}
We make use of the alternate characterization of $\Hyp$ in terms of $\relKAT$
mentioned following Def.~\ref{def:Hyp}.

(diffTest) 
If $i\neq j$ then $\mkt{ \tpc i \land \tpc j }^\relKAT = \emptyset$, 
so $\Hyp$ contains $\mkt{ \tpc i \land \tpc j } = \kzero$.
So we have $\mkt{ \tpc i }\kdot \mkt{ \tpc j } = \mkt{ \tpc i \land \tpc j } =\kzero$
using the definition of $\mkt{-}$.

(setTest)
$\Hyp$ contains $\mkt{\skipc};\mkt{\spc i}\kdot\kneg\mkt{\tpc i}=\kzero$,
and $\mkt{\skipc}$ is $\kone$.
Now observe that 
\(\mkt{\spc i}
= 
\mkt{\spc i}\kdot(\mkt{\tpc i}+\kneg\mkt{\tpc i}) 
= 
\mkt{\spc i}\kdot\mkt{\tpc i} + \kone;\mkt{\spc i}\kdot\kneg\mkt{\tpc i}
= 
\mkt{\spc i}\kdot\mkt{\tpc i}\)
using KAT laws and $\kone;\mkt{\spc i}\kdot\kneg\mkt{\tpc i}=\kzero$.

(totIf)
If $\lifgc{}{gcs}$ that occurs in $c$ then 
we have $(\kneg\mkt{\enab(gcs)})^\relKAT = \emptyset$ 
as a consequence of the $\totalIf$ condition (Def.~\ref{def:lang}).
So the equation $\kneg\mkt{\enab(gcs)} = 0$ is in $\Hyp$.

(testCommuteAsgn)
The equation 
$\mkt{\tpc i} ; \mkt{\lassg{}{x}{e}} = \mkt{\lassg{}{x}{e}} ; \mkt{\tpc i}$,
is equivalent to the conjunction of 
$\mkt{\tpc i} ; \mkt{\lassg{}{x}{e}} ; \kneg\mkt{\tpc i} = 0$
and $\kneg\mkt{\tpc i} ; \mkt{\lassg{}{x}{e}} ; \mkt{\tpc i} = 0$
using KAT laws.  
By definition of $\mkt{-}$ and boolean algebra, 
$\kneg\mkt{\tpc i} ; \mkt{\lassg{}{x}{e}} ; \mkt{\tpc i} = 0$
is equivalent to 
$\mkt{\neg \tpc i} ; \mkt{\lassg{}{x}{e}} ; \kneg\mkt{\neg\tpc i} = 0$.
Both $\tpc i\imp \subst{(\tpc i)}{x}{e}$ and 
$\neg \tpc i\imp \subst{(\neg \tpc i)}{x}{e}$ 
are valid because $pc\not\equiv x$,
so $\Hyp$ contains both equations 
$\mkt{\tpc i} ; \mkt{\lassg{}{x}{e}} ; \kneg\mkt{\tpc i} = 0$
and 
$\mkt{\neg \tpc i} ; \mkt{\lassg{}{x}{e}} ; \kneg\mkt{\neg\tpc i} = 0$.
\end{proof}

\begin{remark}\label{rem:pspace}
\upshape
$\nfax(pc,c,f)$ is finite, and every equation is equivalent to one of the form $\KE=0$.  
% which is NOT true of testSet
So entailments of the form $\nfax(pc,c,f) \proves \KE_0 = \KE_1$
are decidable in PSPACE~\cite{KozenKATcomplex}.
\qed\end{remark}

\begin{lemma}\label{lem:nfax}
\upshape
All equations in $\nfax(pc,c,f)$ are true in $\relKAT$, for any $pc,c,f$.
\end{lemma}
This holds because the equations follow from $\Hyp$ (Lemma~\ref{lem:nfaxHyp})
and the equations in $\Hyp$ are true in $\relKAT$ (Lemma~\ref{lem:Hyp}).
Note that $c$ need not be $\ok$ for this result.
However, truth of (totIf) does depend on the $\totalIf$ condition (Def.~\ref{def:lang}).

\begin{lemma}\label{lem:nf-conseq}
\upshape 
The following hold for any $c,pc,f$ such that $pc$ does not occur in $c$.
\begin{list}{}{}
\item[(diffTestNeg)] $\nfax(pc,c,f)\proves\mkt{\tpc i} = \mkt{\tpc i};\neg\mkt{\tpc j}$ 
and $\mkt{\tpc i} \leq \neg\mkt{\tpc j}$ 
\\
for $i\neq j$ with $i,j$ in $\labs(c)\union\{f\}$.

\item[(nf-enab-labs)]
\( \nfax(pc,c,f)\proves
\mkt{\enab(gcs)} = \mkt{ \quant{\lor}{i}{i\in\labs(c)}{ \tpc  i }} \)
\\
if $\okf(c,f)$ and $\norm{c}{f}{gcs}$.

\item[(nf-lab-enab)] $\nfax(pc,c,f)\proves \mkt{\tpc \lab(c)}\leq\mkt{\enab(gcs)}$
\\
if $\okf(c,f)$ and $\norm{c}{f}{gcs}$.

\item[(nf-enab-disj)] $\nfax(pc,c,f)\proves \mkt{\enab(gcs_0)};\mkt{\enab(gcs_1)}= 0$
\\
if $\okf(c,f)$ and $gcs_0$ and $gcs_1$ are the normal form bodies of different subprograms of $c$.

\item[(nf-enab-corr)]
\( \nfax(pc,c,f)\proves  
\mkt{\enab(gcs)};\mkt{gcs} = 
\mkt{\enab(gcs)};\mkt{gcs};(\mkt{\enab(gcs)}\kplus\mkt{\tpc  f}) \)
\\
if $\okf(c,f)$ and $\norm{c}{f}{gcs}$ and $pc\notin vars(c)$.
\end{list}
(nf-enab-corr) can be abbreviated 
$\nfax(pc,c,f)\proves \mkt{gcs} : \spec{ \mkt{\enab(gcs)} }{ \mkt{\enab(gcs)}\kplus\mkt{\tpc  f}}$.
\end{lemma}
\begin{proof}
(diffTestNeg) follows from (diffTest) using boolean algebra.

(nf-enab-labs) 
By rule induction on the normal form judgment.
\\
For $\ldogc{n}{gcs_0}$, the normal form body is explicitly defined to include, 
not only a command with guard $\tpc n\land e$ for each $e\gcto d$ in $gcs_0$
but also a command with guard $\tpc n\land\neg\enab(\ldots)$ for their complement.
The disjunction of their enabling conditions simplifies to $\tpc n$.
By the normal form  rule and induction hypothesis, the enabling condition for the collected normal form bodies is the disjunction over their label tests, hence the result.
\\
For $\lifgc{n}{gcs_0}$, the normal form body includes disjuncts
$\tpc n\land e$ for each of the guard expressions $e$ in $gcs_0$.
Consider the case of two branches, i.e., $gcs_0$ is $e_0\gcto d_0\gcsep e_1\gcto d_1$,
the disjunction is $(\tpc n\land e_0) \lor (\tpc n\land e_1)$.
Observe
\[\begin{array}{lll}
  & \mkt{(\tpc n\land e_0) \lor( \tpc n\land e_1)}\\
= & \mkt{\tpc n};\mkt{e_0} + \mkt{\tpc n};\mkt{e_1} & \mbox{def $\mkt{-}$} \\
= & \mkt{\tpc n};(\mkt{e_0}+\mkt{e_1}) &\mbox{KAT law} \\ 
= & \mkt{\tpc n};\mkt{e_0\lor e_1} & \mbox{def $\mkt{-}$} \\
= & \mkt{\tpc n} & \mbox{axiom (totIf), unit law} 
\end{array}\]
The rest of the argument is like for \keyw{do}.
The general case (multiple branches) is similar.

(nf-lab-enab) follows from (nf-enab-labs), using boolean algebra and 
definition of $\mkt{-}$ (distributing over $\lor$). 

(nf-enab-disj) is proved by induction on normal form, 
using that by $\ok$ subprograms have distinct labels,
and (nf-enab-labs) and (diffTest).

(nf-enab-corr)
The proof is by induction on the normal form relation and by calculation
using KAT laws and the axioms.
It uses (testCommuteAsgn) so $pc$ must not occur in $c$. 
\end{proof}

With Lemma~\ref{lem:nf-conseq} we are ready to prove the normal form equivalence theorem.
Before proceeding, here is a semantic consequence that will be useful in Sects.~\ref{sec:autUnary} and~\ref{sec:acomplete}.

\begin{lemma}\label{lem:nf-enab-labs}
\upshape
If $\okf(c,f)$ and $\norm{c}{f}{gcs}$ then
$\means{\enab(gcs)} = \means{\quant{\lor}{i}{i\in\labs(c)}{ \tpc  i }}$.
\end{lemma}
\begin{proof}
Follows by definitions from provability lemma (nf-enab-labs) of Lemma~\ref{lem:nf-conseq},
using Lemmas~\ref{lem:nfax} and~\ref{lem:correctInterp}.
\end{proof}

\begin{remark}\label{rem:totIf}
\upshape
The (totIf) axiom of Def.~\ref{def:nfax} is only used for one purpose: to prove its consequence (nf-enab-labs) in Lemma~\ref{lem:nf-conseq}.
In turn, (nf-enab-labs) is used 
in several ways in the proof of the normal form equivalence Theorem~\ref{thm:normEquiv},
and its consequence Lemma~\ref{lem:nf-enab-labs} is used in the proofs
of the Floyd completeness Theorems~\ref{thm:FloydComplete}
and the alignment completeness Theorem~\ref{thm:acomplete}.

An alternate way to develop the theory is as follows.
First, impose totality of if-commands the syntactic condition 
in Remark~\ref{rem:totIfSyn}.
Second, drop (totIf) from Def.~\ref{def:nfax}.
Third, change the proof of (nf-enab-labs) in Lemma~\ref{lem:nf-conseq}.
Specifically,
the expression $\mkt{\tpc n};\mkt{e_0\lor e_1}$ 
becomes $\mkt{\tpc n};\mkt{e\lor \neg e}$ 
and we proceed 
$\mkt{\tpc n};\mkt{e\lor \neg e}
= \mkt{\tpc n};(\mkt{e}+\neg\mkt{e})
= \mkt{\tpc n}$.
\end{remark}

\subsection{Normal form equivalence theorem}\label{sec:nfthm}

Finally we are ready to prove the main result of Sect.~\ref{sec:KATnf},
which loosely speaking says every command is equivalent to one in automaton normal form.

\begin{theorem}\label{thm:normEquiv} 
\upshape
If $\okf(c,f)$, $pc\notin vars(c)$, and 
$\normSmall{c}{f}{gcs}$ (for $pc$) 
then
\begin{equation}\label{eq:normEquiv}
\spc n; \ldogc{}{gcs} \kateq \addPC(c); \spc f   % NOTE: reverse of Coq development, but fits with Rewrite
\quad\mbox{where $n=\lab(c)$.}
\end{equation}
\end{theorem}
\begin{proof}
By definition of $\kateq$ we must prove
$\Hyp\proves \mkt{\spc n; \ldogc{}{gcs}} = \mkt{\addPC(c); \spc f}$
using laws of KAT.
In fact we do not need all of $\Hyp$; our proof will show that
$\nfax(pc,c,f)\proves \mkt{\spc n; \ldogc{}{gcs}} = \mkt{\addPC(c); \spc f}$.
(Aside:  By Remark~\ref{rem:pspace}, such a judgment is decidable, so
any instance of Theorem~\ref{thm:normEquiv}, like (\ref{eq:c0d0}), can be decided.
But this does not help prove the theorem,
where we have infinitely many instances to prove
and must use the induction hypotheses.)

The proof goes by rule induction on $\normSmall{c}{f}{gcs}$.
Once definitions are unfolded, the core of the argument in each case 
is a calculation in KAT, using the normal form  axioms.
There is one normal form rule per command form, so we proceed by cases on those forms.
In each case, aside from unfolding definitions of $\addPC$, $\mkt{-}$, etc.,
we use only KAT reasoning and the $\nfax$ hypotheses (Def.~\ref{def:nfax}) 
and their consequences including those in Lemma~\ref{lem:nf-conseq}.

\medskip
\textbf{Case $c$ is $\lskipc{n}$}.
We have $\normSmall{\lskipc{n}}{f}{(\tpc n\gcto \spc f)}$.
Operationally, the loop $\ldogc{}{\tpc n\gcto \spc f}$
iterates exactly once.  This is reflected in the 
following proof of (\ref{eq:normEquiv}) for this case, in which we unroll the loop once.
(Note: in hints we do not mention associativity, unit law for 1, etc.)
\[\begin{array}{lll}
  & \mkt{\spc n ; \ldogc{}{ \tpc n \gcto \spc f} } \\
= & \hint{def $\mkt{-}$, see (\ref{eq:def:mkt}), and def $\enab$ } \\ 
  & \mkt{\spc n}; (\mkt{\tpc n} ; \mkt{\spc f })^*;  \neg\mkt{\tpc n } \\
= & \hint{star unfold, distrib} \\
  & \mkt{\spc n}; \neg\mkt{\tpc n }  + \mkt{\spc n };  \mkt{\tpc n};  \mkt{\spc f };  (\mkt{\tpc n};  \mkt{\spc f })^*;  \neg\mkt{\tpc n } \\
= &\hint{left term is 0, using axiom (setTest) and lemma (diffTestNeg)} \\
  & \mkt{\spc n}; \mkt{\tpc n}; \mkt{\spc f}; (\mkt{\tpc n};\mkt{\spc f})^* ;\neg\mkt{\tpc n} \\
= &\hint{(setTest) for $n$}\\
  & \mkt{\spc n};  \mkt{\spc f }; (\mkt{\tpc n}; \mkt{\spc f })^*; \neg\mkt{\tpc n} \\
= &\hint{(setTest) for $f$} \\
  & \mkt{\spc n};  \mkt{\spc f };  \mkt{\tpc f };  (\mkt{\tpc n};  \mkt{\spc f })^*;  \neg\mkt{\tpc n } \\
= &\hint{axiom (diffTest) with $f\neq n$ from $\okf(c,f)$; KAT fact $p;(q;a)^* = p$ if $p;q=0$, for any $p,q,a$ } \\ %  Coq lemma star_test_contra 
  & \mkt{\spc n};  \mkt{\spc f };  \mkt{\tpc f }; \neg\mkt{\tpc n } \\
= &\hint{lemma (diffTestNeg), $f\neq n$ by $\okf(c,f)$} \\
  & \mkt{\spc n};  \mkt{\spc f };  \mkt{\tpc f } \\
= &\hint{axiom (setTest)} \\
  & \mkt{\spc n};  \mkt{\spc f } \\
= & \hint{def $\mkt{-}$, 1 is unit of; } \\
  & \mkt{\spc n};  \mkt{\lskipc{n}};  \mkt{\spc f } \\
= & \hint{def $\mkt{-}$} \\
  & \mkt{\spc n ; \lskipc{n}; \spc f }  \\
= & \hint{def $\addPC$} \\
  & \mkt{\addPC(\lskipc{n}); \spc f } 
  \end{array}
\]
\medskip
\textbf{Case $c$ is $x:=e$}.  Similar to skip case; omitted.

\medskip
\textbf{Case $c$ is $c_0;c_1$}. % (nothing to do with our running example $c0$)
So $n=\lab(c_0)$. Let $n_1\eqdef \lab(c_1)$.
Let $p \eqdef \enab(gcs_0)$ and  $q \eqdef \enab(gcs_1)$.
Suppose $\normSmall{c_0;c_1}{f}{gcs_0\gcsep gcs_1}$ where
$\normSmall{c_0}{n_1}{gcs_0}$ and $\normSmall{c_1}{f}{gcs_1}$.
Observe
\[\begin{array}{lll}
  & \mkt{\addPC(c_0;c_1); \spc f }  \\
= & \hint{def $\addPC$, def $\mkt{-}$} \\
  & \mkt{\addPC(c_0)} ; \mkt{\addPC(c_1); \spc f}  \\
= & \hint{ induction hypothesis for $c_1$ } \\
  & \mkt{ \addPC(c_0) } ; \mkt{\spc n_1; \ldogc{}{gcs_1} } \\
= & \hint{ def $\mkt{-}$ } \\
  & \mkt{ \addPC(c_0); \spc n_1}; \mkt{\ldogc{}{gcs_1} } \\
= & \hint{induction hypothesis for $c_0$} \\
  & \mkt{\spc n;\ldogc{}{gcs_0 } } ; \mkt{\ldogc{}{gcs_1} } \\
= & \hint{ def $\mkt{-}$, abbreviations $p=\enab(gcs_0)$, $q = \enab(gcs_1)$} \\
  & \mkt{ \spc n } ; \mkt{ gcs_0 }^* ; \mkt{\neg p} ; \mkt{ gcs_1 }^* ; \mkt{\neg q} \\
= & \hint{step $(\dagger)$, see below} \\
  & \mkt{ \spc n } ; (\mkt{gcs_0} + \mkt{gcs_1})^* ; \mkt{ \neg(p+q) } \\
= & \hint{def $\mkt{-}$, def $\enab$, using $p= \enab(gcs_0)$ and $q = \enab(gcs_1)$} \\
  & \mkt{ \spc n; \ldogc{}{gcs_0 \gcsep gcs_1} } 
  \end{array}
\]
The step $(\dagger)$ relates a sequence of normal forms with the single main normal form.
It is proved using only KAT reasoning from the normal form axioms 
(including Lemma~\ref{lem:nf-conseq}).
To prove step $(\dagger)$ we first introduce more abbreviations besides $p,q$ above.
\[
c  := \mkt{gcs_0}\qquad
d  := \mkt{gcs_1}\qquad
\spc n := \mkt{\spc n} \mbox{ etc.}
\]
This shadowing of $c,\tpc n$ etc.\ is harmless because we are done referring to the originals.
In the following we use juxtaposition instead of semicolon.
Now step $(\dagger)$ can be written
\[ (\dagger)\qquad
 \spc n c^* \neg p d^* \neg q = \spc n (c + d)^* \neg (p+q) 
\] 
To prove $(\dagger)$ we use the following instances and consequences of $\nfax$.
\[\begin{array}{ll}
c = pc & \hint{Lemma~\ref{lem:enabmkt} (recall that $p$ is $\enab(c)$)} \\ % (axiom \verb+mkt_gcs_enab+)
d = qd & \hint{Lemma~\ref{lem:enabmkt}} \\
\tpc n \leq  p & \hint{a command's label implies its enab, (nf-lab-enab) in Lemma~\ref{lem:nf-conseq}} \\ 
\tpc n_1 \leq  q & \hint{(nf-lab-enab)}\\
pq = 0  & \hint{(nf-enab-disj) }\\
p\ \tpc f = 0 & \hint{from (nf-enab-labs)} \\
q\ \tpc f = 0 & \hint{from (nf-enab-labs)} \\
c : \spec{p}{(p + \tpc n_1)} &  \hint{(nf-enab-corr)}\\
d : \spec{q}{(q + \tpc f)} & \hint{(nf-enab-corr)} 
\\
pd = 0 &\hint{by $d=qd$ and $pq=0$} \\
qc = 0 &\hint{by $c=pc$ and $pq=0$} \\
q = q \neg p &\hint{by bool alg and $pq=0$: $q = q(p+\neg p) = 0+q\neg p$} \\
p\  \tpc n_1 = 0 &\hint{using $pq=0$ and $\tpc n_1\leq q$} \\
\tpc n\  \tpc n_1 = 0 &\hint{by $\tpc n\leq p$, $\tpc n_1\leq q$, $pq=0$} \\
\tpc n d = 0 &\hint{by $\tpc n d = \tpc n q d \leq p q d = 0$} \\
d : \spec{\neg p}{\neg p} & \hint{proved using $d : \spec{q}{(q + \tpc f)}$ etc} \\
 %% %   (but not c : \neg q ~> \neg q)  [used in norm_equiv proof of seq case]
 %%  because \neg p d = \neg p q d                 axiom d=qd
 %%               = \neg p q d (q + ?f)        axiom d : q ~> (q + ?f)
 %%               = \neg p q d (q + ?f)        axiom d : q ~> (q + ?f)
 %%               = \neg p q d (q \neg p + ?f)     bool alg using qp=0 (lemma above)
 %%               = \neg p q d (q \neg p + ?f \neg p)  bool alg using axiom p ?f = 0
 %%               = \neg p q d (q + ?f) \neg p     distrib
 %%               = \neg p d \neg p                axioms d=qd and d : q ~> (q + ?f)
c : \spec{p + \tpc n_1}{p + \tpc n_1} &\hint{proved using $c : \spec{p}{(p + \tpc n_1)}$ etc} \\
  %% because (p + ?n_1)c = p c + ?n_1 c
  %%                    = p c + ?n_1 p c    axiom c=pc
  %%                    = p c              lemma p ?n_1 = 0 
  %%                    = p c (p + ?n_1)    axiom c : p ~> (p + ?n_1)
\spc n c^* \neg p = \spc n c^* \tpc n_1 \neg p &\hint{proved using invariance and preceding lemma}    \\
  %% because
  %% !n c^* \neg p = !n (p + ?n_1) c^* \neg p            ax !n=!n?n, ax ?n\leq p, lemma ?n?n_1=0 
  %%          = !n (p + ?n_1) c^* (p + ?n_1) \neg p  invariance and preceding lemma 
  %%          = !n c^* (p + ?n_1) \neg p            reverse first step
  %%          = !n c^* ?n_1 \neg p                  distrib, bool alg for (p \neg p) 
\spc n c^* = \spc n c^* (p + \neg p \tpc n_1)  % lemma X 
  %% because have  c^* : (p+?n_1) ~> (p+?n_1)  by invariance and lemma above 
  %% whence using !n \leq  p and !n ?n_1 = 0 have  c^* : ?n ~> (p+?n_1)
  %% Now use !n = !n ?n to get !n c^* = !n c^* (p + \neg p ?n_1) 
  \end{array}
\]
For the last equation, we have 
$c^* : \spec{p+\tpc n_1}{p+\tpc n_1}$ by invariance and a lemma above.
Using $\tpc n \leq p$ and $\tpc n \tpc n_1 = 0$ 
$c^* : \spec{\tpc n}{p+\tpc n_1}$.
Now use (setTest) to get $\spc n c^* = \spc n c^* (p + \neg p \tpc n_1)$.

We prove $(\dagger)$ by mutual inequality.
To prove $LHS(\dagger) \leq  RHS(\dagger)$, observe that 
\[\begin{array}{lll}
      & \spc n c^* \neg p d^* \neg q \leq  \spc n (c + d)^* \neg (p+q) \\
\impby& \hint{monotonicity of seq} \\
      & c^* \neg p d^* \neg q \leq  (c + d)^* \neg (p+q) \\
\impby& \hint{induction, join property of +} \\ 
      & \neg p d^* \neg q \leq  (c + d)^* \neg (p+q) 
  \;\land\; c (c + d)^* \neg (p+q) \leq  (c + d)^* \neg (p+q) 
  \end{array}\]
We have the first conjunct because
\[\begin{array}{lll}
  & \neg p d^* \neg q  \\
= 
  & \neg p d^* \neg p\neg q  & \hint{lemma $d : \spec{\neg p}{\neg p}$, invariance}\\
= 
  & \neg p d^* \neg (p+q) & \hint{bool alg} \\
\leq
  & d^* \neg (p+q) &\hint{test below 1} \\
\leq 
  & (c+d)^* \neg (p+q) &\hint{* mono, $c\leq c+d$} 
  \end{array}\]
We have the second conjunct because
\[\begin{array}{lll}
     & c (c + d)^* \neg (p+q) \\
\leq 
     & (c + d) (c + d)^* \neg (p+q) &\hint{seq mono} \\
\leq 
     &  (c + d)^* \neg (p+q) & \hint{star fold} 
     \end{array}\]
To prove $RHS(\dagger) \leq LHS(\dagger)$ the key property is that $d$ is not enabled until $c$ disabled:
\[ \mbox{(i)}\qquad \spc n (c + d)^* \leq  \spc n ( c^* + c^* \neg p d^* ) \]
By *-induction, (i) follows from (ii) and (iii).
\[\begin{array}{ll}
\mbox{(ii)}& \quad \spc n \leq  \spc n ( c^* + c^* \neg p d^* )  \\
\mbox{(iii)}& \quad \spc n ( c^* + c^* \neg p d^* ) (c + d) \leq  \spc n ( c^* + c^* \neg p d^* ) 
\end{array}\]
For (ii) we have
$\spc n \leq  \spc n c^* \leq  \spc n ( c^* + c^* \neg p d^* )$ by KAT laws.
\\
For (iii), it is equivalent by distributivity to 
\[
 \spc n c^* c + \spc n c^* d + \spc n c^* \neg p d^* c + \spc n c^* \neg p d^* d \leq  \spc n c^* + \spc n c^* \neg p d^*  
\]
We show each term on the left is below one of the terms of the right.
\begin{itemize}
\item $\spc n c^* c \leq  \spc n c^*$ by star law
\item 
\(\begin{array}[t]{lll}   
    &  \spc n c^* d \\
  = & \spc n c^* (p + \neg p \ \tpc n_1) d  &\hint{by lemma $\spc n c^* = \spc n c^* (p + \neg p \tpc n_1)$ above} \\
  = & \spc n c^* \neg p \ \tpc n_1 d        &\hint{distrib, lemma $pd=0$} \\
\leq& \spc n c^* \neg p d^*           &\hint{by $\tpc n_1 \leq  1$ and $d \leq  d^*$}
\end{array}\)
\item 
\(\begin{array}[t]{lll}   
   & \spc n c^* \neg p d^* c \\
 = & \spc n c^* \neg p d^* \neg p c &\hint{by lemma $d: \spec{\neg p}{\neg p}$  and invariance } \\
 = & 0                         &\hint{by $\neg p c = 0$, $c = p c$ (Lemma~\ref{lem:enabmkt})}\\
\leq& \spc n c^* 
\end{array}\)
\item
$\spc n c^* \neg p d^* d \leq  \spc n c^* \neg p d^*$ by $d^* d \leq  d^*$
\end{itemize}
Finally we have $RHS(\dagger) \leq  LHS(\dagger)$ because
\[\begin{array}{lll}   
  & \spc n (c + d)^* \neg (p+q) \\
= 
  & \spc n (c + d)^* \neg p \neg q  & \hint{bool alg} \\
\leq
  & \spc n ( c^* + c^* \neg p d^* ) \neg p \neg q  & \hint{by (i) above} \\
= 
  & \spc n (c^* \neg p + c^* \neg p d^* \neg p) \neg q &\hint{distrib} \\
=  
  & \spc n c^* \neg p d^* \neg p \neg q  
  &\hint{by $c^* \neg p \leq  c^* \neg p d^* \neg p$ (from  $1\leq d^*$ and test idem)} \\
\leq 
  & \spc n c^* \neg p d^* \neg q  & \hint{tests below 1}
  \end{array}\]
This concludes the proof of the sequence case.

\medskip
\textbf{Case $c$ is $\lifgc{n}{e_0\gcto d_0 \gcsep e_1\gcto d_1}$}. 
This special case, with two guarded commands, is more readable than the general case.
So we are given
\[ 
\norm{\lifgc{n}{e_0\gcto c_0 \gcsep e_1\gcto c_1}}{f}{
  \tpc  n \land e_0 \gcto \spc n_0  \gcsep
  \tpc  n \land e_1 \gcto \spc n_1  \gcsep
  gcs_0 \gcsep gcs_1}
\]
where $n_0 = \lab(c_0)$, $n_1=\lab(c_1)$,
$\norm{c_0}{f}{gcs_0}$, and $\norm{c_1}{f}{gcs_1}$.
Let $p_0 \eqdef \enab(gcs_0)$ and $p_1 \eqdef \enab(gcs_1)$.
The proof is:
\[\begin{array}{lll}
  & 
    \mkt{ \addPC (\lifgc{n}{e_0\gcto c_0 \gcsep e_1\gcto c_1}) ; \spc f } \\
 = & \hint{def $\addPC$} \\
   & \mkt{ \spc n; \lifgc{n}{e_0\gcto \addPC(c_0) \gcsep e_1\gcto\addPC(c_1)}; \spc f }  \\
 = & \hint{def $\mkt{-}$} \\
   & \mkt{ \spc n}; ( \mkt{e_0}; \mkt{\addPC(c_0)} + \mkt{e_1};\mkt{\addPC(c_1)} ); \mkt{\spc f} \\
 = & \hint{distribute} \\
   & \mkt{ \spc n}; ( \mkt{e_0}; \mkt{\addPC(c_0)};\mkt{\spc f} + \mkt{e_1};\mkt{\addPC(c_1)};\mkt{\spc f} ) \\ 
 = & \hint{def $\mkt{-}$} \\
   & \mkt{ \spc n}; ( \mkt{e_0}; \mkt{\addPC(c_0);\spc f} + \mkt{e_1};\mkt{\addPC(c_1);\spc f} ) \\ 
 = & \hint{induction hypotheses for $c_0$ and $c_1$} \\
   & \mkt{ \spc n}; ( \mkt{e_0} ; \mkt{\spc n_0; \ldogc{}{gcs_0}} + \mkt{e_1}; \mkt{\spc n_1; \ldogc{}{gcs_1}  } ) \\
 = & \hint{def $\mkt{-}$, defs $p_0$ and $p_1$ } \\
   & \mkt{\spc n}; (\mkt{e_0}; \mkt{\spc n_0}; \mkt{gcs_0}^*; \neg\mkt{p_0} 
                  + \mkt{e_1}; \mkt{\spc n_1}; \mkt{gcs_1}^*; \neg\mkt{p_1} ) \\
 = & \hint{step $(+)$, see below} \\
   & \mkt{ \spc n}; (\mkt{\tpc n};\mkt{e_0};\mkt{\spc n_0 } 
                     + \mkt{\tpc n};\mkt{e_1};\mkt{\spc n_1} 
                     + \mkt{gcs_0} + \mkt{gcs_1})^* ;  \\
   & \qquad \neg (\mkt{\tpc n};\mkt{e_0} + \mkt{\tpc n};\mkt{e_1} + \mkt{p_0} +\mkt{p_1}) \\
 = &\hint{def $\mkt{-}$} \\
   & \mkt{ \spc n}; (\mkt{\tpc n\land e_0}; \mkt{\spc n_0 } + \mkt{\tpc n\land e_1}; \mkt{\spc n_1} 
                    + \mkt{gcs_0} + \mkt{gcs_1})^* ; 
                       \neg \mkt{(\tpc n\land e_0)\lor(\tpc n\land e_1)\lor p_0 \lor p_1} \\
 = &\hint{def $\mkt{-}$, defs $p_0$ and $p_1$} \\
& \mkt{ \spc n; \ldogc{}{ \tpc n\land e_0 \gcto \spc n_0 \gcsep
                         \tpc n\land e_1 \gcto \spc n_1 \gcsep gcs_0 \gcsep gcs_1 } } 
\end{array}\]
The step marked $(+)$  uses KAT laws and the normal form axioms
(including consequences Lemma~\ref{lem:nf-conseq}).  We omit it,
because it is similar to, but simpler than, the proof for the sequence case.

For the general case, i.e., if-command with more than two guarded commands, or just one, is
essentially the same, and omitted. 
(An if with one guarded command is allowed but pointless, since by $\totalIf$ the guard must be true.)

\medskip
\textbf{Case $c$ is $\ldogc{n}{e\gcto d}$}.  We consider the special case
of a single guarded command because it is more readable than the general case, and all the interesting ingredients are the same as in the general case.  
So we have 
\[ \norm{\ldogc{n}{e\gcto d}}{f}{
  \tpc  n\land e \gcto \spc n_0  \gcsep
  \tpc  n\land\neg e \gcto \spc f \gcsep gcs } \]
where $\normSmall{d}{n}{gcs}$ and $m := \lab(d)$.
After unfolding definitions and applying the induction hypothesis, as in the sequence and if cases, it comes down to this proof goal:
\[ 
\begin{array}{ll} % NOTE reversed from what's in Norm.v
% 
 %!n (?n e0 !n0 + ?n ~e0 !f + gcs0)^* ~(?n + enab gcs0) 
& \mkt{\spc n};(\mkt{\tpc n};\mkt{e};\mkt{\spc m}
             + \mkt{\tpc n};\neg\mkt{e};\mkt{\spc f}
             + \mkt{gcs})^* ; \neg(\mkt{\tpc n}+\mkt{\enab(gcs)}) \\
= &
%!n (e0 !n0 gcs0^* ∽(enab gcs0))^* ∽e0 !f
 \mkt{\spc n}; (\mkt{e};\mkt{\spc
  m};\mkt{gcs}^*; \neg\mkt{\enab(gcs)})^*; \neg\mkt{e};\mkt{\spc f}
\end{array}
\]
It relates the main normal form (LHS) to a loop containing the body's normal form (RHS, note the nested star).
For brevity we again use juxtaposition for sequence, and set some abbreviations, re-using identifiers since their original bindings are no longer relevant.
\[\begin{array}{l}
e := \mkt{e} \\
c := \mkt{gcs} \\
p := \mkt{\enab(gcs)} \\
b := \mkt{\tpc n};\mkt{e};\mkt{\spc m} + \mkt{\tpc n};\neg\mkt{e};\mkt{\spc f} + \mkt{gcs}   
\end{array}\]
Note that $b$ is the body of the main normal form.
Now the above proof goal can be written
\[ (\ddagger) \qquad
\spc n (e \,\spc m c^* \neg p)^* \neg e \,\spc f = \spc n b^* \neg(\tpc n + p)
\]
We have these lemmas as consequence of the axioms (see similar lemmas and justifications in the preceding case for sequence).
\[
c = p c \qquad
c :\spec{p}{(p + \tpc n)} \qquad
p \, \tpc f = 0 \qquad
\tpc n\, \tpc m = 0 \qquad
\tpc m \leq  p
\]
The crux of the proof is the following expansion lemma (proved later):
\[ (\star) \qquad % Lemma Y: 
\tpc n b^* =   \tpc n (\tpc n e \, \spc m c^* \neg p)^* \tpc n 
         + \tpc n (\tpc n e \,\spc m c^* \neg p)^* \tpc n e \,\spc m c^* \neg \tpc n p
         + \tpc n (\tpc n e \,\spc m c^* \neg p)^* \tpc n \neg e \,\spc f
\]
The summands represent completed iterations of the original loop, 
completed iterations followed by an incomplete one, and completed
iterations followed by termination.  (These are not mutually exclusive 
and we do not need that.)

Using $(\star)$ we prove $(\ddagger)$ as follows:
\[\begin{array}{lll}
  & RHS(\ddagger) \\ 
= & \hint{(setTest) and de Morgan} \\
  & \spc n \tpc n b^* \neg \tpc n \neg p \\            
= & \hint{expansion lemma $(\star)$, distrib} \\
  & \spc n \tpc n (\tpc n e \spc m c^* \neg p)^* \tpc n  \neg \tpc n \neg p            
 + \spc n \tpc n (\tpc n e \spc m c^* \neg p)^* \tpc n e \spc m c^* \neg \tpc n p \neg \tpc n \neg p            
 + \spc n \tpc n (\tpc n e \spc m c^* \neg p)^* \tpc n \neg e \spc f \neg \tpc n \neg p            
\\ 
= & \hint{cancel first two terms using $\tpc n \neg \tpc n = 0$ 
          and $p \neg p = 0$}\\
  & \spc n \tpc n (\tpc n e \spc m c^* \neg p)^* \tpc n \neg e \spc f \neg \tpc n \neg p \\           
= &\hint{using $\spc f = \spc f\tpc f$, $\tpc f \leq  \neg \tpc n$, $\tpc f \leq  \neg p$, (setTest),
         $c :\spec{p}{(p + \tpc n)}$ and invariance law }\\
  & \spc n (e \spc m c^* \neg p)^* \neg e \spc f \\
= & \\
  & LHS(\ddagger)
\end{array}\]
It remains to prove Lemma $(\star)$, by mutual inequality.

We have $RHS(\star) \leq  LHS(\star)$ because each term of RHS is 
$\leq  LHS(\star)$, as follows.
\begin{itemize}
\item
\(\begin{array}[t]{lll}
     &  \tpc n (\tpc n e \spc m c^* \neg p)^* \tpc n  \\
\leq 
     &  \tpc n ( \tpc n e \spc m c^* )^* 
     &\hint{using tests below 1 twice} \\
\leq 
     & \tpc n ( b b^* )^* 
     &\hint{using $\tpc n e \spc m \leq  b$ and $c^* \leq  b^*$ } \\
\leq &  \tpc n b^* &\hint{star laws}
     \end{array}\)
\item
\(\begin{array}[t]{lll}
    & \tpc n (\tpc n e \spc m c^* \neg p)^* \tpc n e \spc m c^* \neg \tpc n p \\
\leq
    &\tpc n b^* \tpc n e \spc m c^* \neg \tpc n p 
    &\hint{using first bullet}\\
\leq
    &\tpc n b^* \tpc n e \spc m c^* 
    &\hint{tests below 1}\\
\leq
    &\tpc n b^* b b^* 
    &\hint{using $\tpc n e \spc m \leq  b$ and $c^* \leq  b^*$ } \\
\leq
    &\tpc n b^* &\hint{star laws}
    \end{array}\)
\item
\(\begin{array}[t]{lll}
    & \tpc n (\tpc n e \spc m c^* \neg p)^* \tpc n \neg e \spc f \\
\leq
    &\tpc n b^* \tpc n \neg e \spc f 
    &\hint{using first bullet}\\
\leq
    &\tpc n b^* b &\hint{using $\tpc n \neg e \spc f \leq  b$} \\
\leq
    &\tpc n b^* &\hint{star law}\\
    \end{array}\)
\end{itemize}
For the reverse, $LHS(\star) \leq  RHS(\star)$ follows by *-induction and 
join property from the following.

The base: 
\[ \tpc n \leq  \tpc n 1 \tpc n \leq  \tpc n (\tpc n e \spc m c^* \neg p)^* \tpc n \leq  RHS(\star) \]

The induction step is $RHS(\star); b \leq  RHS(\star)$. 
Observe that, by def $(\star)$ and distribution,
\[\begin{array}{lll}
RHS(\star); b & = & 
   \begin{array}[t]{ll}
    \tpc n (\tpc n e \spc m c^* \neg p)^* \tpc n b             
    &  (i)\\
   + \tpc n (\tpc n e \spc m c^* \neg p)^* \tpc n e \spc m c^* \neg \tpc n p b    
   & (ii) \\
   + \tpc n (\tpc n e \spc m c^* \neg p)^* \tpc n \neg e \spc f  b  
   & (iii)
   \end{array}
   \end{array}\]
So we show each term is below $RHS(\star)$.
\begin{list}{}{}
\item[(i)]
\(\begin{array}[t]{lll}
   &  \tpc n (\tpc n e \spc m c^* \neg p)^* \tpc n b \\
=  &\hint{def $b$, and $c = p c$}\\
   & \tpc n (\tpc n e \spc m c^* \neg p)^* \tpc n (\tpc n e \spc m + \tpc n \neg e \spc f + p c)   \\
=  &\hint{distrib, cancel by $\tpc n p = 0$, test idem} \\
   & \tpc n (\tpc n e \spc m c^* \neg p)^* \tpc n e \spc m  
  + \tpc n (\tpc n e \spc m c^* \neg p)^* \tpc n \neg e \spc f 
  \end{array}\)
\\
The second summand is a term of $RHS(\star)$. For the first summand continue
\\
\(\begin{array}{lll}
  & \tpc n (\tpc n e \spc m c^* \neg p)^* \tpc n e \spc m \\
= & \hint{(setTest)} \\
  &\tpc n (\tpc n e \spc m c^* \neg p)^* \tpc n e \spc m \tpc m \\
\leq&\hint{using $1 \leq  c^*$ and $\tpc m \leq  \neg \tpc n$ and $\tpc m \leq  p$} \\
  & \tpc n (\tpc n e \spc m c^* \neg p)^* \tpc n e \spc m c^* \neg \tpc n p
  \end{array}\)\\
and this is in $RHS(\star)$.
\item[(ii)]
\(\begin{array}[t]{lll}
  & \tpc n (\tpc n e \spc m c^* \neg p)^* \tpc n e \spc m c^* \neg \tpc n p b \\
= &\hint{def b} \\
  &\tpc n (\tpc n e \spc m c^* \neg p)^* \tpc n e \spc m c^* \neg \tpc n p (\tpc n e \spc m + \tpc n \neg e \spc f + c)   \\
= &\hint{distrib, cancel by $\neg \tpc n \tpc n = 0$} \\
  &\tpc n (\tpc n e \spc m c^* \neg p)^* \tpc n e \spc m c^* \neg \tpc n p c   \\
= &\hint{by  $c :\spec{p}{p+\tpc n}$} \\
  &\tpc n (\tpc n e \spc m c^* \neg p)^* \tpc n e \spc m c^* \neg \tpc n p c (p + \tpc n)\\
= &\hint{distrib} \\
  &  \tpc n (\tpc n e \spc m c^* \neg p)^* \tpc n e \spc m c^* \neg \tpc n p c p 
  + \tpc n (\tpc n e \spc m c^* \neg p)^* \tpc n e \spc m c^* \neg \tpc n p c \tpc n
  \end{array}\)\\
For the first summand,\\
\(\begin{array}[t]{lll}
    &  \tpc n (\tpc n e \spc m c^* \neg p)^* \tpc n e \spc m c^* \neg \tpc n p c p \\
\leq&\hint{test below 1} \\
    &  \tpc n (\tpc n e \spc m c^* \neg p)^* \tpc n e \spc m c^* c p \\
\leq&\hint{star fold, $p \leq  \neg \tpc n p$  (from $\tpc n p = 0$)} \\
    & \tpc n (\tpc n e \spc m c^* \neg p)^* \tpc n e \spc m c^* \neg \tpc n p 
    \end{array}\)\\
and this is in $RHS(\star)$.
For the second summand,\\
\(\begin{array}[t]{lll}
    & \tpc n (\tpc n e \spc m c^* \neg p)^* \tpc n e \spc m c^* \neg \tpc n p c \tpc n\\
\leq&\hint{test below 1} \\
    &\tpc n (\tpc n e \spc m c^* \neg p)^* \tpc n e \spc m c^* c \tpc n \\
\leq&\hint{using $\tpc n \leq  \neg p$ and star fold}\\
    &\tpc n (\tpc n e \spc m c^* \neg p)^* (\tpc n e \spc m c^* \neg p) \\
\leq&\hint{star fold} \\
    & \tpc n (\tpc n e \spc m c^* \neg p)^* 
    \end{array}\)\\
and this is in $RHS(\star)$.
\item[(iii)]
\(\begin{array}[t]{lll}
  &  \tpc n (\tpc n e \spc m c^* \neg p)^* \tpc n \neg e \spc f  b  \\
= &\hint{def b}\\
  & \tpc n (\tpc n e \spc m c^* \neg p)^* \tpc n \neg e \spc f (\tpc n e \spc m + \tpc n \neg e \spc f + c)   \\
= &\hint{distrib, cancel by $\tpc f \tpc n =0$ and $c = p c$ and $\tpc f p = 0$}\\
  & 0 \\
\leq& \\
    &  RHS(\star)
    \end{array}\)
\end{list}
\end{proof}

\subsection{On alternative presentations of command equivalence}\label{sec:alt}

HL+ is a logic of correctness judgments $c:\spec{P}{Q}$ and RHL+ is a
logic of relational correctness judgments $c\sep d:\rspec{\P}{\Q}$.
As usual with Hoare logic, both logics rely on entailment of unary and relational assertions.
We write the entailments as $P\imp Q$ in \rn{Conseq} and $\P\imp\Q$ in \rn{rConseq} but interpret them to mean set containment in the metalanguage,
a common technique to dodge the well understood issue of a proof system for entailment~\cite{Cook78}.
What is unusual is that the logics also rely on command equivalence ($\kateq$).
Our formalization of command equivalence (Def.~\ref{def:kateq})
is meant to be conceptually clear and in particular to highlight
what is needed to achieve alignment completeness, while abstracting from some details
that can be handled in more than one way.  
This section addresses two issues about the equivalence judgment for which some readers
may have concerns.  Other readers should skip to Sect.~\ref{sec:autUnary}.

The first concern is that one might like a deductive system for proving equivalences.
We sketch one way such a system could be formulated.  
First, include equality of KAT expressions as another judgment of the system
(and treat $KE_0\leq KE_1$ as sugar for $KE_0+KE_1 = KE_1$).
Second, include a single rule for introduction of command equivalence:
\[ \inferrule*[left=KATeq]{ \mkt{c} = \mkt{d} }{ c \kateq d } \]
Third, include schematic axioms and rules for KAT expressions.
For example: 
\begin{mathpar}
\inferrule{ }{ KE_0 \kplus KE_1 = KE_1\kplus KE_0 }

   \inferrule{ KE_1 \kplus KE_0 \kdot KE_2 \leq KE_2 }
             { KE_0^\kstar\kdot KE_1 \leq KE_2 }
\end{mathpar}
which correspond to $x+y=y+x$ and $y \kplus x\kdot z \leq z \imp x^\kstar\kdot y \leq z$
in Def.~\ref{def:KAT}.
Finally, include these two axiom schema (with side conditions) for the
equations of $\Hyp$ (Def.~\ref{def:Hyp}):
\begin{mathpar}
\inferrule{e\imp\mathit{false}}{ \mkt{e}=0 }

\inferrule{e_0\imp \subst{e_1}{x}{e}}{ \mkt{e_0};\mkt{x:=e} = \mkt{e_0};\mkt{x:=e};\mkt{e_1} }
\end{mathpar}
The side conditions use validity of assertions, which are already needed 
for the \rn{Conseq} and \rn{rConseq} rules.
% Owing to shallow embedding of assertions, the antecedent
% $e\imp\mathit{false}$ amounts to $\means{e}=\emptyset$ and the antecedent 
% $e_0\imp \subst{e_1}{x}{e}$ amounts to $\models\lassg{}{x}{e} :\spec{e_0}{e_1}$.
With this approach, any command equivalence needed for an instance of 
\rn{Rewrite} or \rn{rRewrite}, can be established using rule \rn{KATeq}
together with the preceding rules.

A second concern is that the proof systems HL+ and RHL+ are about commands and specs,
yet they include KAT expressions in a subsidiary role.  
It would seem more parsimonious to formulate command laws directly.
In the current setup, a simple general law like $c;\skipc \kateq c$ 
is obtained from $\mkt{c;\skipc} = \mkt{c}$
which in turn holds because $\mkt{c;\skipc}$ is $\mkt{c}\kdot\kone$
and $\mkt{c}\kdot\kone = \mkt{c}$ is an instance of a KAT axiom.
Several of the KAT axioms can be formulated directly as command equalities,
as in the case of $c;\skipc \kateq c$,
and rules of congruence, symmetry, reflexivity, and transitivity can be 
stated directly for $\kateq$.

This raises the question whether we can formulate the KAT structure entirely in terms of commands. 
One approach is to add assume statements to the language of commands. 
These take the form $\keyw{assume}~e$, for program expression $e$ (we do not need assumptions for predicates).  The idea is that $\keyw{assume}~e$ will play the role currently played by translated boolean expressions $\mkt{e}$.
But it is not clear that KAT laws like $\kone\kplus x\kdot x^\kstar = x^\kstar$
can be expressed with the more limited GCL syntax.

There is another approach to reducing the number of syntactic forms.
One can take the KAT expressions of (\ref{eq:KATexp}) as the core syntax
and treat guarded commands as syntax sugar in unary and relational correctness judgments
and proof rules.   
This approach is taken, for example, in O'Hearn's 
presentation of incorrectness logic~\cite{OHearn2019}.

\section{Floyd completeness}\label{sec:autUnary}

In this section we put the normal form equivalence theorem to work
in a proof that an IAM proof can be translated to one in HL+.
This gives a way to prove completeness of HL+ in the sense of Cook,
as spelled out in Corollary~\ref{cor:CookComplete}.
It also provides a pattern that guides the proof of alignment completeness, Theorem~\ref{thm:acomplete}.

\begin{theorem}[Floyd completeness]\label{thm:FloydComplete}
\upshape
Suppose $\okf(c,f)$.
Suppose $an$ is a valid annotation of $\aut(c,f)$ for $\spec{P}{Q}$.
Then $c: \spec{P}{Q}$ can be proved in HL+, using only assertions derived from $an$.
\end{theorem}
The proof will use only a single instance of the \rn{Do} rule and no instance of the \rn{If} rule.
More importantly, the judgments use only assertions derived from those of $an$ in simple ways.
In particular, we use boolean combinations of the following:
assertions $an(i)$, boolean expressions that occur in $c$,
and equality tests $pc=n$ of the program counter variable and numeric literals;
and we use substitution instances $\subst{an(i)}{x}{e}$ for assignments
$x:=e$ that occur in $c$.

Nagasamudram and Naumann~\cite{NagasamudramN21} introduce the term Floyd completeness 
and use it for a similar result where they refer to ``associated judgments'' which are not precisely defined but rather highlighted in the theorem's proof.  Inspection of their proof reveals that their associated judgments are built from the annotation by substitution (for assignments) and conjoining boolean expressions (and their negations) that occur in the program.  
The practical importance of both Floyd completeness results is that the HL proof does not require more expressive assertions than used in the IAM proof.  
Compared with~\cite{NagasamudramN21}, our Theorem~\ref{thm:FloydComplete} has two
differences:  their judgments are all about subprograms of $c$ whereas ours involves the derived normal form program,
and for that program we build a loop invariant using tests of the $pc$ 
together with the annotation's assertions.  

Following the semi-formal treatment of ``associated judgments'' in~\cite{NagasamudramN21}, we do not precisely define the ``assertions derived from $an$'' mentioned in Theorem~\ref{thm:FloydComplete}, but instead spell out the proof in detail so the interested reader can reconstruct them.

\begin{proof}
Choose variable $pc$ that is fresh with respect to $c$ and $an(i)$ for all $i$ (hence with respect to $P$ and $Q$).
By Lemma~\ref{lem:normExists} we have some $gcs$
such that $\normSmall{c}{f}{gcs}$.
By Theorem~\ref{thm:normEquiv} we have
\begin{equation}\label{eq:nfEqu}
\spc n; \ldogc{}{gcs} \: \kateq \: \addPC(c); \spc f 
\quad\mbox{where $n=\lab(c)$.}
\end{equation}
For a loop invariant to reason about the normal form, with an eye on
the example we might try this formula:
\[ 1 \leq pc \leq f \land (\land i : 0\leq i\leq f : pc=i \imp an(i)) \]
But this only makes sense if the labels form a contiguous sequence, which
we do not require.  There is no need to reason arithmetically about labels.
We define the invariant $I$ as follows:
\[ I: \quad
(\lor i : i\in\labs(c)\union\{f\} : \tpc i) \land 
(\land i : i\in\labs(c)\union\{f\} : \tpc i \imp an(i)) 
\]
The disjunction says the current value of $pc$ is in $\labs(c)\union\{f\}$.

The next step is to obtain proofs of 
\begin{equation}\label{eq:DOprem}
b:\spec{I\land e}{I} \qquad \mbox{for each $e\gcto b$ in $gcs$ }
\end{equation}
To do so, first note that we have for any $m$ that 
\begin{equation}\label{eq:set}
\spc m : \spec{an(m)}{an(m)\land \tpc m} 
\end{equation}
using rules \rn{Asgn} and \rn{Conseq}, because by freshness $pc$ is not in $an(m)$.
Second, note that by definition of $I$ we have valid implications
\begin{equation}\label{eq:Ian}
I\land \tpc n \imp an(n) \qquad\mbox{and}\qquad an(n)\land \tpc n \imp I 
\qquad\mbox{for any $n$ in $\labs(c)\union\{f\}$}
\end{equation}
Now go by the possible cases of $b$ in (\ref{eq:DOprem}),
which are given by Lemma~\ref{lem:nfCases}.
\begin{itemize}
\item 
$b$ has the ``set nf'' form 
$\tpc n\gcto \lassg{}{x}{e};\spc m$, where $\sub(n,c)$ is $\lassg{n}{x}{e}$ 
and $m = \fsuc(n,c,f)$.

To show: $\lassg{}{x}{e};\spc m : \spec{I\land \tpc n}{I}$.
%% By the VC in Fig.~\ref{fig:VC} we have
%% $\models \lassg{}{x}{e} : \spec{an(n)}{an(m)}$,
%% whence by Lemma~\ref{lem:compAss} we have $\lassg{}{x}{e} : \spec{an(n)}{an(m)}$.
By rule \rn{Asgn} we have 
$\lassg{}{x}{e} : \spec{\subst{an(m)}{x}{e}}{an(m)}$.
(We are not writing explicit $\proves$ for provability of correctness judgments.)
By the VC in Fig.~\ref{fig:VC} we have
$an(n)\imp \subst{an(m)}{x}{e}$,
so using \rn{Conseq} we get
$\lassg{}{x}{e} : \spec{an(n)}{an(m)}$.
By fact (\ref{eq:set}) we have
$\spc m : \spec{an(m)}{an(m)\land \tpc m}$, so
by rule \rn{Seq} we have $\lassg{}{x}{e};\spc m : \spec{an(n)}{an(m)\land \tpc m}$.
So by \rn{Conseq} using both implications in fact (\ref{eq:Ian}) we get 
$\lassg{}{x}{e};\spc m : \spec{I\land \tpc n}{I}$.
\item 
$b$ has the ``test nf'' form $\tpc n\gcto \spc m$ (i.e., 
$\tpc n\land true\gcto \spc m$),
where $\sub(n,c)$ is $\lskipc{n}$ and $m=\fsuc(n,c,f)$.

To show: $\spc m : \spec{I\land \tpc n}{I}$.
By (\ref{eq:set}) we have
$\spc m : \spec{an(m)}{an(m)\land \tpc m}$.
By VC for skip (Fig.~\ref{fig:VC}) we have
$an(n)\imp an(m)$ so using \rn{Conseq} we get 
$\spc m : \spec{an(n)}{an(m)\land \tpc m}$.
So by \rn{Conseq} using fact (\ref{eq:Ian}) we get 
$\spc m : \spec{I\land \tpc n}{I}$.
\item 
$b$ has the ``test nf'' form $\tpc n\land e\gcto \spc m$ where $\sub(n,c)$ is 
$\lifgc{n}{gcs_0}$ for some $gcs_0$, and $m=\lab(d)$ where $e\gcto d$ is in $gcs_0$.

To show: $\spc m : \spec{I \land \tpc n \land e}{I}$.
We have $\spc m : \spec{an(m)}{an(m)\land \tpc m}$ by (\ref{eq:set}).
We have $an(n)\land e \imp an(m)$ by VC for \keyw{if},
so by \rn{Conseq} we get $\spc m : \spec{an(n)\land e }{an(m)\land \tpc m}$.
So by \rn{Conseq} using fact (\ref{eq:Ian}) we get 
$\spc m : \spec{I\land \tpc n\land e}{I}$.
\item 
$b$ has the ``test nf'' form $\tpc n\land e\gcto \spc m$ where $\sub(n,c)$ is 
$\ldogc{n}{gcs_0}$, for some $gcs_0$, and $m=\lab(d)$ where $e\gcto d$ is in $gcs_0$.
This is just like the preceding case, 
using the corresponding VC for \keyw{do}.
\item 
$b$ has the ``test nf'' form $\tpc n\land \neg\enab(gcs_0)\gcto \spc m$ where $\sub(n,c)$ is 
$\ldogc{n}{gcs_0}$ and $m=\fsuc(n,c,f)$.

To show: $\spc m : \spec{I \land \tpc n \land \neg\enab(gcs_0)}{I}$.
Again, the proof is like the preceding case, using the VC for loop exit.
\end{itemize}
Having established the premises of rule \rn{Do}, we get its conclusion:
\[  \ldogc{}{gcs} : \spec{I}{I\land \neg\enab(gcs)} \]
By Lemma~\ref{lem:nf-enab-labs} and definition of $I$,
$I\land \neg\enab(gcs)$ is equivalent to $I\land \tpc f$, 
so by consequence we get 
\[ \ldogc{}{gcs} : \spec{I \land \tpc n}{I\land \tpc f} \] 
Now $I\land \tpc n$ is equivalent to $an(n)\land \tpc n$
which is $P\land \tpc n$ because $an$ is an annotation for the spec $\spec{P}{Q}$.
Likewise, the postcondition implies $Q$.  So by consequence we get 
\[ \ldogc{}{gcs} : \spec{P \land \tpc n}{Q} \] 
By the assignment rule and consequence using that $pc$ is fresh for $P$ we get
$\spc n:\spec{P}{P\land\tpc n}$,
so using the sequence rule we get
\[ \spc n; \ldogc{}{gcs} : \spec{P}{Q} \]
Then \rn{Rewrite} using (\ref{eq:nfEqu}) yields 
\[ \addPC(c); \spc f : \spec{P}{Q} \]
By Lemma~\ref{lem:addPC} and freshness of $pc$ we have that
$pc$ is ghost in $\addPC(c); \spc f$, also $P$ and $Q$ are independent from $pc$.
So by rule \rn{Ghost} we get 
\[ \erase(pc,\addPC(c);\spc f) : \spec{P}{Q} \] 
Now using Lemma~\ref{lem:addPC} together with the general law $c;\skipc\kateq c$
and transitivity of $\kateq$,\footnote{Both of which are easily
derived using the definition of $\kateq$.} 
we have that $\erase(pc,\addPC(c); \spc f) \kateq c$,
so by \rn{Rewrite}
we get $c: \spec{P}{Q}$.
\end{proof}

\section{RHL+ is alignment complete}\label{sec:acomplete}

%% By an argument similar to the proof of Lemma~\ref{lem:compAss},
%% we have the following.

%% \begin{lemma}[completeness for assignment]\label{lem:compAssR}
%% \upshape
%% For any valid judgment 
%% $\models x:=e\sep x':=e' : \rspec{\R}{\S}$,
%% $\models x:=e\sep \skipc : \rspec{\R}{\S}$,
%% and 
%% $\models \skipc\sep x':=e' : \rspec{\R}{\S}$,
%% there is a proof.
%% \end{lemma}
%% \begin{proof} (Sketch)
%% For, $x:=e\sep \skipc : \rspec{\R}{\S}$, first show that 
%% substitution gives the weakest precondition:
%% $\models x:=e\sep \skipc : \rspec{\R}{\S}$ iff $\R\imp\subst{\S}{x|}{e|}$.
%% By rule \rn{AsgnSkip} we get
%% $x:=e\sep\skipc : \rspec{\subst{\S}{x|}{e|}}{\S}$
%% and then using \rn{rConseq} we get $x:=e\sep \skipc : \rspec{\R}{\S}$.

%% For $x:=e\sep x':=e'$ and $\skipc\sep x':=e' $ the argument is very similar but using 
%% \rn{rAsgn} with $\subst{\S}{x|x'}{e|e'}$, 
%% and \rn{SkipAsgn} with $\subst{\S}{|x'}{|e'}$, 
%% respectively.
%% \end{proof}

%% \begin{lemma}\label{lem:LRJenab}
%% \upshape
%% Suppose $\ok(c)$, $f\notin\labs(c)$, $\ok(c')$, $f'\notin\labs(c')$,
%% and consider any alignment automaton $\aprod(\aut(c,f),\aut(c',f'),L,R,J)$.  
%% We have $J^\nbl = J\land\neg [\fin|\fin']$,
%% $L^\nbl = L\land\mbox{``control on left is not at $f$''}$,
%% and 
%% $R^\nbl = R\land\mbox{``control on right is not at $f'$''}$.
%% \end{lemma}
%% \begin{proof}
%% By Def.~\ref{def:manifest} and Lemma~\ref{lem:autLive}.
%% \end{proof}

Theorem~\ref{thm:FloydComplete} says that a valid annotation of the automaton for a command gives rise to a HL+ proof of its correctness.
The main result of the paper is that a valid annotation of an alignment automaton for a pair of commands gives rise to a RHL+ proof,
provided adequacy is proved in the manner of Lemma~\ref{lem:anLRJF}.

\begin{theorem}\label{thm:acomplete}
\upshape
Suppose 
\begin{list}{}{}
\item[(a)] $\okft(c,f)$ and $\okft(c',f')$
\item[(b)] $an$ is a valid annotation of $\aprod(\aut(c,f),\aut(c',f'),L,R,J)$
for $\rspec{\S}{\T}$.  
\item[(c)] $\hat{an}(i,j)\imp L \lor R \lor J \lor [\fin|\fin']$ for all control points $(i,j)$ of 
$\aprod(\aut(c,f),\aut(c',f'),L,R,J)$
\end{list}
Then the judgment $c\sep c': \rspec{\S}{\T}$ has a proof in RHL+.
\end{theorem}
As in Theorem~\ref{thm:FloydComplete}, the proof we construct uses only relational assertions derived in simple ways from assertions $an(i,j)$, specifically: boolean combinations 
of the annotation's assertions, conjunctions with conditional tests in the code, and substitutions for expressions in assignment commands.

Conditions (b) and (c) imply that
$\aprod(\aut(c,f),\aut(c',f'),L,R,J)$ is manifestly $\S$-adequate (using Lemma~\ref{lem:anLRJF}).
So by Corollary~\ref{cor:relIAM} we get $\aut(c,f),\aut(c',f') \models \rspec{\S}{\T}$
and so by Lemma~\ref{lem:autConsistent} we get $\models c\sep c' : \rspec{\S}{\T}$.
By a different path, we get $\models c\sep c' : \rspec{\S}{\T}$ 
by the conclusion of Theorem~\ref{thm:acomplete} together with 
soundness of RHL+ (Prop.~\ref{prop:RHLsound}).

Restriction (a) in Theorem~\ref{thm:acomplete} is just a technicality.
The $\okf$ condition just says labels of $c$ are unique and do not include $f$. 
Labels have no effect on program semantics so they can always be chosen to satisfy the condition.

\begin{proof}
Suppose $\S,\T,c,c',f,f',L,R,J$ and $an$ satisfy the hypotheses (a)--(c) of the theorem.
Choose variable $pc$ that is fresh with respect to $\S,\T,c,c',an,L,R,J$.
To be precise: $\indep(pc|pc,\S)$, $\indep(pc|pc,\T)$,  $pc$ does not occur in $c$ or $c'$,
and $L$, $R$, $J$ are independent from $pc$ on both sides (see (\ref{eq:indep}).

By Lemma~\ref{lem:normExists} there are $gcs$ and $gcs'$ such that 
$\normSmall{c}{f}{gcs}$ and $\normSmall{c'}{f'}{gcs'}$.
Let $n=\lab(c)$ and $n'=\lab(c')$.
By Theorem~\ref{thm:normEquiv} we have
\begin{equation}\label{eq:nfEquLR}
\begin{array}{l}
%\nfax(pc,c,f)\proves 
\spc n; \ldogc{}{gcs} \: \kateq \: \addPC(c); \spc f 
\\
%\nfax(pc,c',f')\proves 
\spc n'; \ldogc{}{gcs'} \: \kateq \: \addPC(c'); \spc f' 
\end{array}
\end{equation}
Define store relation $\Q$ to be $\Q_{an}\land\Q_{pc}$ where 
\[ \begin{array}{l}
\Q_{an}: \qquad
   \quant{\land}{i,j}{i\in \labs(c)\union\{f\} \land j\in \labs(c')\union\{f'\}
   }{\bothF{\tpc i\sep \tpc j} \imp an(i,j)} 
\\
\Q_{pc}: \qquad 
   \quant{\lor}{i,j}{i\in \labs(c)\union\{f\} \land j\in \labs(c')\union\{f'\}
   }{\bothF{\tpc i\sep \tpc j}}
   \end{array}
\]
We will derive 
\begin{equation}\label{eq:C}
 \ldogc{}{gcs} \sep \ldogc{}{gcs'} : \rspec{\Q}{\Q\land\neg\leftF{\enab(gcs)}\land\neg\rightF{\enab(gcs')}}  
 \end{equation}
using rule \rn{dDo} instantiated with $\Q:=\Q$, $\Lrel:=\encode{L}$, and $\R:=\encode{R}$.
The side condition of \rn{dDo} is
\begin{equation}\label{eq:side} \Q \imp 
(\leftex{\enab(gcs)} = \rightex{\enab(gcs')})
          \lor (\encode{L} \land \leftF{\enab(gcs)})
          \lor (\encode{R} \land \rightF{\enab(gcs')})
\end{equation}
To prove (\ref{eq:side}), first rewrite $\Q$ using distributivity and renaming dummies, to the equivalent form
\[
   \quant{\lor}{i,i'}{}{ 
     \bothF{\tpc i\sep \tpc i'} \land 
     \quant{\land}{k,k'}{}{ 
         (\bothF{\tpc k\sep\tpc k'} \imp an(k,k')) }}
\]
where we omit that $i,k$ range over $\labs(c)\union\{f\}$
and $i',k'$ range over $\labs(c')\union\{f'\}$.
This implies 
\[
   \quant{\lor}{i,i'}{}{  \bothF{\tpc i\sep\tpc i'} \land an(i,i'       ) }
\] 
Thus by (c), any $\Q$-state satisfies
$\encode{L} \lor \encode{R} \lor \encode{J} \lor \bothF{\tpc f\sep\tpc f'}$.
We show each of these disjuncts the right side of (\ref{eq:side}).
\begin{itemize}
\item $\encode{L}$ implies $\encode{L}\land\leftF{\enab(gcs)}$ 
because (i) $L$ is a set of states of the alignment automaton,
with control on the left ranging over $\labs(c)\union\{f\}$,
(ii) by liveness (Assumption~1), $L$ allows transitions 
by $\aut(c,f)$, and so excludes control being at $f$,
and (iii) $\leftF{\enab(gcs)}$ means control on the left is in $\labs(c)$,
             by Lemma~\ref{lem:nf-enab-labs}.
\item $\encode{R}$ implies $\encode{R}\land\rightF{\enab(gcs')}$
for reasons symmetric to the $L$ case
\item $\encode{J}$ implies $\leftF{\enab(gcs)}$ and $\rightF{\enab(gcs')}$ are both true
(using liveness and Lemma~\ref{lem:nf-enab-labs} again),
so $\leftF{\enab(gcs)} = \rightF{\enab(gcs')}$
\item $\bothF{\tpc f\sep\tpc f'}$ implies 
both $\leftF{\enab(gcs)}$ and $\rightF{\enab(gcs')}$ are false (again using  Lemma~\ref{lem:nf-enab-labs}) 
so $\leftF{\enab(gcs)} = \rightF{\enab(gcs')}$
\end{itemize}
So the side condition (\ref{eq:side}) of \rn{dDo} is proved.
Before proceeding to prove the premises for \rn{dDo}, note that we can prove
\begin{equation}\label{eq:setR}
\spc m \sep \spc m' : \rspec{\P}{\P\land \leftF{\tpc m}\land\rightF{\tpc m'}} 
\end{equation}
for any $\P$ in which $pc$ does not occur, and any $m,m'$,
using rules \rn{rAsgn} and \rn{rConseq}.
Also, by definition of $\Q$ we have valid implications 
\begin{equation}\label{eq:IanRel}
\begin{array}[t]{l} 
\Q\land \bothF{\tpc m\sep\tpc m'} \imp an(m,m') 
\qquad\mbox{and}\qquad 
an(m,m') \land \bothF{\tpc m\sep\tpc m'} \imp \Q \\
\mbox{for any $m$ in $\labs(c)\union\{f\}$ and 
$m'$ in $\labs(c')\union\{f'\}$} 
\end{array}
\end{equation}
From condition (c) of the theorem, by definitions we get 
\[ %begin{equation}\label{eq:dlift}
an(i,j)\land\bothF{\tpc i\sep\tpc j}\imp\encode{L}\lor\encode{R}\lor\encode{J}\lor\bothF{\tpc \fin\sep\tpc\fin'}
\quad\mbox{for all $i,j$}
\]
and hence
\begin{equation}\label{eq:dliftJ}
an(i,j)\land\bothF{\tpc i\sep\tpc j}\land\neg\encode{L}\land\neg\encode{R}\imp \encode{J}
\quad\mbox{for all $i,j$ with $(i,j)\neq(\fin,\fin')$}
\end{equation}

There are three sets of premises of \rn{dDo} for the loops in (\ref{eq:C}), with these forms:
\begin{list}{}{}
\item[(left-only)]
$b\sep \skipc : \rspec{\Q\land \leftF{e}\land\encode{L} }{\Q}$
for each $e\gcto b$ in $gcs$ 
\item[(right-only)]
$\skipc\sep b' : \rspec{\Q\land \rightF{e'}\land\encode{R} }{\Q}$
for each $e'\gcto b'$ in $gcs'$ 
\item[(joint)]
$b\sep b' : \rspec{\Q\land \bothF{e\sep e'} \land \neg\encode{L} \land \neg\encode{R}}{\Q}$
for each $e\gcto b$ in $gcs$ and $e'\gcto b'$ in $gcs'$
\end{list}
By Lemma~\ref{lem:nfCases}, the guarded commands in $gcs$ and $gcs'$ have five possible forms,
so there are five left-only cases to consider, five right-only, and 25 joint ones.
We start with the latter.

\paragraph{Joint cases}

For each of the five possibilities for $e\gcto b$ in $gcs$ for $c$,
we must consider it with each of the five possibilities for $e'\gcto b'$ in $gcs'$ for $c'$.
First we consider the  cases where both sides are of the same kind.

\begin{itemize}
\item $\spc m\sep \spc m' : \;
\rspec{\Q\land \bothF{\tpc k\sep \tpc k'} \land \neg\encode{L} \land \neg\encode{R}}{\Q}$ , where \\
$\sub(k,c) = \lskipc{k}$, $m=\fsuc(k,c,f)$,
$\sub(k',c') = \lskipc{k'}$, $m=\fsuc(k',c',f')$.

By (\ref{eq:setR}) (and freshness of $pc$) we have
$\spc m \sep \spc m' : \rspec{an(m,m')}{an(m,m')\land \bothF{\tpc m\sep\tpc m'}}$.

So by \rn{rConseq} using the second implication in (\ref{eq:IanRel}) we have
\begin{equation}\label{eq:mm}
\spc m \sep \spc m' : \rspec{an(m,m')}{\Q}
\end{equation}
By the first implication in (\ref{eq:IanRel}) we have
$\Q\land\bothF{\tpc k\sep\tpc k'}\imp an(k,k')$.
So using (\ref{eq:dliftJ}) we get 
$\Q\land\bothF{\tpc k\sep\tpc k'}\land\neg\encode{L}\land\neg\encode{R}\imp \encode{J}$.
By validity of the annotation, we have the VC in the first row of Fig.~\ref{fig:RVCjo},
i.e., $J\land \breve{an}(k,k')\imp \hat{an}(m,m')$.
Then by Lemma~\ref{lem:liftRVC} we get the lifted VC
$\encode{J}\land\bothF{\tpc k\sep\tpc k'}\land an(k,k')\imp an(m,m')$.
So this is valid:
\[ \Q\land \bothF{\tpc k\sep\tpc k'} \land \neg\encode{L} \land \neg\encode{R}
\imp an(m,m') \]
Using this with \rn{rConseq} and (\ref{eq:mm}) yields
$\spc m\sep \spc m' : \;
\rspec{\Q\land \bothF{\tpc k\sep \tpc k'} \land \neg\encode{L} \land \neg\encode{R}}{\Q}$.

\item $\lassg{}{x}{e};\spc m \sep \lassg{}{x'}{e'};\spc m' : \;
\rspec{\Q\land \bothF{\tpc k\sep\tpc k'} \land \neg\encode{L} \land \neg\encode{R}}{\Q}$ , where \\
$\sub(k,c) = \lassg{k}{x}{e}$ , $m = \fsuc(k,c,f)$,
$\sub(k',c') = \lassg{k'}{x'}{e'}$, and $m' = \fsuc(k',c',f')$.

As in the previous case we have 
$\spc m \sep \spc m' : \rspec{an(m,m')}{\Q}$ (see (\ref{eq:mm})).
By rule \rn{dAsgn} we get
$\lassg{}{x}{e} \sep \lassg{}{x'}{e'}: \;\rspec{\subst{an(m,m')}{x|x'}{e|e'}}{an(m,m')}$.
By the VC (Fig.~\ref{fig:RVCjo}) we have 
$J\land \breve{an}(k,k') \imp \subst{\hat{an}(m,m')}{x|x'}{e|e'}$, 
so by Lemma~\ref{lem:liftRVC} we have
$\encode{J}\land\bothF{\tpc k\sep\tpc k'}\land an(k,k') \imp \subst{an(m,m')}{x|x'}{e|e'}$. 
So by \rn{rConseq} we get
\[ \lassg{}{x}{e} \sep \lassg{}{x'}{e'}: \;\rspec{\encode{J}\land\bothF{\tpc k\sep\tpc k'}\land an(k,k')}{an(m,m')} 
\]
As $k,k'$ are not the final labels,
by (\ref{eq:dliftJ}) we have
$an(k,k,)\land\bothF{\tpc k\sep\tpc k'}\land\neg \encode{L}\land\neg \encode{R} \imp \encode{J}$,
and we have $Q\land\bothF{\tpc k\sep\tpc k'}\imp an(k,k')$ by the first implication in (\ref{eq:IanRel}).  Hence we have 
$Q\land\bothF{\tpc k\sep\tpc k'}\land\neg \encode{L}\land \neg\encode{R} \imp
\encode{J}\land\bothF{\tpc k\sep\tpc k'}\land an(k,k')$.
So using \rn{rConseq} we get 
\[ \lassg{}{x}{e} \sep \lassg{}{x'}{e'}: \; \rspec{Q\land\bothF{\tpc k\sep\tpc k'}\land\neg \encode{L}\land \neg\encode{R}}{an(m,m')} 
\]
Combining this and $\spc m \sep \spc m' : \rspec{an(m,m')}{\Q}$ 
by rule \rn{dSeq} completes the proof.
 
\item $\spc m \sep \spc m' : \;
\rspec{\Q\land \leftF{\tpc k\land e}\land\rightF{\tpc k'\land e'} \land \neg\encode{L} \land \neg\encode{R}}{\Q}$ , where \\
$\sub(k,c) = \lifgc{k}{gcs_0}$, $m=\lab(d)$ with $e\gcto d$ in $gcs_0$,
$\sub(k',c') = \lifgc{k'}{gcs'_0}$, $m=\lab(d')$ with $e'\gcto d'$ in $gcs'_0$
--- similar to preceding.

\item $\spc m \sep \spc m' : \;
\rspec{\Q\land \leftF{\tpc k\land e}\land\rightF{\tpc k'\land e'} \land \neg\encode{L} \land \neg\encode{R}}{\Q}$ , where \\
$\sub(k,c) = \ldogc{k}{gcs_0}$, $m=\lab(d)$ with $e\gcto d$ in $gcs_0$,
$\sub(k',c') = \ldogc{k'}{gcs'_0}$, $m=\lab(d')$ with $e'\gcto d'$ in $gcs'_0$ 
--- similar to preceding.

\item $\spc m \sep \spc m' : \;
\rspec{\Q\land \leftF{\tpc k\land\neg\enab(gcs_0)}\land\rightF{\tpc k'\land\neg\enab(gcs'_0)} \land \neg\encode{L} \land \neg\encode{R}}{\Q}$ , where \\
$\sub(k,c) = \ldogc{k}{gcs_0}$, $m=\fsuc(k,c,f)$,
$\sub(k',c') = \ldogc{k'}{gcs'_0}$, $m'=\fsuc(k',c',f')$ ---
similar to preceding.
\end{itemize}

Next we consider the four cases of $\skipc$ on the right and one of the other cases on the left.

\begin{itemize}
% ass|skip
\item $\lassg{}{x}{e};\spc m \sep\lskipc{} : \;
\rspec{\Q\land \bothF{\tpc k\sep\tpc k'} \land \neg\encode{L} \land \neg\encode{R}}{\Q}$ , where \\
$\sub(k,c) = \lassg{k}{x}{e}$, $m=\fsuc(k,c,f)$,
$\sub(k',c')=\lskipc{k'}$, and $m' = \fsuc(k',c',f')$.

As in earlier cases, we have
$\spc m\sep\spc m': \rspec{an(m,m')}{\Q}$ (see(\ref{eq:mm}).
By rule \rn{AsgnSkip} we get
\[ \lassg{}{x}{e}\sep\lskipc{}: \;\rspec{\subst{an(m,m')}{x|}{e|}}{an(m,m')} \]
By the VC (Fig.~\ref{fig:RVCjo}) for this case we have 
$J\land \breve{an}(k,k') \imp \subst{\hat{an}(m,m')}{x|}{e|}$, 
so by Lemma~\ref{lem:liftRVC} we have
$\encode{J}\land\bothF{\tpc k\sep\tpc k'}\land an(k,k') \imp \subst{an(m,m')}{x|}{e|}$. 
So by \rn{rConseq} we get
\[ \lassg{}{x}{e} \sep \lskipc{}: \;\rspec{\encode{J}\land\bothF{\tpc k\sep\tpc k'}\land an(k,k')}{an(m,m')} 
\]
As $k,k'$ are not the final labels,
by (\ref{eq:dliftJ}) we have
$an(k,k,)\land\bothF{\tpc k\sep\tpc k'}\land\neg \encode{L}\land\neg \encode{R} \imp \encode{J}$,
and we have $Q\land\bothF{\tpc k\sep\tpc k'}\imp an(k,k')$ by the first implication in (\ref{eq:IanRel}) so we have 
$Q\land\bothF{\tpc k\sep\tpc k'}\land\neg \encode{L}\land \neg\encode{R} \imp
\encode{J}\land\bothF{\tpc k\sep\tpc k'}\land an(k,k')$.
So using \rn{rConseq} we get 
\[ \lassg{}{x}{e} \sep \skipc{}: \; \rspec{Q\land\bothF{\tpc k\sep\tpc k'}\land\neg \encode{L}\land \neg\encode{R}}{an(m,m')} 
\]
Combining this and $\spc m \sep \spc m' : \rspec{an(m,m')}{\Q}$ 
by \rn{dSeq} completes the proof.

% if|skip
\item $\spc m\sep \spc m' : \;
\rspec{\Q\land \bothF{\tpc k\land e\sep\tpc k'} \land \neg\encode{L} \land \neg\encode{R}}{\Q}$ , where \\
$\sub(k,c) = \lifgc{k}{gcs_0}$, $m=\lab(d)$ 
$\sub(k',c') = \lskipc{k'}$, and $m'=\fsuc(k',c',f')$,
with $e\gcto d$ in $gcs_0$ ---
similar to preceding.

% do enter | skip
\item $\spc m\sep \spc m' : \;
\rspec{\Q\land \bothF{\tpc k\land e\sep\tpc k'} \land \neg\encode{L} \land \neg\encode{R}}{\Q}$ , where \\
$\sub(k,c) = \ldogc{k}{gcs_0}$, $m=\lab(d)$ 
$\sub(k',c') = \lskipc{k'}$, and $m'=\fsuc(k',c',f')$,
with $e\gcto d$ in $gcs_0$ ---
similar to preceding.

% do exit | skip
\item $\spc m\sep \spc m' : \;
\rspec{\Q\land \bothF{\tpc k\land\neg\enab(gcs_0)\sep \tpc k'} \land \neg\encode{L} \land \neg\encode{R}}{\Q}$ , where \\
$\sub(k,c) = \ldogc{k}{gcs_0}$, $m=\fsuc(k,c,f)$, 
$\sub(k,c) = \lskipc{k}$, and $m=\fsuc(k,c,f)$ --- similar to preceding.
\end{itemize}

We refrain from spelling out details of the remaining cases:
4 cases with assign on right, 4 cases with if on the right, 4 cases with do-enter on the right,
and 4 cases with do-exit on the right.
The arguments are all similar: every case uses a VC and rule \rn{rConseq},
together with (\ref{eq:setR}) which is proved using \rn{dAsgn} and \rn{rConseq}.
Cases that involve an assignment in the original program $c$ or $c'$ also
use rules \rn{dAsgn}, \rn{SkipAsgn}, or \rn{AsgnSkip}.
No other rules are used to establish the premises of rule \rn{dDo}.

%% % FROM LEMMA

%% \item $\tpc k\gcto \spc m$, 
%% for some $k,m$ such that $\sub(k,c)$ is $\lskipc{k}$ and $m=\fsuc(k,c,f)$.

%% \item $\tpc k\gcto \lassg{}{x}{e};\spc m$,
%% for some $k,m,x,e$ such that $\sub(k,c)$ is $\lassg{k}{x}{e}$ 
%% and $m = \fsuc(k,c,f)$.

%% \item $\tpc k\land e\gcto \spc m$,
%% for some $k,m,e,d,gcs_0$ such that $\sub(k,c)$ is $\lifgc{k}{gcs_0}$ 
%% and $m=\lab(d)$ where $e\gcto d$ is in $gcs_0$.

%% \item $\tpc k\land e\gcto \spc m$, 
%% for some $k,m,e,d,gcs_0$ such that $\sub(k,c)$ is $\ldogc{k}{gcs_0}$
%% and $m=\lab(d)$ where $e\gcto d$ is in $gcs_0$.

%% \item $\tpc k\land \neg\enab(gcs_0)\gcto \spc m$, 
%% for some $k,m,gcs_0$ such that $\sub(k,c)$ is 
%% $\ldogc{k}{gcs_0}$ and $m=\fsuc(k,c,f)$.

\paragraph{Left-only cases}

These cases are proved using the same rules as the joint cases,
plus one additional rule: \rn{rDisj}, in the form \rn{rDisjN} derived from it (Sect.~\ref{sec:derived}).  
This is needed due to the following complication.
The joint cases determine a starting pair and ending pair of control points,
which determines which VC to appeal to.
The left-only cases do not determine a control point on the right side; instead
we have VCs for each possible point on the right (Fig.~\ref{fig:RVClo-encoded}).
So we go by cases on the possible control points on the right, for which purpose
we make the following observation.
In virtue of the conjunct $\Q_{pc}$ of $\Q$, we have that $\Q$ is equivalent to this 
disjunction over control points:
\[ %\begin{equation}\label{eq:Qdisj}
\quant{\lor}{i,j}{i\in \labs(c)\union\{f\} \land 
                    j\in \labs(c')\union\{f'\} }{ \Q^{i,j} }
\]
where $\Q^{i,j}$ is defined to say control is at those points:
\[ \Q^{i,j}: \quad \Q\land\bothF{\tpc i\sep\tpc j} \]
Moreover, we have the equivalence
\[ \Q\land\leftF{\tpc i}
 \iff 
\quant{\lor}{j}{j\in \labs(c')\union\{f'\} }{ \Q^{i,j} } \]
Now we consider the left-only premises, by cases.
\begin{itemize}
\item $\spc m\sep \skipc : \rspec{\Q\land \leftF{\tpc k}\land\encode{L} }{\Q}$, where 
$\sub(k,c) = \lskipc{k}$ and $m=\fsuc(k,c,f)$.

In accord with the discussion above, we have that
$\Q\land \leftF{\tpc k}$ is equivalent to 
$\quant{\lor}{j}{}{ \Q^{k,j} }$ (omitting the range $j\in \labs(c')\union\{f'\}$), 
so the goal can be obtained by \rn{rConseq} 
from 
\[ \spc m\sep \skipc : \rspec{\quant{\lor}{j}{}{ \Q^{k,j} }\land\encode{L} }{\Q}
\]
In turn, this can be obtained by derived rule \rn{rDisjN} (Sect.~\ref{sec:derived}) from judgments 
\begin{equation}\label{eq:mkj}
\spc m\sep \skipc : \rspec{\Q^{k,j} \land\encode{L} }{\Q}
\end{equation}
for all $j$ (in range $j\in \labs(c')\union\{f'\}$ that we continue to omit). 

It remains to prove (\ref{eq:mkj}) for arbitrary $j$.  First, by \rn{AsgnSkip} and \rn{rConseq} 
(and $an(m,j)$ independent from $pc$)
we get
\begin{equation}\label{eq:mkjj}
\spc m\sep \skipc : \rspec{an(m,j)\land\rightF{\tpc j}}{an(m,j)\land\bothF{\tpc m\sep\tpc j}} 
\end{equation}
The disjunction over $j$ was introduced so that we can appeal to a VC,
specifically the lifted VC for $((k,j),(m,j))$.  
It is an instance of the first line in Fig.~\ref{fig:RVClo-encoded} and it says this is valid:
\[ \encode{L}\land \bothF{\tpc k\sep\tpc j}\land an(k,j) \imp an(m,j) \]
By definitions we have $\Q^{k,j} \land \encode{L} \imp \encode{L}\land\bothF{\tpc k\sep\tpc j}\land an(k,j) $ 
so we have 
\( \Q^{k,j}\land \encode{L} \imp an(m,j) \).
We also have \( \Q^{k,j}\land \encode{L} \imp \rightF{\tpc j} \), so we have
\( \Q^{k,j}\land \encode{L} \imp an(m,j)\land \rightF{\tpc j} \).
Using this with (\ref{eq:mkjj}), \rn{rConseq} yields
\[ 
\spc m\sep \skipc : \rspec{\Q^{k,j}\land\encode{L}}{an(m,j)\land\bothF{\tpc m\sep\tpc j}} 
\]
So by the second implication in
(\ref{eq:IanRel}) and \rn{rConseq} we get (\ref{eq:mkj}).

\item $\lassg{}{x}{e};\spc m \sep \skipc : \rspec{\Q\land \leftF{\tpc k}\land\encode{L} }{\Q}$, where 
$\sub(k,c) = \lassg{k}{x}{e}$ and $m = \fsuc(k,c,f)$.

As in the preceding case, we use \rn{rDisjN} to reduce the goal to 
the goals 
\begin{equation}\label{eq:Y}
\lassg{}{x}{e};\spc m \sep \skipc : \rspec{\Q^{k,j} \land \encode{L} }{\Q} 
\quad\mbox{for every $j$.}
\end{equation}
So fix $j$.
By \rn{AsgnSkip} and \rn{rConseq} (using that $\rightF{\tpc j}$ is independent from $pc$ on the left) we get
\[ \spc m\sep\skipc : \rspec{ an(m,j)\land\rightF{\tpc j} }{ an(m,j) \land \bothF{\tpc m\sep\tpc j}} \]
from which using (\ref{eq:IanRel}) we get  
\[ \spc m\sep\skipc : \rspec{ an(m,j)\land\rightF{\tpc j} }{ \Q } \]
By \rn{AsgnSkip}, using that $\rightF{\tpc j}$ is independent from $x$ because $pc$ is fresh, we can prove
\[ x:=e\sep\skipc : \rspec{ \subst{an(m,j)}{x|}{e|} \land \rightF{\tpc j}}{an(m,j)\land\rightF{\tpc j}} 
\]
By derived rule \rn{SeqSkip}, from the above we get 
\[ x:=e;\spc m\sep\skipc : \rspec{ \subst{an(m,j)}{x|}{e|} \land \rightF{\tpc j}}{\Q} 
\]
Now by the lifted VC for assignment on the left (Fig.~\ref{fig:RVClo-encoded} second row) we have
\[ \encode{L}\land\bothF{\tpc k\sep\tpc j}\land an(k,j) \imp \subst{an(m,j)}{x\sep}{e\sep}
\]
so by \rn{rConseq} we get 
\[ x:=e;\spc m\sep\skipc : \rspec{ \encode{L}\land\bothF{\tpc k\sep\tpc j}\land an(k,j) }{\Q} 
\]
As noted in the previous case, 
$\Q^{k,j} \land \encode{L} \imp \encode{L}\land\bothF{\tpc k\sep\tpc j}\land an(k,j) $.
Using this, another use of \rn{rConseq} yields (\ref{eq:Y}).

\item $\spc m \sep \skipc : \rspec{\Q\land \leftF{\tpc k\land e}\land\encode{L} }{\Q}$, where \\  
$\sub(k,c)$ is $\lifgc{k}{gcs_0}$ and $m=\lab(d)$ with $e\gcto d$ in $gcs_0$.

This follows by \rn{rDisjN} from the judgments (for all $j$)
\begin{equation}\label{eq:BB}
\spc m \sep \skipc : \rspec{ \Q^{k,j}\land\leftF{e}\land\encode{L} }{ \Q }
\end{equation}
We get 
$\spc m \sep \skipc : \rspec{an(m,j)}{an(m,j)\land\bothF{\tpc m\sep\tpc j}}$ 
using \rn{AsgnSkip} and \rn{rConseq}.
Then another use of \rn{rConseq} with the second implication in (\ref{eq:IanRel}) yields
$ \spc m \sep \skipc : \rspec{an(m,j)}{\Q} $.
We have $\encode{L}\land\bothF{\tpc k\sep\tpc j}\land an(k,j)\land\leftF{e} \imp an(m,j)$ by lifted VC.
We have $Q^{k,j}\land\leftF{e}\land \encode{L}\imp \encode{L}\land\bothF{\tpc k\sep\tpc j}\land an(k,j) \land\leftF{e}$ by def $Q^{k,j}$ and (\ref{eq:IanRel}).
So we have $Q^{k,j}\land\leftF{e}\land \encode{L}\imp an(m,j)$
and by \rn{rConseq} we get (\ref{eq:BB}).
 
\item $\spc m \sep \skipc : \rspec{\Q\land \leftF{\tpc k\land e}\land\encode{L} }{\Q}$, where \\
$\sub(k,c)$ is $\ldogc{k}{gcs_0}$ and $m=\lab(d)$ with $e\gcto d$ in $gcs_0$.

\item $\spc m\sep \skipc : \rspec{\Q\land \leftF{\tpc k\land\neg\enab(gcs_0)}\land\encode{L} }{\Q}$, where \\
$\sub(k,c) = \ldogc{k}{gcs_0}$ and $m=\fsuc(k,c,f)$.

\end{itemize}

\paragraph{Right-only cases}

These are symmetric with the left-only cases and omitted.

\paragraph{Finishing the proof}

Having proved the premises and the side condition (\ref{eq:side}), rule \rn{dDo} yields 
(\ref{eq:C}).   The remaining steps are similar to corresponding steps in the proof of
the Floyd completeness Theorem~\ref{thm:FloydComplete} and we spell them out.

Using Lemma~\ref{lem:nf-enab-labs} twice, and the definition of $\Q$,
we have 
\[ \Q\land\neg\leftF{\enab(gcs)}\land\neg\rightF{\enab(gcs')} 
\imp \Q\land\bothF{\tpc f\sep\tpc f'} \]
So, using the first implication in (\ref{eq:IanRel}) and assumption (c) (which says $an(f,f')$ is $\T$), we can use
\rn{rConseq} with (\ref{eq:C}) to get
\[
 \ldogc{}{gcs} \sep \ldogc{}{gcs'} : \rspec{\Q}{\T}
\]
Using the second implication in (\ref{eq:IanRel}) and assumption (c) (which says $an(n,n')$ is $\S$
since $n,n'$ are the initial control points),
we have $\bothF{\tpc n\sep \tpc n'} \land \S \imp \Q$,
so by \rn{rConseq} we get 
\[
 \ldogc{}{gcs} \sep \ldogc{}{gcs'} : \rspec{\bothF{\tpc n\sep \tpc n'} \land \S}{\T}
\]
By freshness assumption for $pc$, by (\ref{eq:setR})
we have a proof of $\spc n\sep \spc n': \rspec{\S}{\bothF{\tpc n\sep\tpc n'} \land \S}$.
So by \rn{rSeq} we get
\[
 \spc n;\ldogc{}{gcs} \sep \spc n';\ldogc{}{gcs'} : \rspec{\S}{\T}
\]
Now using rule \rn{rRewrite} with the 
equivalences (\ref{eq:nfEquLR}) we get 
\[
\addPC(c); \spc f \sep \addPC(c'); \spc f' : \rspec{\S}{\T}
\]
By freshness of $pc$ and Lemma~\ref{lem:addPC}, 
it has the ghost property for both 
$\addPC(c); \spc f$ and $\addPC(c'); \spc f'$,
and does not occur in $\S$ or $\T$, so by rule \rn{rGhost} we get 
\begin{equation}\label{eq:H}
\erase(pc,\addPC(c); \spc f) \sep \erase(\addPC(c'); \spc f') : \rspec{\S}{\T}
\end{equation}
By definition of $\erase$, we have
$\erase(pc,\addPC(c); \spc f) = \erase(pc,\addPC(c)); \skipc$ 
and also 
\\
$\erase(\addPC(c'); \spc f') = \erase(pc,\addPC(c')); \skipc$. 
So using Lemma~\ref{lem:addPC} together with the general law $c;\skipc\kateq c$
and transitivity of $\kateq$, we have 
\[ \erase(pc, \addPC(c); \spc f) \kateq c \qquad\mbox{and}\qquad 
   \erase(pc, \addPC(c'); \spc f') \kateq c' \]
Using these equivalences with rule \rn{rRewrite}, from (\ref{eq:H}) we obtain the result:
$c  \sep c'  : \rspec{\S}{\T}$. 
\end{proof}

\section{Cook completeness without sequential product}\label{sec:cook}

Applied to HL, the usual notion of logical completeness would say
any true correctness judgment $c:\spec{P}{Q}$ can be derived using the rules. 
The rules involve an ancillary judgment, entailment between assertions,
for which one could adopt a deductive system.  But for a reasonably expressive assertion language there can be no complete deductive system.
Cook~\cite{Cook78} showed that HL is complete, in a \emph{relative} sense:
any true correctness judgment can be derived, given an oracle for entailment.
Cook completeness also requires \emph{expressiveness} of the assertion language, meaning
that weakest preconditions can be expressed, for whatever types of data are manipulated
by the program~\cite{AptOld3}.  
Using shallow embedding of assertions, as we do, is one way to set those issues aside

In this section we revisit Cook completeness in light of our results.
The discussion centers on the following rule, 
which assumes a function $\fdot$ that renames all variables 
in a program to fresh names (with some specific decoration like $'$) and $^\eplus$ that 
converts store relations to predicates on stores with the renamed+original variables.
\begin{mathpar}
\inferrule*[left=SeqProd\quad]
{ d' = \fdot(d) \\ c;d' : \spec{\R^\eplus}{\S^\eplus} }
{ c\sep d : \rspec{\R}{\S} }
\end{mathpar}
Francez~\cite{Francez83} observes that sequential product is complete relative to HL,
but does not work it out in detail.
Beringer~\cite{Beringer11} proves semantic completeness of sequential product
and leverages it to derive rules including two conditionally aligned loop rules.
Barthe et al.~\cite{BartheDArgenioRezk} emphasize the completeness of the sequential product rule 
for the special case of relating a program to itself; it is clear that it holds generally as noted in~\cite{BartheCK-FM11}.
A number of other works have shown Cook completeness of a relational Hoare 
logic based on (their versions of) the sequential product rule 
(e.g.,~\cite{SousaD2016,BartheGHS17,WangDilligLahiriCook}).

Our RHL+ does not include \rn{SeqProd}, nor does it make any other use of unary correctness judgments.  It does enjoy Cook completeness, 
and this is a direct consequence of alignment completeness.
Before spelling that out, we tell the corresponding story for unary correctness.

%% By Prop.~\ref{prop:iamcomplete}, together with the soundness of IAM (Prop.~\ref{prop:IAM}),
%% any valid correctness judgment can be proved using IAM.
%% So?

\begin{corollary}[Cook completeness of HL+]\label{cor:CookComplete}
\upshape
Suppose $\models c: \spec{P}{Q}$
and suppose $\okf(c,f)$.
Then there is a proof of $c: \spec{P}{Q}$ in HL+.
\end{corollary}
\begin{proof}
Suppose $\models c: \spec{P}{Q}$.
By Lemma~\ref{lem:autConsistent} we have $\aut(c,f) \models \spec{P}{Q}$.
By Prop.~\ref{prop:iamcomplete} there is a valid annotation $an$
of $\aut(c,f)$ for $\spec{P}{Q}$.
So by the Floyd completeness Theorem~\ref{thm:FloydComplete} there is a proof in HL+
of $c: \spec{P}{Q}$.
\end{proof}
The assumption $\okf(c,f)$ can be ignored: 
wlog we can always relabel the program.

\begin{theorem}[Cook completeness of RHL+]
\upshape
Suppose $\models c\sep c': \rspec{\S}{\T}$.
Then there is a proof of $c\sep c': \rspec{\S}{\T}$ in RHL+.
\end{theorem}
\begin{proof}   
Suppose $\models c\sep c': \rspec{\S}{\T}$, and assume wlog
that $\okf(c,f)$ and $\okf(c',f')$. 
Using Lemma~\ref{lem:autConsistent} 
we have $\aut(c,f), \aut(c',f') \models \rspec{\S}{\T}$.
By Cor.~\ref{cor:relIAMcomplete} there are $L,R,J,an$ such that 
$an$ is a valid annotation of $\aprod(\aut(c,f),\aut(c',f'),L,R,J)$ for $\spec{\S}{\T}$ and 
$\aprod(\aut(c,f),\aut(c',f'),L,R,J)$ is manifestly $\S$-adequate.
By inspection of the proof of Cor.~\ref{cor:relIAMcomplete} we can choose 
$L$, $R$, $J$ to all be simply $\neg[\fin|\fin']$. 
So we have for any $i,j$ that 
$\hat{an}(i,j)\imp L \lor R \lor J \lor [\fin|\fin']$ 
which is condition (c) of Theorem~\ref{thm:acomplete}.
Conditions (a) and (b) hold as well, so by the theorem we 
get a proof of 
$c\sep c': \rspec{\S}{\T}$ in RHL+.
\end{proof}

Although RHL+ does not include unary judgments or the sequential product rule,
it effectively embeds HL by way of the one-sided rules like \rn{AsgnSkip}
and (derived rule) \rn{SeqSkip}.
Moreover we get the effect of \rn{SeqProd} with the rule
\[
\inferrule{
c\sep \skipc : \rspec{\P}{\Q} \\
\skipc\sep d : \rspec{\Q}{\R} 
}{
c\sep d : \rspec{\P}{\R}
}
\]
which can be derived using \rn{dSeq} and \rn{rRewrite} with
$ c \kateq \skipc;c$ and $d \kateq d;\skipc$.

\section{Conclusion}\label{sec:discuss}

We augmented a collection of relational Hoare logic rules with a straightforward
rule for elimination of ghost variables, and a rule for deriving one correctness
judgment from another by rewriting the commands involved to equivalent ones.
The chosen notion of equivalence is that of Kleene algebra with tests (KAT),
allowing for the use of hypotheses to axiomatize the meaning of primitive
commands and expressions when reasoning with KAT.  Using a small set of
hypotheses, we prove that any command is equivalent to one in automaton normal
form, once it is instrumented with assignments to a ``program counter''
variable.  On this basis, we show that any correctness judgment proved by the
inductive assertion method using an alignment automaton can be turned into a
proof in RHL+ using essentially the same assertions.  This shows that RHL+ is
alignment complete.  As a consequence, we obtain Cook completeness for RHL+.  To
the best of our knowledge, this is the first Cook completeness result for a
relational Hoare logic that does not rely on unary judgments and a sequential
product (self-composition) rule.

The key idea is that an program is equivalent to one structured similarly to an 
automaton, so an automaton-based annotation can be adapted correctness judgments
in a deductive proof.  While KAT provides a well understood notion of equivalence,
other sound equivalences could also serve as basis for alignment complete logics.

Alignment serves to enable use of relatively simple assertions,
but our results are not restricted to any particular assertion language.
In Remark~\ref{rem:moduli} we point out that the example of Sect.~\ref{sec:condEg}
can be adapted to one that requires non-linear arithemetic in alignment conditions,
but nonetheless admits a proof with simple assertions that do not express full functional correctness.

In practice, many verification tools work directly with verification conditions, although some interactive tools are directly based on a Hoare logic (e.g.,~\cite{cao2018vst}) and modular tools implicitly implement linking rules for modular verification of procedures and their callers.  Deductive systems for verification are nonetheless interesting for various reasons.  One practical use for deductive rules is for independently checkable certificates to represent proofs that may be found by other means.  The alignment completeness result shows that, as in the case of unary Hoare logic, a small collection of mostly syntax-directed rules suffices to represent proofs of relational correctness judgments.  

Imperative commands are at the core of most practical programming languages but such languages also
include heap operations and other features that can have runtime faults, for which one should augment partial correctness with the avoidance of fault.
For fault-avoiding relational correctness, the notion of adequacy needs to be refined to ensure that even divergent traces are covered;\footnote{See for example
the adequacy Theorem~7.11 and it's supporting lemma~C.9 in~\cite{BNNN19}.}
this is spelled out by Nagasamudram and Naumann~\cite{NagasamudramN21}.
To extend our alignment complete logic to such a language, the command equivalence needs to be revisited.  Although KAT is sound for trace models, the standard translation of loops 
to KAT discards their divergent executions, see Eqn.~(\ref{eq:def:mkt}). 
One approach would be to replace the translation $\mkt{-}$, in the definition of $\kateq$, by a translation that maps each command to a KAT term that describes all prefixes of its traces.  (This is easy to define.) If two commands agree on all traces up to any fault, then they have the same faults, so the rewrite rules would still be sound.  We conjecture that Theorem~\ref{thm:normEquiv} 
holds for this definition of $\kateq$.  Note that the approach still uses KAT reasoning and does not require to model failure in the command equivalence.  (That could be done, however,
using FailKAT~\cite{Mamouras17}.)

%% Remark: a natural variation on RHL+ would be to replace the unconditional KAT-based equivalence by some notion of equivalence under a precondition.  Then the precondition in rule \rn{Rewrite} would also be used as precondition for the equivalence (and similarly for \rn{rRewrite}.  
%% Our results show that conditional equivalence is not needed for alignment completeness.
%% Given that conditional equivalence can be expressed by relational judgments, one might 
%% go further and eliminate the KAT-based equivalence in favor of relational judgment of equivalence.
%% Alternatively, one can dispense with relational judgments in favor of an equational 
%% means of expressing them; this is done in BiKAT~\cite{AntonopoulosEtal2022}.

\bibliographystyle{ACM-Reference-Format}
\bibliography{biblio}

\vfill
\pagebreak
\appendix
%\onecolumn 
\section{Appendix: Technical details}\label{app:details}

Definition of \graybox{$\enab(gcs)$} (an expression for the condition under which guarded command list $gcs$ is enabled):
\[
\begin{array}{lcl}
\enab(e\gcto c) & = & e \\
\enab(e\gcto c \gcsep gcs) & = & e \lor \enab(gcs) 
\end{array}\]
Definition of \graybox{$\lab(c)$}:
\[
\begin{array}{lcl}
\lab(\lskipc{n})            & = & n \\
\lab(\lassg{n}{x}{e})       & = & n \\
\lab(c;d)                 & = & \lab(c) \\
\lab(\lifgc{n}{gcs})   & = & n \\
\lab(\ldogc{n}{gcs})    & = & n 
\end{array}
\]
Definition of \graybox{$\labs(c)$}:
\[
\begin{array}{lcl}
\labs(\lskipc{n})           & = & \{ n \}\\
\labs(\lassg{n}{x}{e})      & = &\{ n\} \\
\labs(c;d)                & = & \labs(c)\union \labs(d) \\
\labs(\lifgc{n}{gcs})  & = & \{n\} \union \{ m \mid e\gcto c \mbox{ is in }gcs \land
                                             m\in labs(c) \} \\
\labs(\ldogc{n}{gcs})  & = & \{n\} \union \{ m \mid e\gcto c \mbox{ is in }gcs \land
                                             m\in labs(c) \} 
\end{array}
\]
Definition of \graybox{$\ok(c)$} which says $c$ has unique labels, all of which are positive.
\[
\begin{array}{lcl}
\ok(\lskipc{n})            & = & n >  0 \\
\ok(\lassg{n}{x}{e})       & = & n >  0 \\
\ok(c;d)                 & = & \labs(c)\intersect \labs(d)=\emptyset\land \ok(c)\land \ok(d) \\
\ok(\lifgc{n}{gcs})    & = & n > 0 \land n \notin \labs(gcs) \land \okg(gcs) \\
\ok(\ldogc{n}{gcs})    & = & n > 0 \land n \notin \labs(gcs) \land \okg(gcs) \\
\okg(e\gcto c) & = & \ok(c) \\
\okg(e\gcto c \gcsep gcs) & = & \ok(c)\land \labs(c)\intersect \labs(gcs)=\emptyset \land
             \okg(gcs)
\end{array}
\]
Definition of \graybox{$\sub(m,c)$}, the sub-command of $c$ at $m$,
assuming $\ok(c)$, and $m\in \labs(c)$.
\[
\begin{array}{lcl}
\sub(m,\lskipc{m})        & = & \lskipc{m} \\
\sub(m,\lassg{m}{x}{e})   & = & \lassg{m}{x}{e} \\
\sub(m,c;d)               & = & \sub(m,c)\mbox{ , if $m\in \labs(c)$}  \\
                          & = & \sub(m,d)\mbox{ , otherwise} \\
\sub(m, \lifgc{n}{gcs}) & = &  \lifgc{n}{gcs} \mbox{ , if $n=m$}  \\
                           & = &  \subg(m,gcs) \mbox{ , otherwise}  \\
\sub(m, \ldogc{n}{gcs}) & = &  \ldogc{n}{gcs} \mbox{ , if $n=m$}  \\
                           & = &  \subg(m,gcs) \mbox{ , otherwise}  \\

\subg(m, e\gcto c) & = & \sub(m,c) \\
\subg(m, e\gcto c \gcsep gcs) & = &\sub(m,c) \mbox{ , if $m\in \labs(c)$ }\\
                 & = & \subg(m,gcs) \mbox{ , otherwise } 
\end{array}
\]
%This covers all cases, given that $ m\in \labs(c)$ and  $\ok(c)$.
Definition of \graybox{$\erase(x,c)$}, the command $c$ with each assignment to $x$ replaced by $\skipc$.
\[
\begin{array}{l@{\;}c@{\;}l}
\erase(x,\lskipc{n})          & = & \lskipc{n} \\
\erase(x,\lassg{n}{x}{e})     & = & \lskipc{n} \\
\erase(x,\lassg{n}{y}{e})     & = & \lassg{n}{y}{e} \quad\mbox{for } y\nequiv x\\
\erase(x,\lifgc{n}{gcs})      & = & \lifgc{n}{\map(\erasegc,gcs)} \\
\erase(x,\ldogc{n}{gcs})      & = & \ldogc{n}{\map(\erasegc,gcs)} \\[1ex]
\erasegc(x,e\gcto c)            & = & e \gcto erase(x,c)
\end{array}
\]

\medskip

\newcommand{\bodies}{\mathconst{bodies}}
\newcommand{\concat}{\mathconst{concat}}

\begin{figure*}
\begin{small}
\begin{mathpar}

\inferrule{}{
\norm{\lskipc{n}}{f}{
(\tpc n\gcto \spc f)}
}

\inferrule{}{
\norm{\lassg{n}{x}{e}}{f}{ 
(\tpc  n \gcto \lassg{}{x}{e}; \spc  f)}
}

\inferrule{
\norm{c}{\lab(d)}{gcs_0} \\
\norm{d}{f}{gcs_1}
}{
\norm{c;d}{f}{gcs_0\gcsep gcs_1}
}

\inferrule{
\norm{\bodies(gcs)}{f}{\mathit{nfs}}
}{ 
\norm{ \lifgc{n}{gcs} }{f}{ 
  map((\lambda(e\gcto c)\,.\, \tpc n\land e\gcto \spc\lab(c)),\, gcs)
  \gcsep \concat(\mathit{nfs})
}}

\inferrule{
\norm{\bodies(gcs)}{f}{\mathit{nfs}}
}{ 
\norm{ \ldogc{n}{gcs} }{f}{ 
  map((\lambda(e\gcto c)\,.\, \tpc n\land e\gcto \spc\lab(c)),\, gcs)
  \gcsep \tpc n\land\neg\enab(gcs) \gcto \spc f
  \gcsep \concat(\mathit{nfs})
}}

% normlst in Coq
\inferrule{ \norm{c}{f}{gcs} }{ \norm{[c]}{f}{[gcs]} }

\inferrule{ \norm{c}{f}{gcs} \\ \norm{cs}{f}{\mathit{nfs}} }
{\norm{c::cs}{gcs::\mathit{nfs} }}

\end{mathpar}
\end{small}
%\vspace*{-5ex}
\caption{Normal form bodies.}\label{fig:norm:general}
\end{figure*}

Fig.~\ref{fig:norm:general} defines the normal form relation.
It includes the general cases for if- and do-commands,
which subsume the special cases given in Fig.~\ref{fig:norm}.
The definition of normal form bodies for commands is mutually inductive 
with the definition of normal form body list for nonempty command lists,
which is given by the last two rules in Fig.~\ref{fig:norm:general}.
The rules use square brackets for singleton list and $::$ for list cons.
They use $\bodies(gcs)$ and $\concat(nfs)$ defined by
\[\begin{array}{lll}
\bodies(e\gcto c) &=& [c] \\
\bodies(e\gcto c\gcsep gcs) &=& c::\bodies(gcs)\\
\concat([gcs]) &=& gcs\\
\concat(gcs:nfs) &=& gcs \gcsep \concat(nfs)
\end{array}\]

% omit needless sections 
%% \section{Appendix: unary and relational proof rules}

%% \begin{figure}[t]
%% \begin{small}
%% \begin{mathpar}

%% \mprset{sep=1.3em}

%% \inferrule[Skip]{}{ 
%% \skipc:\spec{P}{P} 
%% }

%% \inferrule[Seq]
%% {
%% c:\spec{P}{R} \\ d:\spec{R}{Q}
%% }{
%% c;d : \spec{P}{Q} 
%% }

%% \inferrule[If]
%% {
%% c: \spec{e\land P}{Q} \mbox{ for every $e\gcto c$ in $gcs$}
%% }{
%% %    \ifgc{gcs}: \spec{P\land\enab(gcs)}{Q}
%%     \ifgc{gcs}: \spec{P}{Q}
%% }

%% \inferrule[Conseq]
%% {
%% P\imp R \\ 
%% c : \spec{R}{S} \\
%% S\imp Q
%% }{
%% c : \spec{P}{Q} 
%% }

%% \inferrule[False]{}{ 
%% c:\spec{\mathit{false}}{P} 
%% }

%% \end{mathpar}
%% \end{small}
%% \vspace*{-1ex}
%% \caption{Rules of HL omitted from Fig.~\ref{fig:HLplus} (command labels elided).}\label{fig:HL}
%% \end{figure}

%% Figure \ref{fig:HL} gives the rules of HL not given in the body of the paper.

%% % Obsolete; the condition was removed since we restrict to totalIf programs
%% Remark about the \rn{If} rule:
%% The precondition in the conclusion, $\enab(gcs)$,
%% is not necessary in our semantics, where the command gets stuck (i.e., has no final state) when no guard is enabled.  
%% But the standard semantics faults in such a situation, and anyway the rule as stated
%% has a pleasant symmetry with respect to the rule \rn{Do}.
%% In this paper, we are mainly concerned with $\totalIf$ programs which means $\enab(gcs) = true$.

\section{Appendix: Additional examples}

\newcommand{\ccol}{{\, : \;}}
\newcommand{\lassgX}[3]{#2:=#3}

\begin{figure}
\begin{small}
  \begin{tikzpicture}
    [->,every state/.append style={fill=gray!10},
initial text=$ $,
auto,node distance=2.5cm,line width=0.1mm,inner sep=1pt,
scale=0.70,transform shape]
  \node[state] (lck1) {$1,1$};
  \node[state, right of=lck1] (lck2) {$2,2$};
  \node[state, right of=lck2, xshift=-0.1cm] (lck3) {$3,3$};
  \node[state, right of=lck3, xshift=5em] (lck4) {$4,4$};
  \node[state, right of=lck4, xshift=0.5cm] (lck5) {$5,5$};
  \node[state, below of=lck5, yshift=1cm] (lck6) {$6,6$};
%  \node[state, above of=lck3, xshift=-1.5cm] (lo4) {$4,3,lo$};
%  \node[state, right of=lo4] (lo5) {$5,3,lo$};
  \node[state, below of=lck3, xshift=-1.5cm] (ro4) {$3,4$};
  \node[state, right of=ro4] (ro5) {$3,5$};
  \draw
  (lck1) edge node[above]{$\lassgX{}{y}{x}$}
              node[below]{$\lassgX{}{y'}{x'}$} (lck2)
  
  (lck2) edge node[above,yshift=0.1cm]{$\lassgX{}{z}{1}$}
              node[below]{$\lassgX{}{z'}{1}$} (lck3)

  (lck3) edge node[above]{$y\neq 0$}
              node[below]{$y'\neq 0 \land \graybox{$y\nless y'$}$} (lck4)
  
  (lck4) edge node[above]{$\lassgX{}{z}{z*y}$} node[below]{$\lassgX{}{z'}{z'*y'}$} (lck5)

%  (lck3) edge[bend left=10] node[left,yshift=0.2cm,xshift=0.2cm]{$y\neq0 \land \Lrel$} (lo4)

%  (lo4) edge[bend left=10] node[above]{$\lassgX{}{z}{z*y}$} (lo5)

  (lck3) edge[bend right=10] node[left,yshift=-0.2cm,xshift=0.2cm]{$y'\neq0 \land \graybox{$y<y'$}$} (ro4)

  (ro4) edge[bend right=10] node[below]{$\lassgX{}{z'}{z'*y'}$} (ro5)

%  (lo5) edge[bend left=10] node[right,xshift=-0.5cm,yshift=0.25cm]{$\lassgX{}{y}{y-1}$} (lck3)

  (ro5) edge[bend right=10] node[right,xshift=-0.5cm,yshift=-0.25cm]{$\lassgX{}{y'}{y'-1}$} (lck3)

  (lck3) edge[bend right=20] node[above,xshift=1.2cm,yshift=-0.18cm]{$y=0$}
                            node[below,xshift=1.2cm,yshift=-0.18cm]{$y'=0$} (lck6)
  (lck5) edge[bend right=40]  node[above]{$\lassgX{}{y}{y-1}$}
                              node[below]{$\lassgX{}{y'}{y'-1}$} (lck3)
  ;
  \end{tikzpicture}
\end{small}
\vspace*{-2ex}
\caption{Conditionally aligned product for \\
$c1:\quad  y:= x; z:= 1; 
\ldogc{}{y \neq 0 \gcto  z:= z*y; y:= y-1 }$}\label{fig:aut3}
\end{figure}

\begin{example}
This shows that invariance of $L\lor R\lor J\lor [\fin|\fin']$ is not necessary for adequacy (apropos Lemma~\ref{lem:anLRJF}).
Command $c1$ is taken from~\cite{NagasamudramN21} where it is named $c0$.
\[ c1: \quad y:= x; z:= 1; 
\ldogc{}{y \neq 0 \gcto  z:= z*y; y:= y-1 } \]
Fig.~\ref{fig:aut3} shows an alignment automaton for proving $c1 \sep c1 \ccol \rspec{0\leq x\leq x'}{z\leq z'}$.

Omitting the nodes $(3,4)$ and $(3,5)$ and their edges, 
and the test $\nless$, gives an automaton suitable for proving 
$c1 \sep c1 \ccol \rspec{0\leq x = x'}{z = z'}$.

We can also keep the nodes, and the
edges $(3,3)\to(3,4)$ and $(3,4)\to(3,5)$ but omit the edge
$(3,5)\to(3,3)$,
and omit both tests $y<y$ and $y\nless y'$.
This gives a product that is $(0\leq x=x')$-adequate 
but also has gratuitous paths that get stuck.
By annotating $(3,4)$ and $(3,5)$ true, we get a valid annotation.
This product does not have $R$ disjoint from $J$
(they overlap at $(3,3)$),
and at $(3,5)$ we do not have $L\lor R\lor J\lor [\fin|\fin']$.
\qed\end{example}

\end{document}

%% file: macros.tex
\newcommand{\dt}[1]{\textbf{\emph{#1}}} % defined term

\newcommand{\sep}{\mathbin{\mid}} % larger version of | 
\newcommand{\Sep}{\mathbin{\:\mid\:}} % still larger version of | 

\newcommand{\means}[1]{\llbracket\, #1 \,\rrbracket}

\newcommand{\eqdef}{\mathrel{\,\hat{=}\,}}

\newcommand{\mkt}[1]{\llparenthesis\, #1 \,\rrparenthesis}

\renewcommand{\neg}{\mbox{\relsize{-1}$\lnot$}}

\newcommand{\union}{\cup}
\newcommand{\intersect}{\cap}
\newcommand{\proves}{\vdash}
\newcommand{\aftqua}{.\:}
\newcommand{\all}[2]{\forall #1 \aftqua #2}
\newcommand{\some}[2]{\exists #1 \aftqua #2}

\newcommand{\nat}{\mathbb{N}} % nat but be careful
\newcommand{\Z}{\mathbb{Z}} % integers 

\newcommand{\imp}{\Rightarrow}
\newcommand{\impby}{\Leftarrow}
\renewcommand{\iff}{\Leftrightarrow}

\newcommand{\ok}{\mathconst{ok}}
\newcommand{\okf}{\mathconst{okf}}
\newcommand{\okft}{\mathconst{okft}}
\newcommand{\okg}{\mathconst{okg}}
\newcommand{\enab}{\mathconst{enab}}
\newcommand{\totalIf}{\mathconst{total\mbox{-}if}}

\newcommand{\lab}{\mathconst{lab}}
\newcommand{\labs}{\mathconst{labs}}
\newcommand{\sub}{\mathconst{sub}}
\newcommand{\subg}{\mathconst{subg}}

\newcommand{\fsuc}{\mathconst{fsuc}}
\newcommand{\ghost}{\mathconst{ghost}}
\newcommand{\erase}{\mathconst{erase}}
\newcommand{\erasegc}{\mathconst{erasegc}}
\newcommand{\addPC}{\mathconst{add}^{pc}}
\newcommand{\map}{\mathconst{map}}
\newcommand{\indep}{\mathconst{indep}} 
\newcommand{\nfax}{\mathconst{nfax}} 
\newcommand{\aprod}{\prod}
\newcommand{\dom}{\mathconst{dom}}  % domain of a function

\newcommand{\nbl}{\shortrightarrow} % see text
\newcommand{\encode}{\tilde} % lift an automaton state predicate to one on pc variable

\newcommand{\eplus}{{\mbox{\relsize{-4}$\oplus$}}}

\newcommand{\POST}{\mathconst{post}} % strongest post 

\newcommand{\Var}{\mathconst{Var}} % program variables 

\newcommand{\aut}{\mathconst{aut}}

\newcommand{\init}{\mathit{init}}
\newcommand{\fin}{\mathit{fin}}
\newcommand{\trans}{\mapsto} % automaton transition
\newcommand{\tranSeg}[1]{\stackrel{#1}{\longmapsto}} 
\newcommand{\biTrans}{\Mapsto} 
\newcommand{\biTranSeg}[1]{\stackrel{#1}{\Mapsto}} 

\newcommand{\ctrans}{\rightarrowtriangle} % command transition
\newcommand{\config}[2]{\langle #1,\: #2\rangle}

\newcommand{\kateq}{\mathrel{\simeq}} % KAT+AX equality 

\newcommand{\relKAT}{\mathfrak{L}} % the relational model for GCL 
\newcommand{\anyKAT}{\mathfrak{M}} % arbitrary KAT model

\newcommand{\mathconst}[1]{\mbox{\upshape\textsf{#1}}} % font for fixed math entities

\newcommand{\rn}[1]{\textsc{\relsize{-1}#1}} % rule names

\newcommand{\update}[3]{[#1\, |\, #2\scol\, #3]} % function override - Reynolds' notation
\newcommand{\scol}{\mathord{:}} % spaceless colon

\newcommand{\leftex}[1]{ \raisebox{.25ex}{\relsize{-1}$\langle\hspace*{-2.1pt}[$} #1 \raisebox{.25ex}{\relsize{-1}$\langle\hspace*{-2.5pt}]$} }
\newcommand{\rightex}[1]{ \raisebox{.25ex}{\relsize{-1}$[\hspace*{-2.5pt}\rangle$} #1 \raisebox{.25ex}{\relsize{-1}$]\hspace*{-2.1pt}\rangle$} }
\newcommand{\leftF}[1]{\leftex{#1}}
\newcommand{\rightF}[1]{\rightex{#1}}
\newcommand{\bothF}[1]{%
  \raisebox{.25ex}{\relsize{-1}$\langle\hspace*{-2.1pt}[$}
  #1 
  \raisebox{.25ex}{\relsize{-1}$]\hspace*{-2.1pt}\rangle$ } 
}

\newcommand{\hint}[1]{{\qquad\mbox{#1}}}

\newcommand{\keyw}[1]{\ensuremath{\mathsf{#1}}} 

\newcommand{\skipc}{\keyw{skip}}
\newcommand{\lskipc}[1]{\keyw{skip}^{#1}}
\newcommand{\lassg}[3]{#2:=^{#1}#3}

\newcommand{\ifgc}[1]{\keyw{if}\ {#1}\ \keyw{fi}}
\newcommand{\dogc}[1]{\keyw{do}\ {#1}\ \keyw{od}}
\newcommand{\lifgc}[2]{\keyw{if}^{#1}\ {#2}\ \keyw{fi}}
\newcommand{\ldogc}[2]{\keyw{do}^{#1}\ {#2}\ \keyw{od}}
\newcommand{\gcsep}{\talloblong} % guarded command separator 
\newcommand{\gcto}{\mathrel{\shortrightarrow}} % guarded command arrow

\renewcommand{\P}{\mathcal{P}}  % ALERT renew
\renewcommand{\S}{\mathcal{S}} % ALERT renew
\newcommand{\Lrel}{\mathcal{L}} 
\newcommand{\Q}{\mathcal{Q}} 
 
\newcommand{\R}{\mathcal{R}} 
\newcommand{\T}{\mathcal{T}} 
\newcommand{\fdot}{\mathconst{dot}} 
\newcommand{\subst}[3]{{#1}^{#2}_{#3}}

\definecolor{light-gray}{gray}{0.88}
\definecolor{dark-gray}{gray}{0.25}
\newcommand{\graybox}[1]{\colorbox{light-gray}{#1}} % grey background in math

\newcommand{\specSym}{\leadsto}
\newcommand{\rspecSym}{\ensuremath{\mathrel{\mbox{\footnotesize$\thickapprox\hspace*{-.4ex}>$}}}}
\newcommand{\spec}[2]{#1\specSym #2} 
\newcommand{\rspec}[2]{#1\rspecSym #2} 

\newcommand{\norm}[3]{#1 \mathbin{/} #2 \hookrightarrow  #3}
\newcommand{\normSmall}[3]{#1 \mathbin{/} #2 \hookrightarrow #3}

\newcommand{\spc}{\mathord{!}}
\newcommand{\tpc}{\mathord{?}}

\newcommand{\quant}[4]{(#1\,#2 \::\: #3 \::\: #4)}

\newenvironment{ditemize}{%   itemized lists w/o interline space
\begin{list}{$\bullet$}{%
\setlength{\itemsep}{0pt}\setlength{\rightmargin}{0pt}%
\setlength{\leftmargin}{.6em}\setlength{\parsep}{0ex}}}{
\end{list}}